\documentclass[12pt]{article}
\usepackage{amsmath}
\usepackage{graphicx}
\usepackage{enumerate, setspace}
\usepackage{subcaption}
\usepackage[round]{natbib}
\usepackage{url} 
\usepackage{float}
\usepackage{mathtools}
\usepackage{chngcntr}
\usepackage{fontenc}  
\usepackage{lmodern}

\usepackage{ragged2e} 
\usepackage{amsthm}
\allowdisplaybreaks

\usepackage{algorithm,algpseudocode,float}
\usepackage{lipsum}
\usepackage{amssymb}
\usepackage{xcolor}
\usepackage{color}

\usepackage{chngcntr}
\usepackage{booktabs}
\usepackage{array}
\usepackage{makecell}
\usepackage{adjustbox}

\usepackage[normalem]{ulem}  

\newcommand{\blind}{1}



\usepackage{hyperref}					
\hypersetup{colorlinks,
	linkcolor=blue,%
	citecolor=blue}

\newtheorem{theorem}{Theorem}

\newtheorem{lemma}{Lemma}
\newtheorem{proposition}{Proposition}
\newtheorem{assumption}{Assumption}

\newtheorem{claim}{Claim}
\theoremstyle{definition} 
\newtheorem{remark}{Remark}

\newcommand{\II}{\mathbb I}

\begin{document}

	\def\spacingset#1{\renewcommand{\baselinestretch}%
	{#1}\small\normalsize} \spacingset{1}


\if1\blind
{
	\title{\bf Structure-Adaptive Conformal Inference for Large-Scale Out-of-Distribution Testing} 
	\author{Rongyi Sun\thanks{
			Center for Data Science and School of Mathematical Sciences, Zhejiang University.}\hspace{.2cm}, Wenguang Sun\thanks{
			Center for Data Science and School of Management, Zhejiang University.}
		\hspace{0.01cm} and Zinan Zhao\thanks{
			Center for Data Science and School of Mathematical Sciences, Zhejiang University.}}
	\date{}
	\maketitle
} \fi

\if0\blind
{
	\bigskip
	\bigskip
	\bigskip
	\begin{center}
		{\LARGE\bf Structure-Adaptive Conformal Inference for Large-Scale Out-of-Distribution Testing}
	\end{center}
	\medskip
} \fi

\bigskip

\begin{abstract}
	This paper addresses structured out-of-distribution (OOD) testing in high-stakes machine learning applications. Traditional conformal methods rely on joint exchangeability, making it difficult to incorporate auxiliary information such as spatiotemporal or grouping structures. To overcome this limitation, we propose the structure-adaptive conformal q-value (SCQ), a significance index that integrates individual test evidence with structural patterns. We also develop pseudo-score-guided transductive automated model selection (P-TAMS), which adapts conformalized model selection to structured OOD testing across a toolbox of candidate models. Together, SCQ and P-TAMS form a unified framework under pairwise exchangeability, providing finite-sample error-rate control, improved power, and enhanced interpretability. Experiments on simulated and real data demonstrate that the proposed approach controls the false discovery rate and performs well across diverse settings.
\end{abstract}

\noindent%


\spacingset{1} 


\maketitle

\section{Introduction}\label{sec:intro}

Out-of-distribution (OOD) testing, also known as outlier or novelty detection, is a crucial task in a range of important domains such as medical diagnostics, image detection, security monitoring, and autonomous driving 
(\citealp{1995OOD, PIMENTEL2014215, lee2018simple, 2024OOD}). 
The goal is to identify instances that deviate from a reference distribution of labeled inliers. 
State-of-the-art OOD testing methods often rely on complex machine learning (ML) algorithms that do not provide rigorous uncertainty quantification for their outputs. 
These limitations present significant challenges in high-stakes settings, where precise risk assessment and strict error control are essential. 
Conformal inference 
(\citealp{vovk1999machine, vovk05, balasubramanian2014conformal, lei2014distribution}) 
has emerged as a powerful framework that grounds ML algorithms in rigorous theoretical foundations; in particular, methods based on conformal p-values 
(\citealp{bates2023testing, marandon2024adaptive, liang2024integrative}) 
provide finite-sample error-rate guarantees, enabling reliable OOD testing in risk-sensitive environments.

When thousands of test units are examined simultaneously, controlling the false discovery rate 
(FDR; \citealp{benjamini1995controlling}) provides a practical strategy to balance the benefits of detecting true outliers against the cost of reporting false discoveries. 
In many data-intensive applications, test units are accompanied by auxiliary information such as spatial locations, timestamps, group memberships, data-driven clusters, or domain-specific annotations. 
Such information often reveals meaningful structure in the distribution of outliers. 
For instance, in cybersecurity monitoring, attacks may concentrate during particular working hours or traffic regimes; in document analysis, non-text elements may cluster according to page layout or document source; and in spatial or temporal anomaly detection, abnormal events often appear in localized clusters rather than in isolation. 
Leveraging these structures can improve both power and interpretability by prioritizing signal-enriched regions, borrowing information across related units, and aggregating weak but coherent signals.

However, existing structured FDR methods 
(\citealp{li2019multiple, ignatiadis2021covariate, cai2022laws}) 
are not directly applicable to OOD testing, as they typically rely on p-values derived from distributional models that may be unavailable, unverifiable, or unreliable in modern ML applications with complex data structures. 
Our goal is to harness the model-free validity of conformal inference to develop principled and effective procedures for structured OOD testing. 
This entails addressing two key challenges.

The first challenge is a conflict between structural adaptivity and conformal validity. 
Incorporating auxiliary structure places test units on unequal footing, whereas standard conformal methods rely on joint exchangeability among calibration and test samples. 
Enforcing joint exchangeability prevents the procedure from fully exploiting structural heterogeneity, while naively using such heterogeneity can invalidate FDR control. 
The second challenge concerns model selection: although conformal inference is model-agnostic in principle, its practical implementation typically relies on a fixed pre-trained score function. 
This is restrictive in structured OOD testing, where the performance of different classifiers can vary substantially across data regimes, sparsity levels, and structural patterns. 
To address these challenges, we propose data-adaptive, provably valid procedures that integrate structural information into both conformal calibration and model selection.

\subsection{A preview of our proposal and contributions}

To bypass the strict requirement of joint exchangeability, Section~\ref{sec:SCQ} introduces the structure-adaptive conformal q-value (SCQ), which operates under the weaker assumption of pairwise exchangeability. The SCQ combines unit-specific evidence with relevant structural information, resulting in a structure-aware risk measure with guaranteed statistical validity.

To address the challenge of model selection, Section~\ref{sec:P-TAMS} extends conformalized automated model selection (CAMS) methods (\citealp{liang2023conformal, magnani2023collective, marandon2024adaptive,liang2024integrative}) to structured settings. The proposed algorithm, termed pseudo-score-guided transductive automated model selection (P-TAMS), employs innovative min-max pairing and coin-flipping strategies to maintain pairwise exchangeability throughout the model selection process. P-TAMS cohesively integrates  with the SCQ procedure, forming a unified framework for the efficient detection of structured anomalies.

The contributions of our work include: 

\begin{itemize}
	\setlength{\itemsep}{0pt}
	\setlength{\parsep}{0pt}
	\setlength{\topsep}{0pt}
	
    \item \emph{A structure-aware significance index:} 
    SCQ offers a principled, user-friendly, and interpretable significance index in structured settings, enabling a precise assessment of the relative importance of test units in light of side information. 
    Unlike standard conformal p-values, SCQ incorporates structural information from the test sample through weighted conformity-score pairs while preserving the pairwise exchangeability needed for valid FDR control.

    \item \emph{A validity-preserving principle for structural adaptation:} 
    In structured OOD testing, auxiliary information such as spatial location, timestamp, or group membership can improve power, but using it to reweight or rerank test units may break the joint exchangeability required by standard conformal methods. 
    Our formulation shows that joint score exchangeability is not necessary for FDR control; it suffices to preserve pairwise exchangeability for null score pairs. 
    This provides a practical principle for learning data-dependent structural weights from the test sample through swap-invariant constructions. In particular, the SCQ theory accommodates transductive, data-dependent weights that are not covered by classical weighted-BH or PRDS-based arguments.

    \item \emph{A novel CAMS strategy:} 
    P-TAMS extends existing CAMS methods to structured settings, thereby enhancing SCQ by offering provable validity and improved efficiency for OOD testing tasks. 
    Its key design is to use swap-invariant pseudo-scores for model selection, allowing the selected classifier to depend on the test data without violating the pairwise exchangeability.
        
    \item \emph{A unified framework under weaker assumptions:} 
    Both SCQ and P-TAMS rely on pairwise score exchangeability rather than joint exchangeability of the final scores. 
    This weaker requirement makes it possible to incorporate transductive, data-dependent weighting and automated model selection into a coherent framework tailored for structured OOD testing. 

    \item \emph{An in-depth theoretical analysis of the SCQ procedure:} 
    We develop an asymptotic framework for analyzing the theoretical properties of SCQ. 
    Our theory addresses two critical gaps in the literature. 
    First, we characterize the conditions under which SCQ attains the nominal FDR level. 
    Second, we establish when and why integrating structural knowledge yields increased power. 
\end{itemize}

\subsection{Related work}

Our work integrates several latest advances from related areas, such as conformal inference for OOD testing, structured multiple testing, model-free FDR control, and conformalized automatic model selection. We discuss these related works to clarify our contribution. 

\textbf{Conformal Inference for OOD Testing.} Recent studies (e.g., \citealp{mary2022semi, bates2023testing, marandon2024adaptive, liang2024integrative}) have successfully employed conformal p-values to perform OOD testing with finite-sample guarantees. However, the theoretical validity of existing methods critically depends on the joint exchangeability assumption, which fundamentally limits their applicability to structured testing scenarios. Our research directly addresses this limitation by relaxing this stringent assumption. 

\textbf{Structured Multiple Testing.} 
The design of SCQ is inspired by the structured multiple testing literature, where incorporating side information can improve both power and interpretability \citep{genovese2006false, benjamini2007false, sun2009large, sun2015false}. 
Our work is closely related to weighted p-value methods, including IHW \citep{ignatiadis2021covariate}, SABHA \citep{li2019multiple}, and LAWS \citep{cai2022laws}. 
The finite-sample FDR theory (e.g., \citealp{genovese2006false}) requires the weights to be fixed in advance or chosen independently of the null p-values; under such exogeneity conditions, normalizing the weights so that \(\sum_{j=1}^m w_j=m\) preserves FDR control. 
By contrast, 
SCQ learns weights from the test sample itself, using the paired conformal evidence \(\{(p_j,\widetilde p_j):j\in \mathcal{D}^{\mathrm{test}}\}\) together with side information \(\mathbf{S}\). 
These weights are therefore transductive and data-dependent, and are not covered by the classical theory (e.g. the sufficient conditions in \citealp{BlaRoq08}) even after normalization. 

\textbf{Automated Model Selection.} The effectiveness of OOD testing methods critically depends on the choice of underlying classifier, which can be a one-class classifier (OCC) (\citealp{moya1993one, khan2014one, sabokrou2018adversarially}), a binary classifier (BIC) (\citealp{kumari2017machine, haroush2021statistical}), or a classifier constructed via positive-unlabeled learning techniques (PUC; \citealp{du2014analysis}, \citealp{bekker2020learning}). 
The problem of automatically selecting among these different types of classifiers was recently explored by \cite{liang2023conformal}, \cite{magnani2023collective}, \cite{liang2024integrative}, \cite{marandon2024adaptive}, and \cite{bai2024optimized}. However, existing conformalized model selection methods cannot handle structured setups; we provide a detailed comparison with these methods in Appendix \ref{sec:P-TAMS_vs_CAMS}. Moreover, the BONuS \citep{yang2021bonus}, MetaOD \citep{Zhaetal21-MetaOD} and AutoMS \citep{zhang2022automs} methods explored model selection ideas but not within the conformal context, while \cite{yang2021finite}, \cite{stutz2021learning}, \cite{einbinder2022training}, and \cite{liang2024conformal} studied classifier selection for conformal prediction, but the ideas are inapplicable to OOD testing.  

\textbf{Model-Free FDR Control.} The construction of SCQ is conceptually inspired by FDR methodologies based on mirror processes, including knockoff filters \citep{barber2015controlling}, AdaPT \citep{lei2018adapt}, ZAP \citep{leung2022zap}, SDA \citep{du2023false}, adaptive knockoffs \citep{ren2023knockoffs}, and PLIS \citep{ZhaSun25-PLIS}, among others. SCQ shares with these methods the broad mirror-calibration principle, but targets a different inferential problem and relies on a different score-construction mechanism. Unlike regression knockoffs, which construct synthetic covariates and compare original and knockoff feature statistics, SCQ addresses structured OOD testing with unlabeled test units and mirror null samples. Its key challenge is to construct conformity-score pairs that exploit structural information from the test set while preserving the pairwise exchangeability needed for FDR control. 
Hence, the novelty of SCQ lies not in a BC-type thresholding rule, but in the construction, calibration, and model selection of structure-adaptive conformal scores for OOD testing.

\textbf{Post-Selection Inference.} The literature on post-selection inference with false coverage-statement control \citep{benjamini2005false, bao2024selective, gazin2025selecting, JinRen25, gui2025acs} focuses on drawing inferences after the \emph{selection of test units}. In contrast, our P-TAMS algorithm operates along a different dimension -- namely, \emph{selecting among different classifiers}.


\section{FDR Control for Structured OOD Testing}\label{sec:SCQ}

This section develops a generic methodology for structured OOD testing that integrates OCCs, BICs, and PUCs into a unified framework; this lays the groundwork for Section~\ref{sec:P-TAMS}, where we present the P-TAMS algorithm for model selection across these classifier families. In Section~\ref{subsec:setup}, we begin by formulating the problem and introducing a unified approach to constructing conformal p-values that accommodates various classifiers. Section~\ref{subsec:SCQ} provides detailed descriptions of the SCQ procedure and establishes its theoretical properties.

\subsection{Problem formulation}\label{subsec:setup}

Suppose we observe \(n\) labeled data points \(\{(X_i, Y_i): i\in[n]\}\), where \(X_i \in \mathbb{R}^p\) is a p-dimensional feature and \(Y_i \in \{0,1\}\) denotes the label, with \(Y_i = 0\) and \(Y_i = 1\) corresponding to an inlier and an outlier, respectively. Define the index sets \(\mathcal{D}_0 = \{i \in [n]: Y_i = 0\}\) and \(\mathcal{D}_1 = \{i \in [n]: Y_i = 1\}\). In addition, we have \(m\) unlabeled features \((X_{n+1}, \dots, X_{n+m})\) with corresponding unknown labels \((Y_{n+1}, \dots, Y_{n+m})\). Let \(\mathcal{D}^{\text{test}} = \{n+1, \dots, n+m\}\). 
Denote $\mathcal{H}_0$ and $\mathcal{H}_1$ the index sets of inliers and outliers in \(\mathcal{D}^{\text{test}}\), respectively. Structural information associated with each test sample \(X_{j}\) -- such as group membership or spatial location -- is captured by an external covariate \(S_{j}\). Denote  \(\mathbf{S}=\{S_j : j\in\mathcal{D}^{\text{test}}\}\).

The decision rule for determining which of the \( m \) instances are outliers can be represented by  
$
(\hat{Y}_{n+1}, \ldots, \hat{Y}_{n+m}) \in \{0, 1\}^m,
$
where \(\hat{Y}_j = 1\) indicates that we classify the \(j\)-th instance as an outlier, and \(\hat{Y}_j = 0\) otherwise.  Define the false discovery rate  

{\small$$
\text{FDR} \coloneqq \mathbb{E}\left[ \frac{\sum_{j \in \mathcal{D}^{\rm test}} (1-Y_j) \hat{Y}_j}{\left(\sum_{j \in \mathcal{D}^{\rm test}} \hat{Y}_j \right) \vee 1} \right],
$$}
where the expectation is taken over the labeled and test data.

We briefly outline the basic workflow of structured OOD testing. First, a collection of null samples is partitioned into three distinct subsets, each serving a specific purpose: (i) a pre-training set (\( \mathcal{D}_0^{\text{tr}} \)) for model training; (ii) a calibration set (\( \mathcal{D}_0^{\text{cal}} \)) for quantifying the uncertainty of model outputs; and (iii) a mirror set (\( \mathcal{D}_0^{\text{mir}} \)) for recalibrating with structural knowledge (see Section \ref{subsec:SCQ}). Our method requires that \(|\mathcal{D}_0^{\text{mir}}| = |\mathcal{D}^{\text{test}}| = m.\) Denote the test data as $\mathbf X\coloneqq(X_j: j\in\mathcal{D}^{\rm test})$ and rewrite the mirror data \((X_i: i \in \mathcal{D}_0^{\text{mir}})\) as \(\tilde{\mathbf{X}}\coloneqq(\tilde X_j: j \in \mathcal{D}^{\text{test}})\). 
This modified notation allows us to pair the observations from \(\mathcal{D}_0^{\text{mir}}\) and \(\mathcal{D}^{\text{test}}\) as $
\{(X_j, \tilde{X}_j) : j\in\mathcal{D}^{\text{test}}\},$
thereby simplifying the description of the method. Define \(\mathbf{X}_1=\{X_i : i\in \mathcal{D}_1\}\), \(\mathbf{X}_0^{\text{tr}}=\{X_i : i\in \mathcal{D}_0^{\text{tr}}\}\), and \(\mathbf{X}_0^{\text{cal}}=\{X_i : i\in \mathcal{D}_0^{\text{cal}}\}\).

Next, we train a score function $s(\cdot)$, applicable across various classifier families, that satisfies the following permutation-invariance condition:

{\small\begin{equation}\label{eq:s_function_permu}
	s\left(\cdot; (\mathbf{X},\tilde{\mathbf{X}}, \mathbf{X}_0^{\text{cal}})_{\Pi}, \mathbf{X}_0^{\text{tr}}, \mathbf{X}_1, \mathbf{S}\right) = s\left(\cdot; (\mathbf{X},\tilde{\mathbf{X}}, \mathbf{X}_0^{\text{cal}}), \mathbf{X}_0^{\text{tr}}, \mathbf{X}_1, \mathbf{S}\right),
\end{equation}}
where \(\Pi\) denotes an arbitrary permutation of the data points in $\mathbf{X}$, $\tilde{\mathbf{X}}$, and $\mathbf{X}_0^{\text{cal}}$.
Define the observed training data as \(\mathbf{O}^{\text{tr}}\), which may consist solely of null samples or may also include labeled outliers. 
The permutation-invariance principle in \eqref{eq:s_function_permu} allows the score function to be constructed using the mirror and test samples, together with the auxiliary training information \(\mathbf{O}^{\mathrm{tr}}\), while preserving the required symmetry. This formulation covers several important special cases: 
(a) OCC, where \(s(\cdot) \coloneqq s(\cdot;\mathbf{O}^{\mathrm{tr}})\) with \(\mathbf{O}^{\mathrm{tr}}=\{X_i: i\in\mathcal{D}_0^{\mathrm{tr}}\}\); 
(b) BIC, where \(s(\cdot) \coloneqq s(\cdot;\mathbf{O}^{\mathrm{tr}})\) with \(\mathbf{O}^{\mathrm{tr}}=\{(X_i,Y_i): i\in\mathcal{D}_0^{\mathrm{tr}}\cup\mathcal{D}_1\}\); 
and (c) PUC, where $s(\cdot) \coloneqq s\bigl(\cdot;\mathbf{O}^{\mathrm{tr}},(\mathbf{X},\tilde{\mathbf{X}},\mathbf{X}_0^{\mathrm{cal}})_{\Pi}\bigr)
=
s\bigl(\cdot;\mathbf{O}^{\mathrm{tr}},(\mathbf{X},\tilde{\mathbf{X}},\mathbf{X}_0^{\mathrm{cal}})\bigr),$
with \(\mathbf{O}^{\mathrm{tr}}=\{X_i: i\in\mathcal{D}_0^{\mathrm{tr}}\}\).

In OCC and BIC cases (e.g., \citealp{bates2023testing, liang2024integrative}), the score function \( s(\cdot) \) does not depend on the test or mirror data; hence, property \eqref{eq:s_function_permu} is automatically satisfied. In contrast, the PUC approach (e.g., \citealp{marandon2024adaptive, ZhaSun25-CLAW}) is carefully designed to fulfill property \eqref{eq:s_function_permu}. As elaborated in Appendix~\ref{sec:P-TAMS_vs_CAMS}, each classifier family has distinct advantages and limitations; a key advantage of the generic formulation outlined in \eqref{eq:s_function_permu} is its ability to provide a unified framework for SCQ construction (Section~\ref{subsec:SCQ}), model selection (Section~\ref{sec:P-TAMS}), and theoretical analyses (Section~\ref{sec:theories}) across a broad class of classifier families.

\subsection{The SCQ procedure}\label{subsec:SCQ}

Consider the partition introduced in Section \ref{subsec:setup}, where \(\mathcal{D}_0 = \mathcal{D}_0^{\text{tr}} \cup \mathcal{D}_0^{\text{cal}} \cup \mathcal{D}_0^{\text{mir}}\) with \(|\mathcal{D}_0^{\text{mir}}| = |\mathcal{D}^{\text{test}}| = m\).  The SCQ procedure operates in three steps.

\noindent\textbf{Step 1: Compute conformal p-values (first calibration).} Let \( s(\cdot) \in \mathcal{G} \) denote a score function that takes the generic form given in \eqref{eq:s_function_permu}, with smaller values indicating stronger evidence against the null hypothesis. The conformal p-value function is defined as:

\begin{equation}\label{eq:conformal_p}
	p(x) = \frac{1 + \left| \left\{ i \in \mathcal{D}_0^{\text{cal}}: s(X_i) \leq s(x) \right\} \right|}{1 + |\mathcal{D}_0^{\text{cal}}|}. 
\end{equation}
 Through \eqref{eq:s_function_permu} and \eqref{eq:conformal_p}, we first build 
a predictive model to compute preliminary conformity scores, and then construct $m$ pairs of conformal p-values: $\left\{\left(p(X_j), p(\tilde{X}_j)\right):j \in \mathcal{D}^{\text{test}}\right\}$.

The conformal p-values are inadequate for structured OOD testing problems, as the heterogeneity among test units has been ignored. The next step details a weighting strategy to incorporate structural knowledge.

\noindent\textbf{Step 2: Compute weighted p-values.} Let \(w(S_j)\) represent the weight associated with covariate \(S_j\), where a larger value of \(w(S_j)\) provides stronger indication that the \(j\)-th instance is an outlier. The key idea is to prioritize instances that are more likely to be outliers through the use of weighted p-values:

\begin{equation}\label{eq:weighted_p}
\left\{\left(V_j \coloneqq {p(X_j)}/{w(S_j)}, \tilde{V}_j \coloneqq {p(\tilde{X}_j)}/{w(S_j)}\right):j\in\mathcal D^{\text{test}}\right\},
\end{equation} 
where \((V_j, \tilde V_j)\) will be utilized as new conformity scores in the next step.

The weight function \(w(S_j)\) can be determined based on prior knowledge or structural information. For example, in the spatial multiple testing setting where \(S_j\) denotes the spatial location, let \(\pi(S_j)\) represent the local sparsity level, with higher values of \(\pi(S_j)\) indicating a greater density of outliers near \(S_j\). In this context, higher weights \(w(S_j)\) should be assigned to locations with higher \(\pi(S_j)\) \citep{li2019multiple, cai2022laws}. We require that if data-driven weights are used, the weight function \(w(S_j;(\mathbf{X},\tilde{\mathbf{X}}), \mathbf{X}_0^{\text{tr}}, \mathbf{X}_0^{\text{cal}}, \mathbf{X}_1, \mathbf{S})\) must belong to a generic class of swap-invariance functions $\mathcal G$, where for each $g\in\mathcal G$, we have

\begin{small}
		\begin{equation}\label{eq:g_func_swap_inva}
		g\left(\cdot; (\mathbf{X},\tilde{\mathbf{X}})_{\text{swap}(\mathcal{J})}, \mathbf{X}_0^{\text{tr}}, \mathbf{X}_0^{\text{cal}}, \mathbf{X}_1, \mathbf{S}\right) = g\left(\cdot; (\mathbf{X},\tilde{\mathbf{X}}), \mathbf{X}_0^{\text{tr}}, \mathbf{X}_0^{\text{cal}}, \mathbf{X}_1, \mathbf{S}\right), \quad \forall \mathcal{J} \subset \mathcal D^{\text{test}}.
	\end{equation}
\end{small}

\begin{remark}	
Here, \((\mathbf{X},\tilde{\mathbf{X}})_{\mathrm{swap}(\mathcal{J})}\) denotes the result of swapping \(X_j\) and \(\tilde{X}_j\) for each \(j \in \mathcal{J}\). 
Appendix~\ref{sec:weight_construction} provides a practical construction of such weights. 
Beyond estimating local sparsity levels, the key role of this construction is to learn structural heterogeneity from the test sample without breaking the pairwise exchangeability required for FDR control. 
Specifically, the weight function is built through a swap-invariant procedure, so that the weighted score pairs \((V_j,\widetilde V_j)\) remain valid inputs for the subsequent mirror calibration step.
\end{remark}

\noindent\textbf{Step 3: Calculate SCQs (second calibration).}
Unlike standard conformal p-values, the weighted p-values constructed from \eqref{eq:weighted_p} are not directly interpretable and do not preserve theoretical properties such as super-uniformity under the null. Hence, we apply a second calibration step employing the following mirror process \(H(t)\):
\begin{equation}\label{eq:H_function}
H(t) = \frac{1 + \sum_{j\in\mathcal{D}^{\rm test}} \mathbb{I}(\tilde{V}_j \leq t, \tilde{V}_j < V_j)}{\left[\sum_{j\in\mathcal{D}^{\rm test}} \mathbb{I}(V_j \leq t, V_j < \tilde{V}_j)\right] \vee 1}, \quad t > 0.
\end{equation}
Let \(\mathcal{V}= \{V_j: j \in \mathcal D^{\text{test}}\}\) and \(\widetilde{\mathcal{V}}= \{\tilde{V}_j: j \in \mathcal D^{\text{test}}\}\). Define the conformal q-values:

\begingroup
\renewcommand{\arraystretch}{0.85}
\begin{equation}
	q_j \coloneqq 
	\begin{cases}
		\min\nolimits_{t \in \mathcal{V} \cup \widetilde{\mathcal{V}}:\, t \geq V_j} H(t),
		& \text{if } V_j < \tilde{V}_j,\\
		1, 
		& \text{if } V_j \geq \tilde{V}_j,
	\end{cases}
	\quad j \in \mathcal D^{\text{test}}.
	\label{eq:scq}
\end{equation}
\endgroup
The rejection set of the SCQ procedure is then given by
\begin{equation}\label{dec:scq}
	\mathcal{R}_{scq} = \{j \in \mathcal{D}^{\rm test} : q_j \leq \alpha\}.
\end{equation}
Algorithm \ref{alg:SCQ} in Appendix \ref{sec:SCQ_P-TAMS algorithm} summarizes the key steps of the proposed method.  

\begin{remark}
As explained in Section \ref{subsec:BC=SCQ}, the mirror process \(H(t)\) is carefully constructed to estimate the false discovery proportion; hence, the q-value provides an intuitive and interpretable structure-aware risk measure. Moreover, the SCQ procedure is simple to use, enabling practitioners to make decisions by directly comparing q-values to a pre-specified level \(\alpha\). Finally, the validity of SCQ requires only the pairwise exchangeability of scores -- a considerably weaker condition than joint exchangeability; this flexibility allows SCQ to effectively incorporate local heterogeneities while ensuring its validity for FDR control.
\end{remark}

We now establish the theoretical properties of the SCQ procedure. Consider \(m\) pairs of conformity scores \(\{(V_j, \tilde{V}_j): j \in \mathcal{D}^{\text{test}}\}\). Let \(\mathcal{V}_{-i} = \{V_j: j \in \mathcal{D}^{\rm test}\setminus\{n+i\}\}\) and \(\widetilde{\mathcal{V}}_{-i} = \{\tilde{V}_j: j \in \mathcal{D}^{\rm test}\setminus\{n+i\}\}\). 
The following proposition shows that the conformity scores constructed via Steps 1-2 are pairwise exchangeable: 

\begin{equation}\label{eq:pwex-scores}
(V_i, \tilde{V}_i, \mathcal{V}_{-i}, \widetilde{\mathcal{V}}_{-i}) \overset{d}{=} (\tilde{V}_i, V_i, \mathcal{V}_{-i}, \widetilde{\mathcal{V}}_{-i}), \quad \forall i \in \mathcal{H}_0.
\end{equation}

We clarify two layers of exchangeability in our construction. The first is data-level exchangeability: under \eqref{eq:data_exch_joint} below, the calibration null samples, mirror null samples, and null test samples are exchangeable conditional on the non-null test samples, labeled outliers, and side information. This condition ensures that the conformal p-values in \eqref{eq:conformal_p} provide valid, model-agnostic evidence on a common scale. The second is score-level exchangeability for the final mirror procedure. After incorporating structural information through
$
\big(V_j = p(X_j)/w(S_j), \tilde{ V}_j = p(\widetilde X_j)/w(S_j)\big),
$
the weighted scores are generally no longer jointly exchangeable, since test units are intentionally treated differently according to their side information. 
For finite-sample FDR control, it suffices to retain pairwise exchangeability for each null pair. Thus, SCQ deliberately relaxes joint score exchangeability to exploit structural heterogeneity while preserving the pairwise exchangeability needed for validity.

The following proposition formalizes this transition from data-level exchangeability to score-level pairwise exchangeability.

\begin{proposition}\label{prop:pairwise_scores}
Suppose the data points satisfy the exchangeability condition

\begin{small}
	\begin{equation}\label{eq:data_exch_joint}
		\left(X_i:i\in\mathcal{D}_0;X_j:j\in\mathcal{H}_0\right)
		\text{ are exchangeable conditional on }
		\left(X_j:j\in\mathcal{H}_1,\mathbf{X}_1,\mathbf{S}\right).
	\end{equation}
\end{small}
Consider the weighted p-values 
\(\{(V_j, \widetilde V_j): j \in \mathcal{D}^{\mathrm{test}}\}\) 
constructed by \eqref{eq:conformal_p} and \eqref{eq:weighted_p}. 
Assume \(s(\cdot)\) and \(w(\cdot)\) satisfy the invariance conditions in 
\eqref{eq:s_function_permu} and \eqref{eq:g_func_swap_inva}, respectively. 
Then \(\{(V_j, \widetilde V_j): j \in \mathcal{D}^{\mathrm{test}}\}\) 
satisfy the pairwise exchangeability condition \eqref{eq:pwex-scores}.
\end{proposition}

The next theorem establishes the finite-sample validity of the SCQ procedure.

\begin{theorem}\label{thm:scq_validity}
Under the conditions in Proposition \ref{prop:pairwise_scores}, the SCQ procedure  controls the FDR at level \(\alpha\).
\end{theorem}

\begin{remark}\label{rmk:weaker-condi}
Similar assumptions to \eqref{eq:data_exch_joint} are widely adopted in the structured multiple testing literature \citep{ignatiadis2021covariate,li2019multiple,cai2022laws}, where null p-values are typically assumed to be independent and super-uniform conditional on the auxiliary covariates and the non-null p-values.  For clarity and interpretability, Theorem~\ref{thm:scq_validity} is presented under a slightly stronger set of conditions. However, the result remains valid under strictly weaker assumptions: (a) the score function \(s(\cdot)\) need only satisfy swap‑invariance \eqref{eq:g_func_swap_inva} in place of the permutation‑invariance required in \eqref{eq:s_function_permu}, and  
(b) the joint exchangeability condition \eqref{eq:data_exch_joint} can be relaxed to the following pairwise exchangeability: 

\begin{equation}\label{eq:data_exch_pair}
	\left((\mathbf{X}, \tilde{\mathbf{X}})_{\text{swap}(\mathcal{J})} \mid \mathbf{X}_0^{\text{tr}}, \mathbf{X}_0^{\text{cal}}, \mathbf{X}_1, \mathbf{S}\right) \overset{d}{=} \left(\mathbf{X}, \tilde{\mathbf{X}} \mid \mathbf{X}_0^{\text{tr}}, \mathbf{X}_0^{\text{cal}}, \mathbf{X}_1, \mathbf{S}\right), \quad \forall \mathcal{J} \subset \mathcal{H}_0. 
\end{equation}
These relaxations do not compromise the FDR guarantee established in Theorem~\ref{thm:scq_validity}. 
\end{remark}

\section{Conformalized Model Selection with P-TAMS}\label{sec:P-TAMS}

\subsection{General considerations}\label{subsec:P-TAMS-general}

The selection of an effective classifier -- whether OCC, BIC, or PUC -- is critical to enhancing the power of OOD testing. The performance of classifier families can vary considerably across settings, as each family possesses distinct strengths and weaknesses that are influenced by factors such as data dimensionality and class imbalance (cf. \citealp{liang2024integrative} and Appendix~\ref{sec:P-TAMS_vs_CAMS}). Moreover, within each family, specific implementations (such as random forests, support vector machines, or neural networks) can yield substantially different results across various datasets. This variability motivates the development of a principled, data-driven automated model selection (AMS) approach to adaptively identify the most suitable model from a toolbox of candidates.

However, if applied without proper adjustment, a naive model selection strategy that simply cherry‑picks the best‑performing model may lead to an inflation of the error rate. Within the conformal framework, \citet{magnani2023collective} demonstrated that unadjusted model selection breaks the required exchangeability between test and calibration samples, leading to invalid inference -- a phenomenon we demonstrate numerically in Appendix \ref{sec:P-TAMS_eff}. Similar concerns about double-dipping bias in other contexts of conformal inference have been discussed by \citet{einbinder2022training} and \citet{liang2024conformal}.

Conformalized Automated Model Selection (CAMS; \citealp{magnani2023collective, liang2023conformal, liang2024integrative, marandon2024adaptive}, among others) employs carefully designed calibration techniques to preserve exchangeability between test and calibration samples throughout the model training and selection process, ensuring valid inference conditional on the selected model. The CAMS framework represents a paradigm shift from conventional post-selection inference, which relies on fixed models and strong parametric assumptions, such as linearity and sparsity. In contrast, CAMS is applicable to a randomly selected model, chosen in a data-driven manner from a flexible toolbox, fully leveraging the power of the conformal approach. By unifying model training and statistical validation in a single, coherent process, CAMS allows practitioners to harness the flexibility of ML toolboxes without compromising statistical guarantees.

Let \(\mathbf C=\{\mathcal{C}_k : k \in [K]\}\) denote a toolbox of  models. Existing CAMS strategies can be categorized into two types: (a) model-oriented selection and (b) task-oriented selection.
Strategy (a) evaluates a classifier’s ability to distinguish between inliers (null model) and outliers (non-null model). For instance, as proposed by \citet{liang2024integrative}, the model is selected by maximizing the median difference between two sets of conformity scores computed for labeled inliers and labeled outliers, respectively. By contrast, Strategy (b) directly assesses a classifier’s utility in performing a specific task. For example, the methods in \citet{magnani2023collective} and \citet{marandon2024adaptive} select a model that maximizes (a proxy of) the number of rejections in the test samples. However, existing CAMS methods rely on joint exchangeability and are consequently ill-suited for structured inference. The next subsection develops new methodology to address the above issues.

\subsection{The P-TAMS algorithm}\label{subsec:P-TAMS}

We introduce the Pseudo-score-guided Transductive Automated Model Selection (P-TAMS) algorithm, which employs a transductive, task‑oriented model selection strategy for two primary reasons. First, unlike model‑oriented approaches that require labeled outliers \citep{liang2024integrative} -- which are often unavailable in OOD settings -- our model selection process operates without such supervision. Moreover, even when labeled outliers are available, they may exhibit distributional shifts relative to test outliers, potentially misleading model selection (cf. Section \ref{sec:impact_dist_shift}). Second, our design capitalizes on the structural knowledge inherent in test samples by formulating the selection criterion directly on the test set (transductive) while aligning it with the utility for performing a specific task, such as maximizing the number of rejections in OOD testing (task‑oriented).

A central challenge is ensuring pairwise exchangeability throughout the inferential process, where the SCQ procedure is deployed \emph{conditional on the selected model}. One may naturally consider using the number of rejections computed from the true scores \(\mathcal{V}\) and \(\widetilde{\mathcal{V}}\) as the selection criterion. However, this approach is not permissible because the criterion is contingent on \(H(t)\) [Eq.~\eqref{eq:H_function}], which treats \(\mathcal{V}\) and \(\widetilde{\mathcal{V}}\) asymmetrically. This asymmetry violates the pairwise exchangeability assumption essential for valid inference.

To address this challenge, we introduce carefully constructed pseudo-scores that serve as substitutes for the true scores to ensure statistical validity. Our selection criterion is defined as the \emph{pseudo number of rejections}, computed from these pseudo-scores. The design of the pseudo-score functions is guided by two key considerations. First, swap invariance with respect to \(\mathbf{X}\) and \(\tilde{\mathbf{X}}\) must be fulfilled throughout the entire training and selection process. Second, the resulting pseudo number of rejections should effectively capture the model's utility, mirroring the behavior observed with true scores. The P-TAMS algorithm proceeds in three main steps, each accompanied by detailed explanatory remarks.

\noindent\textbf{Step 1: Initial partition of test samples.} For each classifier \(\mathcal{C}_k \in \mathbf{C}\), let \(s_k(\cdot)\) denote the corresponding score function computed according to the invariance principle \eqref{eq:g_func_swap_inva}. We construct \(m\) pairs of preliminary conformal p-values via \eqref{eq:conformal_p}: \(\{(p_i^k \equiv p^k(X_i), \tilde{p}_i^k \equiv p^k(\tilde{X}_i)) : i \in \mathcal{D}^{\rm test}\}\), then form \(m\) pseudo conformal p-values by taking the pairwise minimum:
$
\bar{\mathcal{P}}_k = \left\{\min\left(p_i^k, \tilde{p}_i^k\right) : i \in \mathcal{D}^{\rm test}\right\}.
$
Finally, we apply the BH procedure to \(\bar{\mathcal{P}}_k\) at level \(\alpha_0\) to obtain an initial rejection set \(\bar{\mathcal{R}}_k\).

\begin{remark}
The primary goal of this step is to partition test samples into likely outliers and likely inliers while preserving swap invariance, which underpins the subsequent construction of pseudo-scores. After processing with the minimum operator, the likely outliers (and likely inliers) tend to exhibit smaller (and larger) p-values compared to other test units, making them more (or less) likely to be rejected. Although the minimum operator distorts the original p-values, it retains sufficient discriminatory power to distinguish between the two groups while preserving swap invariance. In contrast, alternative swap invariance operators, such as the maximum or sum, would result in significant information loss and fail to achieve effective separation. Moreover, the P‑TAMS procedure is not sensitive to the choice of \(\alpha_0\), which may be set slightly above the nominal FDR level \(\alpha\) (with \(2\alpha\) as the default).
\end{remark}

\noindent\textbf{Step 2: Constructing pseudo conformity scores.} This step represents the core innovation of P-TAMS.  Denote \(\mathcal{V}_k = \{V_i^k: i \in \mathcal D^{\text{test}}\}\) and \(\widetilde{\mathcal{V}}_k = \{\tilde{V}_i^k: i \in \mathcal D^{\text{test}}\}\) the true conformity scores for classifier \(\mathcal{C}_k\). The pseudo scores \(\mathcal{U}^k = \{U_i^k: i \in \mathcal D^{\text{test}}\}\) and \(\widetilde{\mathcal{U}}^k = \{\tilde{U}_i^k: i \in\mathcal D^{\text{test}}\}\) are then computed as follows. 
\begin{itemize}
\item For likely outliers, we apply the min-max operator:
\begin{equation}\label{eq:pseudo_outlier}
    U_i^k=\min(V_i^k, \tilde{V}_i^k), \quad \tilde{U}_i^k = \max(V_i^k, \tilde{V}_i^k),\quad i \in \bar{\mathcal{R}}_k.
\end{equation}
\item For the remaining units (likely inliers), let \(\{b_i\}\) denote i.i.d. \(\mathrm{Bernoulli}(1/2)\) variables, we apply a random coin-flipping operator:
\begin{equation}\label{eq:pseudo_inlier}
    U_i^k = (1 - b_i)V_i^k + b_i\tilde{V}_i^k, \quad \tilde{U}_i^k = b_i V_i^k + (1 - b_i)\tilde{V}_i^k,\quad i \notin \bar{\mathcal{R}}_k.
\end{equation}
\end{itemize}

\begin{remark}
The utilization of pseudo-scores effectively avoids the pitfalls associated with using true scores for model selection. Specifically, the mirror process \(H(t)\) computed with true scores breaks the pairwise exchangeability (conditional on the selected model). This issue is mitigated by employing pseudo-scores, which faithfully preserve the swap invariance with respect to \(\mathbf{X}\) and \(\tilde{\mathbf{X}}\). Additionally, our choice of symmetrization operators reduces the distortion of the true scores: (a) for likely outliers, operator \eqref{eq:pseudo_outlier} preserves the tendency for \(V_i^k\) to be smaller than \(\tilde{V}_i^k\) under the alternative; (b) for likely inliers, operator \eqref{eq:pseudo_inlier} reflects the random ordering between \(V_i^k\) and \(\tilde{V}_i^k\) under the null.
\end{remark}

\noindent\textbf{Step 3: Model selection and final decisions.} For each \(k \in [K]\), we implement the SCQ procedure using pseudo scores \(\mathcal{U}^k\) and \(\widetilde{\mathcal{U}}^k\), and record the number of pseudo rejections \(r_k\). We then select the classifier  with the largest \(r_k\), denoted by \(\mathcal{C}_*\). Finally, we apply the SCQ procedure again with the selected model \(\mathcal{C}_*\) to output the final decision \(\mathcal{R}_*\).

\begin{remark}
By leveraging the inherent structure of the mirror-based algorithms, P-TAMS integrates seamlessly with SCQ through an innovative design. Together, SCQ and P-TAMS form a unified framework that relies solely on the mild assumption of pairwise exchangeability, thereby offering enhanced flexibility and broader applicability -- making it particularly well-suited for detecting structured patterns. Moreover, rather than resorting to additional sample splitting (cf. \citealp{magnani2023collective, marandon2024adaptive}), P-TAMS employs carefully constructed symmetrization operators that intrinsically satisfy the swap invariance requirement. This novel design not only preserves the validity of the SCQ  but also delivers efficiency gains at no extra cost of valuable samples. A schematic illustration of P-TAMS is provided in Figure~\ref{fig:P-TAMS}, with its key steps outlined in Algorithm~\ref{alg:P-TAMS} in Appendix~\ref{sec:SCQ_P-TAMS algorithm}. An extension of P-TAMS is provided in Algorithm~\ref{alg:P-TAMS+} in Appendix~\ref{sec:P-TAMS+}.

\end{remark} 

\begin{figure}[ht]
	\centering
	\includegraphics[width=0.85\linewidth]{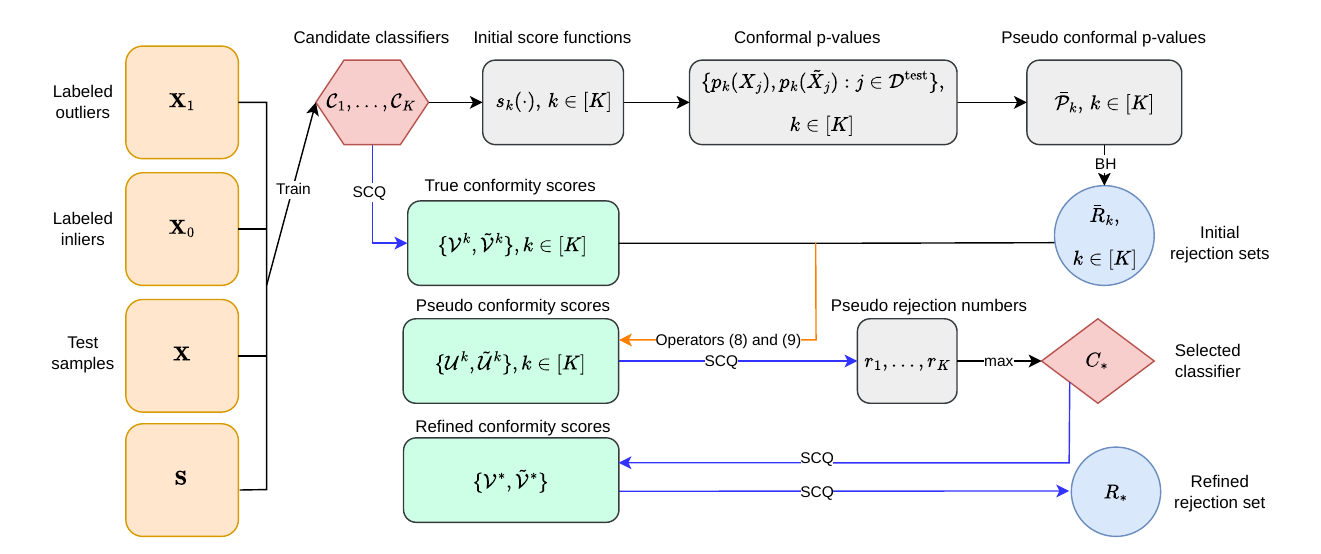}
	\caption{\small Schematic representation of the P-TAMS algorithm.}
	\label{fig:P-TAMS}
\end{figure}

The following theorem establishes the finite-sample validity of P-TAMS. 

\begin{theorem}\label{thm:P-TAMS-validity}
Under the conditions in Proposition \ref{prop:pairwise_scores}, the SCQ procedure refined with P-TAMS controls the FDR at level \(\alpha\).
\end{theorem}

\section{Asymptotic Theories}\label{sec:theories}

\subsection{Mirror Calibration and Equivalent Thresholding Rules}\label{subsec:BC=SCQ}

Existing conformalized OOD testing methods involve constructing conformal p-values followed by applying well-established multiple testing procedures. For example, \citet{bates2023testing} showed that the conformal p-values satisfy the PRDS (positive regression dependence on subsets) condition and utilize the theory of \citet{benjamini2001control} to prove the validity of FDR control.  Meanwhile, recognizing that recalibrated conformal p-values exhibit more complex dependence, \citet{liang2024integrative} employed the conditional calibration framework of \citet{fithian2022conditional} to establish the FDR theory. However, in structured OOD testing scenarios, the frameworks in \cite{benjamini2001control} and \cite{fithian2022conditional} are inapplicable due to the violation of joint exchangeability. Instead, we establish the finite-sample FDR theory in Section \ref{subsec:SCQ} using the martingale theory in \cite{ZhaSun25-PLIS} and the e-BH theory in \cite{wang2022false}. 

This section introduces novel tools to analyze the asymptotic properties of SCQ. The asymptotic theory provides insights into two key questions. First, while SCQ controls the FDR in finite samples, does it tend to be overly conservative? Second, SCQ employs a weighting strategy to leverage side information -- when and why is this approach effective? Sections~\ref{subsec:FDR-analysis} and \ref{subsec:power-analysis} develop a theoretical framework to characterize the conditions under which SCQ is asymptotically anti-conservative and demonstrates gains in power.

We begin with some preliminaries. Suppose we have constructed \(m\) pairs of conformity scores, \(\{(V_j, \tilde{V}_j): j \in \mathcal{D}^{\text{test}}\}\), that satisfy condition \eqref{eq:pwex-scores}. Two thresholding rules can be considered. The first thresholds conformal q-values defined in \eqref{eq:scq}; its rejection set is denoted by 
$
\mathcal{R}_{scq} = \{j \in \mathcal{D}^{\text{test}}: q_j \leq \alpha\}.
$
The second thresholds conformity scores directly. Define 

\begin{equation}\label{eq:tau}
	\tau = \max\left\{t \in \mathcal{V} \cup \tilde{\mathcal{V}}: H(t) \coloneqq \frac{1 + \sum_{j\in\mathcal{D}^{\rm test}} \mathbb{I}(\tilde{V}_j \leq t, \tilde{V}_j < V_j)}{\left[\sum_{j\in\mathcal{D}^{\rm test}} \mathbb{I}(V_j \leq t, V_j < \tilde{V}_j)\right]\vee1} \leq \alpha\right\}.
\end{equation}
The subsequent rejection set is given by

\begin{equation}\label{dec:BC}
	\mathcal{R}_{bc} = \{j \in \mathcal{D}^{\text{test}}: V_j \leq \tau, \, V_j < \tilde{V}_j\}.	
\end{equation}

\begin{remark}
The subscript ``bc'' in \eqref{dec:BC} denotes a Barber–Candès-type (BC) procedure \citep{barber2015controlling}. The generic BC algorithm -- also known as the Selective SeqStep+ algorithm -- was originally developed for knockoff filters in regression. SCQ and knockoff only connect at the thresholding step. In knockoff inference, the main task is to construct knockoff variables that mimic the dependence structure of the original covariates. In contrast, SCQ starts from an OOD testing problem with unlabeled test samples and mirror null samples. The key challenge is therefore the construction of conformity-score pairs that can incorporate side information from the test sample while preserving pairwise exchangeability under the null.
	\end{remark}

The q-value thresholding rule is intuitive and interpretable, but its direct theoretical analysis can be challenging. In contrast, we can leverage exchangeability properties and asymptotic arguments to give a rigorous theoretical treatment of the score-thresholding rule \eqref{dec:BC}. The following proposition establishes the equivalence between the two thresholding rules; as a result, our subsequent analysis will focus on the BC-type algorithm \eqref{dec:BC}, which applies equally to the SCQ procedure.

\begin{proposition}\label{prop:equiv_SCQ_BC}
	Given $m$ pairs of conformity scores $\{(V_j,\tilde{V}_j):j\in\mathcal{D}^{\rm test}\}$. Consider the rejection sets \(\mathcal{R}_{\text{bc}}\) and \(\mathcal{R}_{\text{scq}}\) defined above. Then \(\mathcal{R}_{\text{bc}} = \mathcal{R}_{\text{scq}}\).
\end{proposition}

\subsection{Anti-conservativeness of FDR control}\label{subsec:FDR-analysis}

The BC-type procedure $\mathcal{R}_{bc}$ given in \eqref{dec:BC} can be conservative. However, as signal strength increases, BC can adaptively attain the nominal FDR level -- a phenomenon initially noted by \cite{barber2015controlling}; see also Appendix \ref{sec:FDR_achieve_alpha} for an illustration. This section formally establishes this asymptotic attainment theory, tailored for conformal setups.

Consider the paired conformal \(p\)-values \(\{(p_j\equiv p(X_j), \tilde{p}_j\equiv p(\tilde{X}_j)): j \in \mathcal{D}^{\text{test}}\}\) constructed via \eqref{eq:conformal_p}, where the score function \(s(\cdot)\) is assumed continuous.  Let \(\pmb{w} = (w_j: j \in\mathcal{D}^{\rm test})\) denote a generic sequence of p-value weights, and define the weighted \(p\)-values \((V_j \equiv p_j/w_j,\; \tilde{V}_j \equiv \tilde{p}_j/w_j)\) for each \(j \in \mathcal{D}^{\mathrm{test}}\). Consider a class of thresholding rules \(\pmb{\delta^w}(t) = \{\delta^{w_j}(t): j \in \mathcal{D}^{\rm test}\} \in\{0,1\}^m\), where \(t > 0\) is a threshold and $\delta^{w_j}(t) = \mathbb{I}(p_j/w_j \leq t, p_j/w_j < \tilde{p}_j/w_j)$. In addition, define the marginal FDR of $\pmb{\delta^w}(t)$ and the corresponding oracle threshold: 
{\footnotesize\begin{equation}\label{eq:oracle_threshold} 
	Q(t,\pmb{w}) = \frac{ \mathbb{E}\left[\sum_{j\in\mathcal{D}^{\rm test}} \II(Y_j=0,\delta^{w_j}(t)=1) \right] }{\mathbb{E}\left[\sum_{j\in\mathcal{D}^{\rm test}} \II(\delta^{w_j}(t)=1) \right]},\quad
	t^\alpha_{\text{OR}}(\pmb{w}) = \sup\left\{t \in (0,1): Q(t, \pmb{w}) \leq \alpha\right\}.
\end{equation}}

Asymptotic FDR theory generally requires some form of weak dependence. Let
$R_j(t) = \mathbb{I}(V_j \leq t,\; V_j < \tilde{V}_j)$ and $ \tilde{R}_j(t) = \mathbb{I}(\tilde{V}_j \leq t,\; \tilde{V}_j < V_j).$ To lay the groundwork for our analysis, we adopt the following condition to characterize weak dependence:
\begin{equation}\label{eq:cov}
	\dfrac{1}{m^2}\sum\limits_{\substack{i<j,\\i,j\in\mathcal{D}^{\rm test}}} \text{Cov}\left[R_i(t),R_j(t)\right]=o(1),\quad
	\dfrac{1}{m^2}\sum\limits_{\substack{i<j,\\i,j\in\mathcal{D}^{\rm test}}} \text{Cov}\left[\tilde{R}_i(t),\tilde{R}_j(t)\right]=o(1).
\end{equation}
Denote the numbers of test inliers, test outliers, and calibration data as $m_0$, $m_1$, and $N$, respectively. The lemma below shows that \((V_j, \tilde{V}_j)\) are weakly dependent. 

\begin{lemma}\label{lem:cov}
Suppose condition \eqref{eq:data_exch_joint} holds and there exist $\epsilon_1,\epsilon_N\in(0,1)$, such that $m_1\asymp m^{1-\epsilon_1}$ and $N\asymp m^{1+\epsilon_N}$\footnote{For positive sequences $a(m)$ and $b(m)$, we write $a(m)\asymp b(m)$ if there exist $C_1,C_2>0$ such that $a(m)\leq C_1\cdot b(m)$ and $b(m)\leq C_2\cdot a(m)$ for all $m\geq1$.}. Then $\{(V_j,\tilde{V}_j):j\in\mathcal{D}^{\rm test}\}$ satisfy \eqref{eq:cov}.
\end{lemma}

Next, our asymptotic theory requires the following two regularity conditions.

\begin{assumption}\label{ass:sep}
	\(\dfrac{1}{m}\sum_{j\in\mathcal{D}^{\rm test}} \mathbb{P}\left(\tilde{V}_j < V_j \mid  Y_j = 1\right)=o(1)\).
\end{assumption}

\begin{remark}
This assumption is mild, essentially requiring only that a subset of signals can be reliably separated from noise -- a prerequisite for any meaningful FDR analysis with vanishingly small $\alpha$ (cf. \citealp{TonyCaioptimal2017}). Under the sparse, high‑dimensional settings commonly studied (e.g., \citealp{donoho2004higher, meinshausen2006estimating, arias2017distribution}), Assumption~\ref{ass:sep} is satisfied by the weighted conformal p‑values $(V_j,\tilde{V}_j)$; detailed justifications are provided in Appendix~\ref{sec:proof_sep_ass}.
\end{remark}

The second assumption, implied by the condition in Equation (7) of \cite{storey2004strong}, is also needed in our theoretical analysis.

\begin{assumption}\label{ass:ER=Om}
	$\mathbb{E}\left[\sum_{j\in\mathcal{D}^{\rm test}}\mathbb{I}(V_j\leq t^\alpha_{\text{OR}}(\pmb{w}) ,V_j<\tilde{V}_j)\right]\asymp m$.
\end{assumption}

Finally, combining the results in Proposition \ref{prop:equiv_SCQ_BC} and Lemma \ref{lem:cov}, we show in the next theorem that the FDR of SCQ asymptotically attains the nominal level $\alpha$.

\begin{theorem}\label{thm:asymptotic-FDR}
	Suppose we have weighted p-values $\{(V_j,\tilde{V}_j):j\in\mathcal{D}^{\rm test}\}$ obtained from the SCQ procedure in Section \ref{subsec:SCQ}. Consider the rejection set $\mathcal{R}_{scq}$ defined in \eqref{dec:scq}. Then, under Assumptions \ref{ass:sep}--\ref{ass:ER=Om} and the conditions in Lemma \ref{lem:cov}, we have 
$
\lim_{m \to \infty}  \mathbb{E}\left[ \dfrac{|\mathcal{R}_{scq} \cap \mathcal{H}_0|}{|\mathcal{R}_{scq}|\vee 1} \right] = \alpha.
$
\end{theorem}

\subsection{Asymptotic power analysis}
\label{subsec:power-analysis}

We investigate the effectiveness of p-value weighting through a novel power analysis. Existing theory on p-value weighting (e.g., \citealp{genovese2006false, cai2022laws, liang2024integrative}), developed under the BH framework, is not applicable to BC-type algorithms. Furthermore, prior work based on conformal p-values \citep{liang2024integrative} has considered only oracle thresholds (defined in \eqref{eq:oracle_threshold}), while theories for data-driven thresholds \eqref{dec:BC} remain undeveloped. We address these gaps by examining weight informativeness within the BC framework and developing theory for both oracle and data-driven setups. 

Suppose we have constructed \(m\) pairs of weighted p-values $\{(V_j \equiv p_j/w_j,\; \tilde{V}_j \equiv \tilde{p}_j/w_j):j\in\mathcal{D}^{\rm test}\}$ via \eqref{eq:conformal_p} and \eqref{eq:weighted_p}.
Let \(\eta_j = \mathbb{I}(p_j < \tilde{p}_j)\) for each $j\in\mathcal{D}^{\rm test}$. The following assumption indicates that an effective weighting strategy should, in general, prioritize the rejection of units that are more likely to be outliers in light of side information. 

\begin{assumption}\label{ass:weight_info}
	The weights \((w_j:j\in\mathcal{D}^{\rm test})\) satisfy:
	\begin{equation}\label{eq:weight_inequality}
		\frac{\sum_{j\in\mathcal{D}^{\rm test}} \mathbb{P}(\eta_j = 1,  Y_j = 0 \mid S_j)}{\sum_{j\in\mathcal{D}^{\rm test}} \mathbb{P}( Y_j = 0 \mid S_j) w_j}
		\cdot
		\dfrac{\sum_{j\in\mathcal{D}^{\rm test}} \mathbb{P}(\eta_j = 1,  Y_j = 1 \mid S_j)}{\sum_{j\in\mathcal{D}^{\rm test}} \mathbb{P}(\eta_j = 1,  Y_j = 1 \mid S_j) w_j^{-1}}
		\geq 1.
	\end{equation}
\end{assumption}

\begin{remark}
Assumption~\ref{ass:weight_info}, inspired by \citet{genovese2006false}, requires the weights to be positively aligned with signal enrichment revealed by the side information. 
Intuitively, larger weights should be assigned to regions where, conditional on \(S_j\), signals are more frequent or more likely to yield rejections. This condition arises naturally in many structured settings (cf. \citealp{li2019multiple, cai2022laws}): in grouped testing, it assigns larger weights to groups with higher signal proportions; in spatial testing, it upweights neighborhoods where signals are locally concentrated and downweights regions dominated by nulls.
\end{remark}	

The next assumption is a convexity-type regularity condition on the conditional alternative distribution of the conformal p-values. Its role is to ensure that re-ranking test units by informative weights does not reduce oracle power. This condition formalizes that principle and is analogous to shape restrictions commonly used in the theory of weighted multiple testing (cf. \citealp{hu2010false, cai2022laws,liang2024integrative}).

\begin{assumption}\label{ass:F_alt}
	\(\big\{F_{1,1,j}(t) \coloneqq \mathbb{P}(p_j \leq t \mid \eta_j = 1,  Y_j = 1, S_j) :j\in\mathcal{D}^{\rm test}\big\}\) satisfy:
	\begin{equation}
		\sum_{j\in\mathcal{D}^{\rm test}} a_j F_{1,1,j}(t/x_j) \geq \sum_{j\in\mathcal{D}^{\rm test}} a_j F_{1,1,j}\left(\frac{t \sum_{i=1}^m a_i}{\sum_{i=1}^m a_i x_i}\right),
	\end{equation}
	for any \(0 \leq a_j \leq 1\), \(\min_{1 \leq j \leq m} w^{*-1}_j \leq x_j \leq \max_{1 \leq j \leq m} w^{*-1}_j\), and \(t > 0\), where 
	\begin{equation}\label{eq:w_star}
		w^*_j = w_j \cdot  \dfrac{\mathbb{P}(\eta_j = 1 \mid  Y_j = 0, S_j)\sum_{i\in\mathcal{D}^{\rm test}} \mathbb{P}(\eta_i = 1,  Y_i = 0 \mid S_i)}{\sum_{i\in\mathcal{D}^{\rm test}} w_i \mathbb{P}(\eta_i = 1,  Y_i = 0 \mid S_i)}.
	\end{equation}
\end{assumption}

Our power analysis proceeds in two steps. First, under the oracle setup using the threshold in \eqref{eq:oracle_threshold}, we compare the power of weighted and unweighted schemes at the same FDR level, revealing the benefits of weighting for producing superior rankings. Second, we examine the more intricate data-driven setup, analyzing power gains of weighted BC-type methods using thresholds in \eqref{eq:tau}, a setting not previously addressed in the literature.

\noindent\textbf{1. Power Analysis for the Oracle Setup.} Let \(\pmb{\delta}_{\text{OR}}^{\pmb{w}} = (\delta_{\text{OR}}^{w_j}: j \in\mathcal{D}^{\rm test}) \in \{0, 1\}^m\) denote the oracle rule with weights \(\pmb{w}\) and threshold \(t_{\text{OR}} \equiv t^\alpha_{\text{OR}}(\pmb{w})\) from \eqref{eq:oracle_threshold}, where \(\delta_{\text{OR}}^{w_j} = \mathbb{I}(p_j/w_j \leq t_{\text{OR}}, p_j/w_j < \tilde{p}_j/w_j)\). Define \(\Psi_{\text{OR}}(\pmb{w}) = \mathbb{E}\left[\sum_{j\in\mathcal{D}^{\rm test}} \mathbb{I}( Y_j = 1, \delta_{\text{OR}}^{w_j} = 1)\right]\) as the expected number of true positives. Denote the structure-adaptive weights utilized in the SCQ procedure as $\pmb{w_s}$. The next theorem establishes the efficiency gain by weighting.

\begin{theorem}\label{thm:oracle_power}
	Consider the oracle rules \(\pmb{\delta}_{\text{OR}}^{\pmb{w_s}}\) and \(\pmb{\delta}_{\text{OR}}^{\pmb{1}}\). Suppose that the structure-adaptive weights \(\pmb{w_s}\) are swap invariant (cf. \eqref{eq:g_func_swap_inva}). Under Assumptions \ref{ass:weight_info}--\ref{ass:F_alt} and the exchangeability condition \eqref{eq:data_exch_joint}, if the score function \(s(\cdot)\) utilized in \eqref{eq:conformal_p} is continuous and satisfies \eqref{eq:s_function_permu}, then we have $\Psi_{\text{OR}}(\pmb{w_s}) \geq \Psi_{\text{OR}}(\pmb{1}).$
\end{theorem}

\noindent\textbf{2. Power Analysis for the Data-Driven Setup.}
Let \(\pmb{\delta}_{\text{DD}}^{\pmb{w}} = (\delta_{\text{DD}}^{w_j}: j\in\mathcal{D}^{\rm test}) \in \{0, 1\}^m\) denote the data-driven rule applying the BC algorithm with weights \(\pmb{w}\) and threshold \(\tau \equiv \tau^\alpha_{\text{DD}}(\pmb{w})\) from \eqref{eq:tau}, where \(\delta_{\text{DD}}^{w_j} = \mathbb{I}(p_j/w_j \leq \tau, p_j/w_j < \tilde{p}_j/w_j)\). Define \(\Psi_{\text{DD}}(\pmb{w}) = \mathbb{E}\left[\sum_{j\in\mathcal{D}^{\rm test}} \mathbb{I}( Y_j = 1, \delta_{\text{DD}}^{w_j} = 1)\right]\) as the expected number of true positives of the data-driven BC method with weights $\pmb w$. Consider $t^\alpha_{\text{OR}}(\pmb{w})$ in \eqref{eq:oracle_threshold}. 
The following theorem shows that for generic weights \(\pmb{w}\), the data-driven BC algorithm $\pmb{\delta}_{\text{DD}}$ asymptotically achieves the power performance of $\pmb{\delta}_{\text{OR}}$.

\begin{theorem}\label{thm:BC_power}
	Suppose the generic weighted p-values $\{(p_j/w_j,\tilde{p}_j/w_j):j\in\mathcal{D}^{\rm test}\}$ satisfy pairwise exchangeability \eqref{eq:pwex-scores}. Under Assumptions \ref{ass:sep}--\ref{ass:ER=Om}, and the conditions in Lemma \ref{lem:cov}, we have 
	$
	\lim\limits_{m\to\infty}\dfrac{\Psi_{\text{DD}}(\pmb{w})}{\Psi_{\text{OR}}(\pmb{w})}= 1.
	$
\end{theorem}

Theorem \ref{thm:DD_power} shows the benefits of informative weighting in data-driven setups.

\begin{theorem}\label{thm:DD_power}	
	Consider two sets of conformity scores employing structure-adaptive weights \(\pmb{w_s}\) and constant weights \(\pmb{1}\), respectively. If both sets of scores satisfy condition \eqref{eq:pwex-scores} and \(s(\cdot)\) utilized in \eqref{eq:conformal_p} is continuous and satisfies \eqref{eq:s_function_permu}, then under Condition \eqref{eq:data_exch_joint}, Assumptions \ref{ass:sep}--\ref{ass:F_alt}, and conditions in Lemma \ref{lem:cov}, we have
	\(
	\lim\limits_{m \to \infty}\dfrac{\Psi_{\text{DD}}(\pmb{w_s})}{\Psi_{\text{DD}}(\pmb{1})} \geq 1.
	\)
\end{theorem}

\section{Simulation Studies}\label{sec:simulation}

This section investigates the numerical performance of our proposed methods for structured OOD testing. Section~\ref{sec:simu setup} outlines the basic setup, while Sections~\ref{sec:impact_density} through \ref{sec:impact_imbalance} assess the effects of feature dimensionality, distribution shifts, and data imbalance across different methods. Additional numerical experiments are provided in Appendices \ref{sec:P-TAMS_eff}--\ref{sec:FDR_achieve_alpha}.

\subsection{Simulation setup}\label{sec:simu setup} 

We generate test samples \(X_j \in \mathbb{R}^{p}\) from the following hierarchical model:
\begin{equation}\label{eq:hierarchical_model}
\Big(Y_j = 1 \mid S_j\Big) \overset{\text{ind.}}{\sim} \text{Ber}\big(\pi(S_j)\big), \quad
\Big(X_j \mid Y_j, S_j\Big) \overset{\text{ind.}}{\sim} (1-Y_j)\cdot \mathcal{N}(\mathbf{0}, I_p) + Y_j \cdot F_{1S_j}.
\end{equation}
Here, \(\pi(S_j)\) denotes the local sparsity level, \(\mathbf{c} \in \mathbb{R}^p\) is a constant vector with all components equal to \(c\), and \(I_p\) denotes a \(p \times p\) diagonal matrix. We fix the test sample size at \(m = 3000\) and generate \(5000\) null samples from \(\mathcal{N}(\mathbf{0}, I_p)\). These null samples are randomly split into training, calibration, and mirror subsets, with \(|\mathcal{D}_0^{\rm tr}|=|\mathcal{D}_0^{\rm cal}|=1000\) and \(|\mathcal{D}_0^{\rm mir}|=3000\). The auxiliary variable is set as $S_j = j$, which may represent, for example, the time points or spatial locations at which the samples are collected. Accordingly, samples with close or adjacent values in the indices (e.g., samples collected at nearby time points or locations) tend to exhibit similar patterns.
The structural heterogeneity is characterized by varying sparsity levels, with elevated signal frequencies in the following intervals:

$$
	\pi(S_j) = 0.6, \quad S_j \in [201,300] \cup [601,700]; \quad \pi(S_j) = 0.9, \quad S_j \in [1000,1100] \cup [1400,1500],
$$
and \(\pi(S_j)=0.01\) for the rest. The alternative distribution \(F_{1S_j}\) is defined as

$$
    F_{1S_j} = \mathcal{N}(\mathbf{\mu}, I_p), \quad S_j \in [1,1500]; \quad F_{1S_j} = \mathcal{N}(\mathbf{-2}, 0.5^2 \cdot I_p), \quad S_j \in [1501,3000].
$$
The weights \(w(S_j)\) employed in the SCQ procedure are constructed based on the estimated \(\pi(S_j)\) and are specifically designed to satisfy Condition \eqref{eq:g_func_swap_inva}; we provide detailed description on weight construction in Appendix \ref{sec:weight_construction}. In all experiments, the FDR level is set to 0.05. The efficiency is evaluated using the average power $\mathrm{AP} := \mathbb{E}\left[\frac{\mid\mathcal{R}\cap\mathcal{H}_1\mid}{\max\{1, \mid \mathcal{H}_1 \mid\}}\right].$ Both the FDR and AP are reported by empirically averaging the results over 500 replications.

We compare the following methods in our experiments; implementation details of each method are provided in Appendix \ref{sec:detailed_methods}:

\begin{enumerate}

	\item[(a)] \textbf{cfBH-OCC} and \textbf{ICP-OCC} apply the Storey-BH procedure to the split conformal p-values \citep{bates2023testing} and integrative conformal p-values \citep{liang2024integrative}, respectively, using one-class SVM scores.
	
	\item[(b)] \textbf{AdaDetect-KDE} and \textbf{CLAW-KDE} apply Storey-AdaDetect \citep{marandon2024adaptive} and CLAW \citep{ZhaSun25-CLAW}, respectively, using KDE scores.
	
	\item[(c)] \textbf{AdaDetect-PU} and \textbf{CLAW-PU} apply Storey-AdaDetect and semi-supervised CLAW, respectively, using RF-based PU-learning scores.
	
	\item[(d)] \textbf{SCQ-OCC} and \textbf{SCQ-BIC} apply SCQ with SVM and RF scores, respectively.
	
	\item[(e)] \textbf{SCQ-KDE} and \textbf{SCQ-PU} apply SCQ with KDE and PU-learning scores, respectively.
	
	\item[(f)] \textbf{SCQ+P-TAMS} applies SCQ with a model selected by P-TAMS from a toolbox of candidate ML models specified below.
\end{enumerate}

\subsection{Impacts of feature dimensionality}\label{sec:impact_density}

\begin{figure}[t]
	\centering
	
	\begin{subfigure}{0.66\textwidth} 
		\centering
		\includegraphics[width=0.85\linewidth]{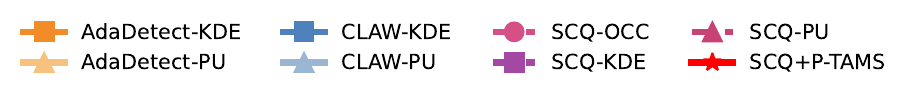} 
	\end{subfigure}
	
	\begin{subfigure}{\textwidth} 
		\centering
		\includegraphics[width=0.9\linewidth]{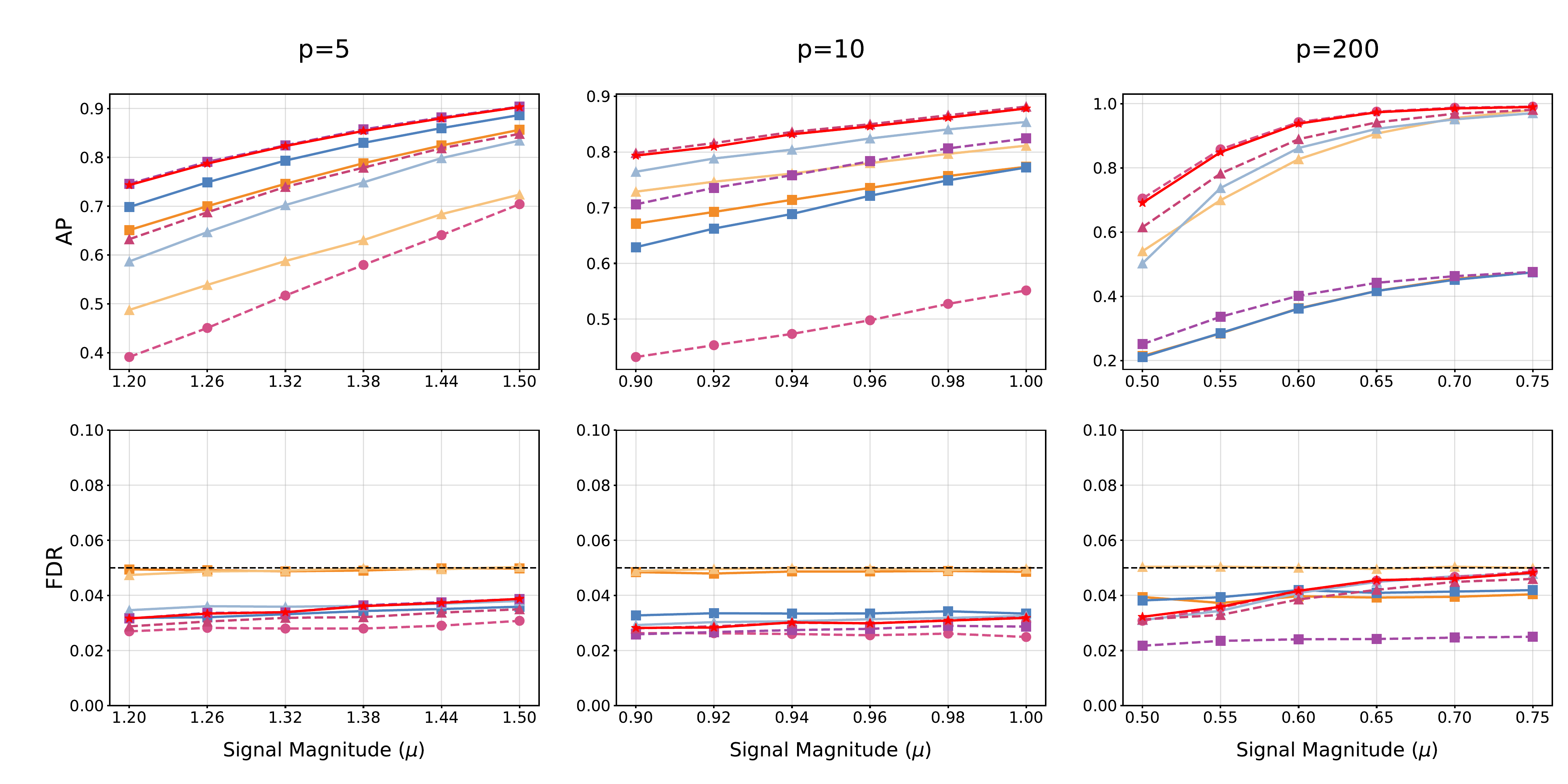} 
	\end{subfigure}
	
	\caption{\small AP and FDR comparison of the SCQ variants, together with P-TAMS that selects among them, against density-ratio–based methods at $\alpha=0.05$.}
	\label{fig:SCQ_vs_density_methods}
\end{figure}
This section investigates how feature dimensionality affects the performance of the following methods: AdaDetect-KDE, AdaDetect-PU, CLAW-KDE, CLAW-PU, and four SCQ variants: SCQ-OCC, SCQ-KDE, SCQ-PU, and SCQ+P-TAMS. We report the AP and FDR as functions of $\mu$ for three settings ($p=5$, $p=10$, and $p=200$), which are summarized in the left, middle, and right columns of Figure \ref{fig:SCQ_vs_density_methods}. The following observations can be made. First, all methods successfully control the FDR below 0.05. Second, both SCQ-KDE and SCQ-PU demonstrate noticeable power improvements over their AdaDetect counterparts by effectively leveraging structural information. Third, in the low-dimensional setting ($p = 5$), KDE-based methods achieve higher AP than PU-learning-based approaches, with SCQ-KDE attaining the best overall performance. As the dimensionality grows, however, AdaDetect-PU and CLAW-PU begin to outperform their KDE-based counterparts. This is because kernel density estimation methods become increasingly unstable in high-dimensional settings.
In this regime, SCQ-PU achieves the highest power. Fourth, in the high-dimensional case ($p = 200$), KDE-based methods are at a clear disadvantage, whereas SCQ-OCC consistently delivers the strongest performance among all SCQ variants. Finally, across all settings, P-TAMS successfully identifies the best-performing SCQ variant, resulting in the uniformly superior performance of SCQ+P-TAMS.

\subsection{Performance under distribution shifts}\label{sec:impact_dist_shift} 

To examine the effect of distribution shifts between labeled and test outliers, we evaluate the performance of SCQ-OCC, SCQ-BIC, cfBH-OCC, and ICP-OCC. Note that ICP \citep{liang2024integrative}, which also uses a p-value weighting strategy, differs fundamentally from SCQ and may be adversely affected by distributional shifts. To ensure a fair comparison, both SCQ and ICP are implemented using a fixed classifier (OneClassSVM). Additional comparisons regarding AMS strategies are provided in Appendix \ref{sec:P-TAMS_vs_ICPAMS}.

\begin{figure}[t] 
	\centering 
	\begin{subfigure}{0.5\textwidth} 
		\centering
		\includegraphics[width=0.85\linewidth]{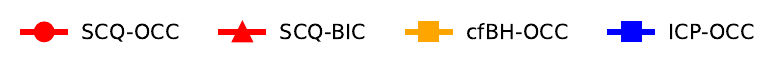} 
	\end{subfigure}
	\begin{subfigure}{\textwidth}
		\centering
		\includegraphics[width=0.9\linewidth]{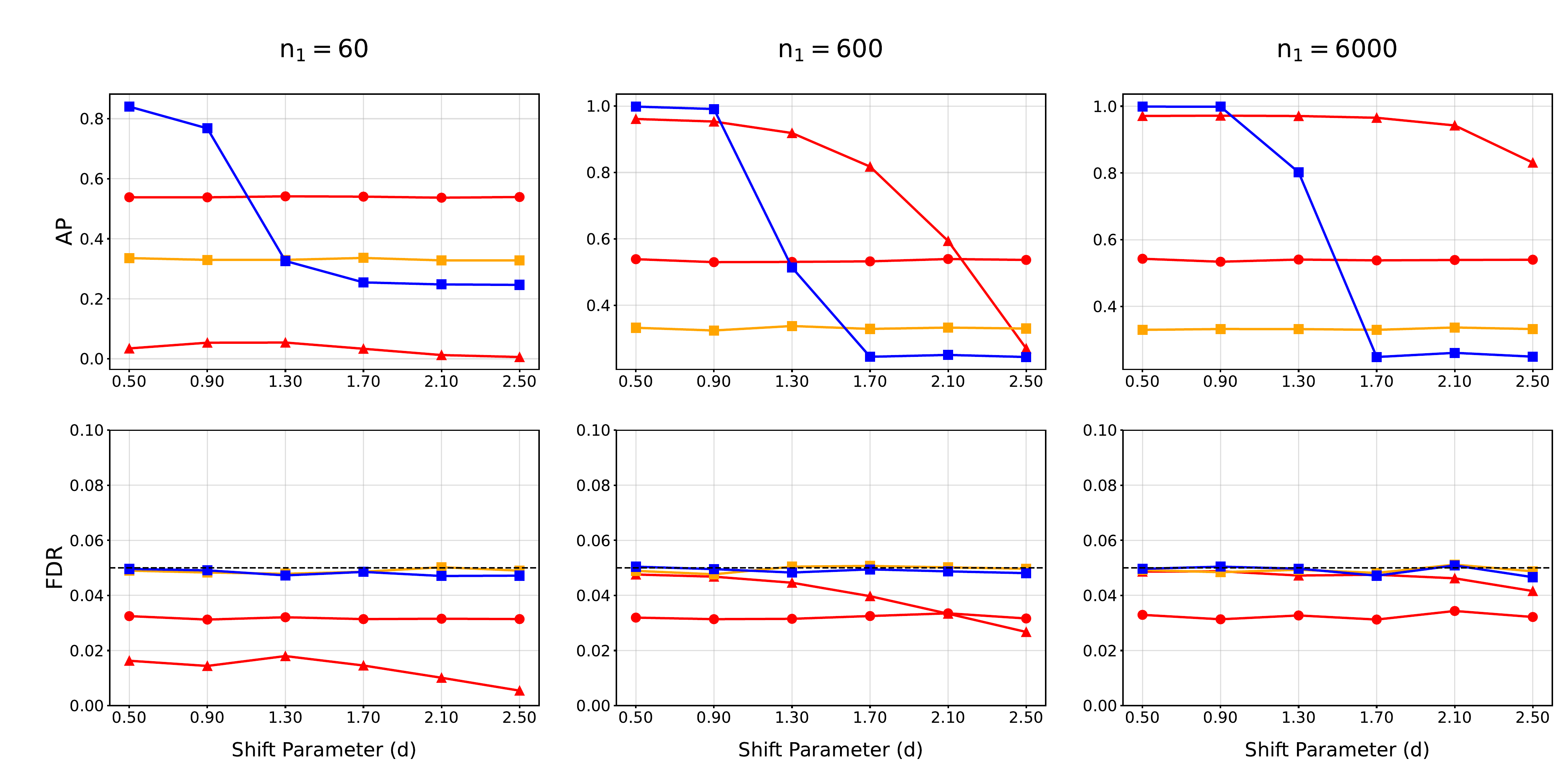} 
	\end{subfigure}
	\caption{\small Comparison of AP and FDR between SCQ, cfBH, and ICP at $\alpha=0.05$. }
	\label{fig:SCQ_vs_ICP}
\end{figure}

Our experiments consider three scenarios, with the number of labeled outliers $n_1$ set to 60, 600, and 6000. To induce distribution shifts, half of the labeled outliers follow $\mathcal{N}(\mathbf{-2},0.5^2\cdot I_p)$, matching the test outlier distribution, while the other half follow $\mathcal{N}(\mathbf{d},I_p)$, differing from the test distribution. With $\mu$ fixed at 0.5 and $p$ set to 200, we vary $d$ from 0.5 (indicating no shift) to 2.5 to represent increasing degrees of distribution shift. 

The results, shown in Figure~\ref{fig:SCQ_vs_ICP}, reveal several noteworthy patterns. First, when the distribution shift is absent or minimal, ICP-OCC performs well, effectively leveraging the labeled outliers to construct informative weights. However, its performance deteriorates progressively as the degree of shift increases, eventually falling below that of the unweighted cfBH-OCC. Second, unlike ICP -- which may suffer from negative learning -- SCQ-OCC consistently outperforms cfBH-OCC across all settings. Third, the flexibility of the SCQ procedure allows it to incorporate labeled outliers by employing a pre-trained BIC. In contrast to ICP, SCQ-BIC demonstrates robust performance and, in most cases, achieves the highest AP, even when the degree of distribution shift is substantial. Finally, the performance of SCQ-BIC relative to SCQ-OCC varies, with the former sometimes outperforming or underperforming the latter. This observation motivates the use of P-TAMS to uniformly enhance both variants of SCQ, as illustrated in the next subsection.

\subsection{Performance under data imbalance}\label{sec:impact_imbalance}

We generate data from a multivariate Gaussian mixture model (see Appendix \ref{sec:detailed_methods} for details) following the approach of \cite{bates2023testing} and \cite{liang2024integrative} to illustrate the selection dilemma under data imbalance and to demonstrate the effectiveness of P-TAMS. In our experiments, we set $\mu=0.5$ and $p=150$, generate 9000 labeled inliers, and vary the number of labeled outliers from 0 to 3000 to examine the impact of data imbalance. We employ two OCC-based score functions (one-class SVM with sigmoid and polynomial kernels, denoted by OCC-SVM.S and OCC-SVM.P, respectively) and two BIC-based score functions (k-nearest neighbors and multi-layer perceptron, denoted by BIC-KNN and BIC-MLP, respectively) to implement the SCQ procedure.

\begin{figure}[h]
	\centering
	\begin{minipage}{0.8\textwidth}
		\centering
		\includegraphics[width=0.85\textwidth]{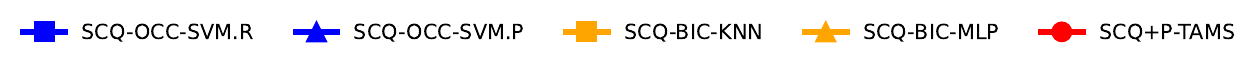} 
	\end{minipage}
	
	\begin{minipage}{\textwidth}
		\centering
		\includegraphics[width=0.85\textwidth]{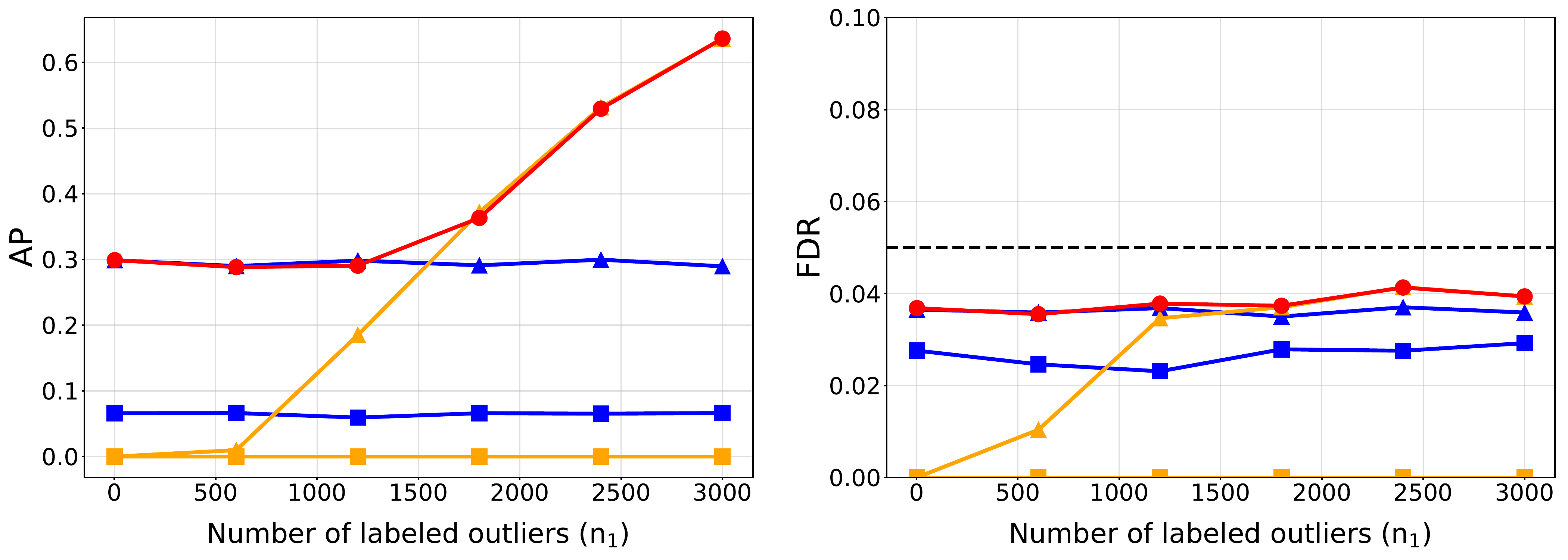} 
	\end{minipage}
	\caption{\small Comparison of AP and FDR for SCQ variants (implemented with two OCC-based and two BIC-based score functions) and P-TAMS, which uniformly enhances SCQ. }
	\label{fig:ams_ob}
\end{figure}

Several trends can be observed in the results presented in Figure \ref{fig:ams_ob}. First, the SCQ-OCC variants, which rely solely on inlier samples, offer enhanced stability but fail to leverage the information provided by labeled outliers. Second, the SCQ-BIC variants incorporate labeled outliers, but may become unstable and underperform the SCQ-OCC variants when the data are highly imbalanced. Finally, SCQ+P-TAMS adaptively switches between the OCC and BIC approaches to maximize detection power, thereby providing an effective solution to the model selection dilemma. Across all settings, P-TAMS consistently identifies the best-performing SCQ variant, leading to uniformly improved performance compared to either the OCC or BIC variants of SCQ.

\section{Experiments with real data}

\subsection{Analysis of the cybersecurity data}\label{sec:cicids}

This section evaluates the performance of the proposed SCQ procedure on the CICIDS2017 dataset \citep{sharafaldin2018toward}. The dataset, provided by the Canadian Institute for Cybersecurity (CIC), consists of normalized and deduplicated network traffic data with 80 features, encompassing both benign traffic and multiple attack categories. From this data, we construct a test set $\mathcal{D}^{\rm test}$ of size $m$, consisting of $m_0$ benign (inlier) samples and $m_1 = 1000$ attack (outlier) samples.
The latter includes three distinct attack types: distributed denial-of-service, port scan, and botnet. We consider settings where only labeled inliers are available. A null set \(\mathcal{D}_0\) of 30,000 samples is constructed and partitioned into three parts: a mirror dataset of size \(m\), while the remaining samples are divided equally into a training set and a calibration set.

Our objective is to efficiently identify network intrusions (outliers) by leveraging available side information to improve detection power and resource allocation. Specifically, we use the timestamp of each sample as the side information variable \(S_j\). This variable naturally defines 9 groups (corresponding to 9 working hours from 9 AM to 5 PM) and informs the construction of structure-aware weights (cf. Appendix \ref{sec:weight_construction} of the Supplement). As illustrated in the right column of Figure \ref{fig:real_data}, attack events are highly non-uniform over time, with pronounced peaks during specific hours. This temporal pattern provides valuable contextual information about the likelihood of anomalous activity. We demonstrate that this informative pattern can be exploited by SCQ to prioritize high-risk periods for anomaly detection, thereby enhancing the overall effectiveness of OOD testing. 
\begin{figure}[t]
	\centering

	\begin{minipage}{0.65\textwidth}   
		\centering
		\includegraphics[width=0.8\linewidth]{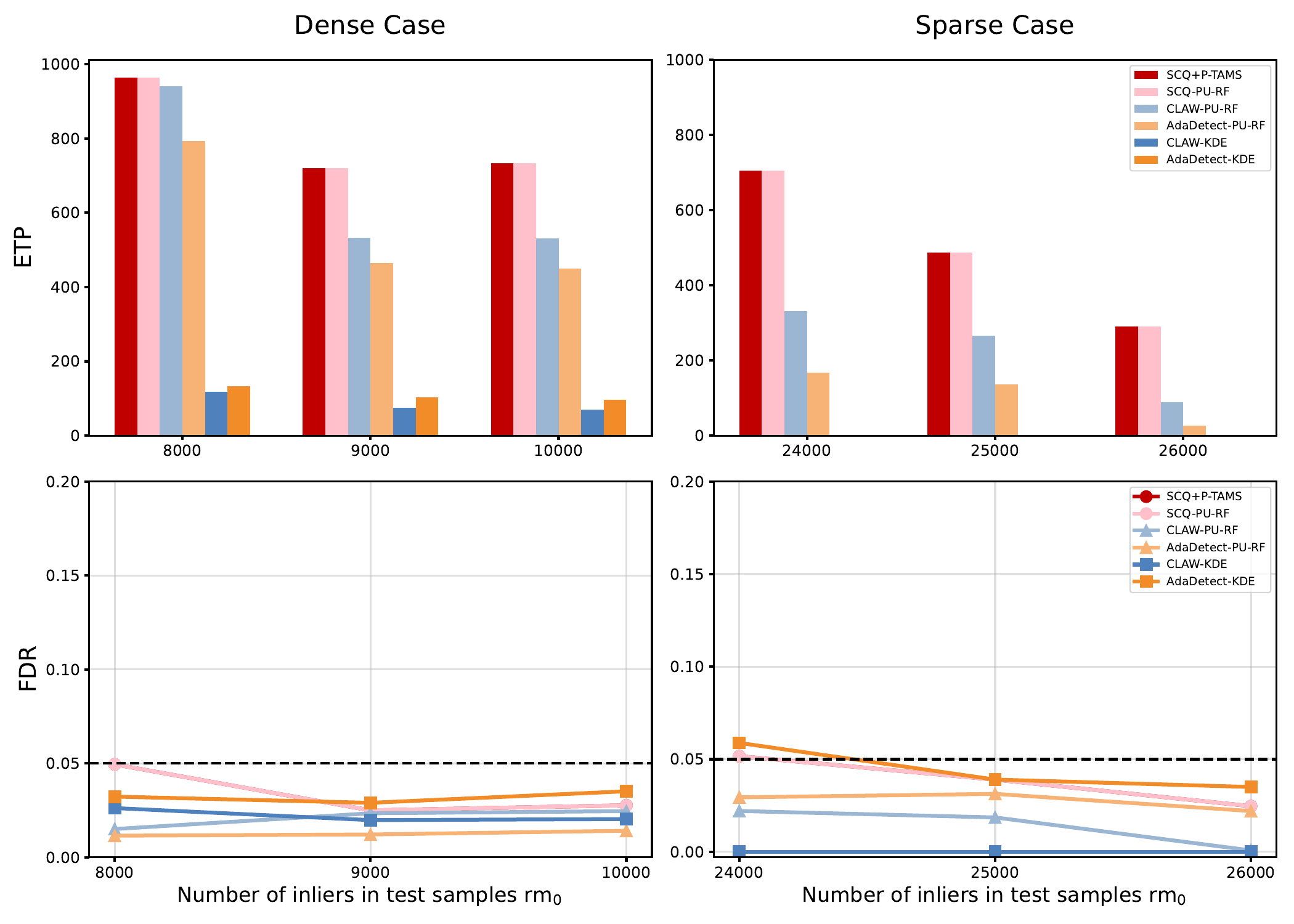}
	\end{minipage}
	\hspace{0.5em}                     
	\begin{minipage}{0.3\textwidth}   
		\centering
		\includegraphics[width=0.8\linewidth]{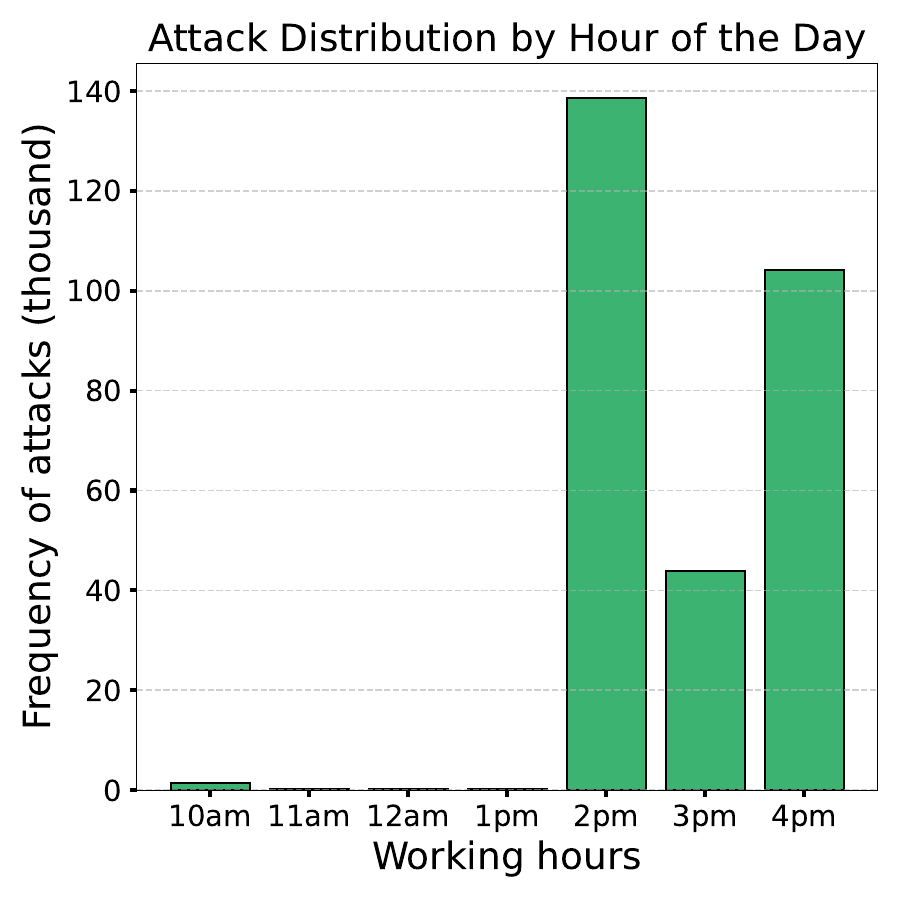}

		\vspace{2mm}
	
		\begingroup
		\scriptsize             
		\setstretch{0.9}       
		\RaggedRight            
		
		We partition data into 9 groups using side information (working hours),
		which captures informative temporal patterns in the data (no attack records at 9AM or 5PM). 
		SCQ exploits this structure to prioritize high-risk periods
		for effective anomaly detection.
		
		\endgroup
	\end{minipage}
	\caption{\small Left two columns: FDR and ETP of different methods; right column: hourly distribution of attack samples in the test set. }
	\label{fig:real_data}
\end{figure}

We compare the performance of different methods, including AdaDetect-KDE, AdaDetect-PU-RF, CLAW-KDE, CLAW-PU-RF, SCQ-PU-RF, and SCQ+P-TAMS, for OOD testing at a target FDR level of 0.05. SCQ+P-TAMS selects the optimal classifier from a toolbox comprising SCQ-PU-RF, SCQ-KDE, and SCQ-OCC-SVM. The number of test outliers $m_1$ is fixed at 1000, while the number of test inliers \(m_0\) is varied to create different sparsity settings. Each method is applied to 200 independent datasets with $\mathcal{D}^{\rm test}$ and $\mathcal{D}_0$ being randomly sampled from the original dataset. The FDR and the expected number of true positives (ETP) are computed by averaging results over these 200 replications. Results are presented for both dense setting (with \(m_0\) between 8000 and 10000) and sparse setting (with \(m_0\) between 24000 and 26000), shown in the left and right panels of Figure~\ref{fig:real_data}, respectively.

The comparisons reveal several important patterns: (1) leveraging side information (temporal patterns) significantly improves performance of OOD testing (SCQ‑PU-RF $>$ AdaDetect‑PU-RF), with the efficiency gain being especially pronounced in sparse settings; (2) PU methods outperform KDE methods in high‑dimensional settings; and (3) P‑TAMS consistently selects the SCQ‑PU‑RF variant, and the resulting SCQ+P‑TAMS method achieves the best overall performance.

\subsection{Analysis of the PageBlocks data}\label{sec:P-TAMS_eff_realdata}

We analyze the PageBlocks dataset \citep{campos2016evaluation}, treating text blocks
as inliers and non-text blocks as outliers. This experiment illustrates a setting in
which the best base classifier may depend on the number of labeled outliers and on
the realized sample. We apply SCQ with four candidate classifiers: LOF and GMM as
OCCs, and KNN and MLP as BICs. P-TAMS is then used to select among the resulting
SCQ variants, yielding SCQ+P-TAMS. The test data are divided into three groups with
different outlier proportions, and the group label is used as side information.
Details of the dataset, splitting scheme, and evaluation protocol are deferred to
Appendix~\ref{sec:app_pageblocks}.

\begin{figure}[h]
	\centering
	\includegraphics[width=0.8\textwidth]{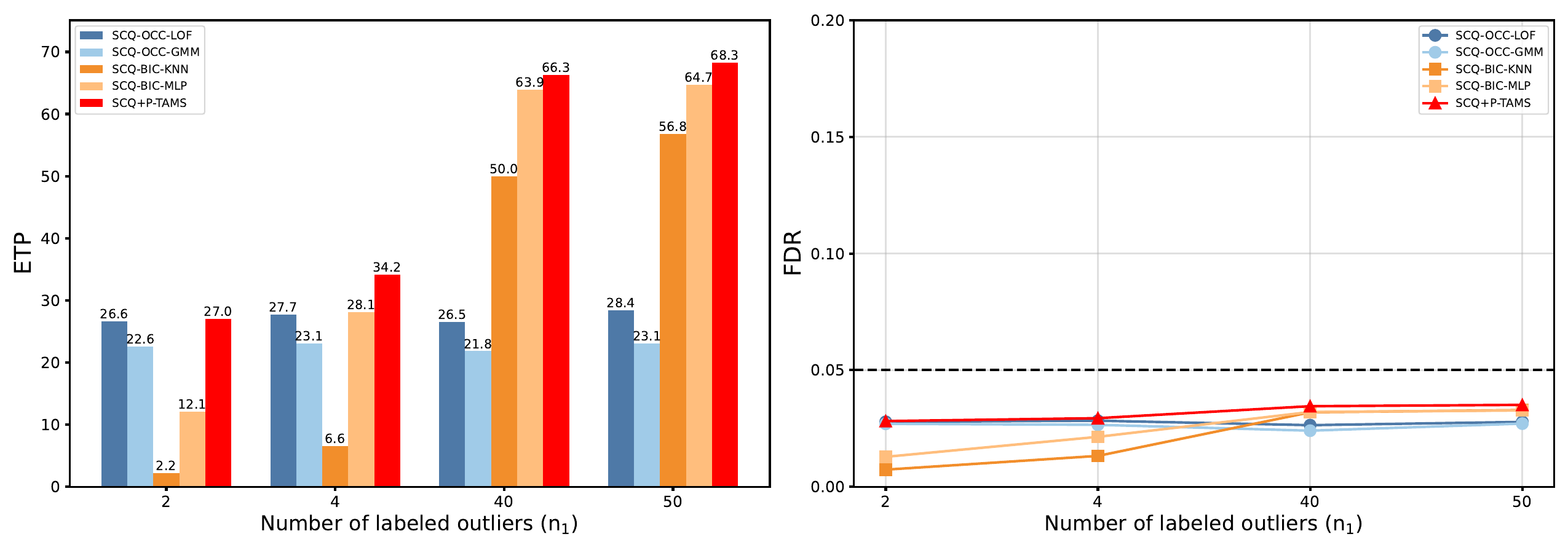}
	\caption{\small FDR and ETP of four SCQ variants and SCQ+P-TAMS on the PageBlocks data.}
	\label{fig:ams_ob_real}
\end{figure}

Figure~\ref{fig:ams_ob_real} shows three main patterns. 
First, OCC methods
perform better when a small number of labeled outliers is available, whereas
BIC methods become more effective as $n_1$ increases. 
Second, no single pre-specified classifier is uniformly optimal across all settings.
Third, SCQ+P-TAMS achieves the best overall performance by adaptively selecting among
the candidate classifiers. In particular, the selected classifier may vary across
repetitions even for the same \(n_1\), indicating that P-TAMS can exploit
dataset-specific information rather than relying on a fixed classifier.

\section{Data availability}

Code for reproducing simulation, real-data, and supplementary results is available at
\url{https://github.com/srysxr/SCQ_paper}; access details for the publicly available
CICIDS2017 and PageBlocks datasets are provided in the repository.

\begingroup
\small
\singlespacing
\setlength{\bibsep}{0pt plus 0.3ex}

\bibliographystyle{abbrvnat}
\bibliography{reference}

\endgroup

\newpage
\appendix

\begin{center}
	{\Large\bf Online Supplementary Material for \\ \vskip 0.5cm \LARGE ``Structure-Adaptive Conformal Inference for Large-Scale Out-of-Distribution Testing''}
\end{center}

This supplement contains the proofs for primary theory (Section \ref{sec:supple_proofs}), the supplementary methodological details (Section \ref{sec:supple_methdos}), the conceptual and theoretical comparisons (Section \ref{sec:supple_comparisons}), and additional numerical results (Section \ref{sec:supple_numerical_results}).

\setcounter{equation}{0}
\renewcommand{\theequation}{A.\arabic{equation}}
\setcounter{figure}{0}
\setcounter{page}{1}
\renewcommand{\thefigure}{A.\arabic{figure}}
\renewcommand{\thealgorithm}{A.\arabic{algorithm}}

\section{Technical Proofs}\label{sec:supple_proofs}
\subsection{Proof of Proposition \ref{prop:pairwise_scores}}

\begin{proof}
	Denote the conformity score function utilized in the SCQ procedure as $h$, in the sense that we set
	$h(x,S_j)\coloneqq\dfrac{p(x)}{w(S_j)}$ for any $j\in\mathcal{D}^{\rm test}$. In light of Remark \ref{rmk:weaker-condi}, we will verify the pairwise exchangeability of the conformity scores $\{(V_j\equiv h(X_j,S_j),\tilde{V}_j\equiv h(\tilde{X}_j,S_j)):j\in\mathcal{D}^{\rm test}\}$ under the weaker data exchangeability \eqref{eq:data_exch_pair} and the weaker condition \eqref{eq:g_func_swap_inva} for $s(\cdot)$ to construct conformal p-values. We begin by showing that, for any $j\in\mathcal{D}^{\rm test}$, the function $h(x,S_j)$ satisfies the following swap-invariance property:
	\begin{equation}\label{eq:swap-inva-h}
		h\left(x,S_j;(\mathbf{X},\tilde{\mathbf{X}})_{\mathrm{swap}(\mathcal{J})},\mathbf{X}_0^{\text{tr}}, \mathbf{X}_0^{\text{cal}}, \mathbf{X}_1,\mathbf{S}\right)=h\left(x,S_j;(\mathbf{X},\tilde{\mathbf{X}}),\mathbf{X}_0^{\text{tr}}, \mathbf{X}_0^{\text{cal}}, \mathbf{X}_1, \mathbf{S}\right),\quad\forall\mathcal{J}\subset\mathcal{D}^{\rm test}\:.
	\end{equation}
	Firstly, the conformal p-value function $p(x)$ defined in \eqref{eq:conformal_p} is invariant under any swap of $(\mathbf{X},\tilde{\mathbf{X}})$, as $p(x)$ is constructed based on $\left\{s(X_i):i\in\mathcal{D}_0^{\rm cal}\right\}$ with $s(\cdot)$ satisfying the swap-invariance property \eqref{eq:g_func_swap_inva}. Secondly, the weight function $w(\cdot)$ is also required to satisfy \eqref{eq:g_func_swap_inva}. Therefore, condition \eqref{eq:swap-inva-h} naturally holds for the score function $h(x,S_j)\coloneqq\dfrac{p(x)}{w(S_j)}$.

	Next, we will show that, under the data exchangeability condition \eqref{eq:data_exch_pair}, if the swap-invariance \eqref{eq:swap-inva-h} holds for $h(x,S_j)$, then conformity scores $\{(V_j\equiv h(X_j,S_j),\tilde{V}_j\equiv h(\tilde{X}_j,S_j)):j\in\mathcal{D}^{\rm test}\}$  satisfy the desired pairwise exchangeability \eqref{eq:pwex-scores}. 
	For $j\in\mathcal{H}_0$,
	denote $\mathbf{X}_{-j}=\{X_i:i\in\mathcal{D}^{\rm test}\setminus\{n+j\}\}$ and $\tilde{\mathbf{X}}_{-j}=\{\tilde{X}_i:i\in\mathcal{D}^{\rm test}\setminus\{n+j\}\}$, let 
	\begin{equation*}
		\mathcal{G}_j
		:=
		\Big(
		\mathcal{V}_{-j},
		\, \tilde{\mathcal{V}}_{-j},
		\, \{X_j,\tilde{X}_j\}
		\Big),
	\end{equation*}
	where $\{X_j,\tilde{X}_j\}$ is the unordered set of the elements in $(X_j,\tilde{X}_j)$. By construction, $\mathcal{G}_j$ is symmetric with respect to $(X_j,\tilde X_j)$ given $(\mathbf{X}_{-j},\tilde{\mathbf{X}}_{-j},\mathbf{X}_0^{\text{tr}}, \mathbf{X}_0^{\text{cal}}, \mathbf{X}_1,\mathbf{S})$.
	Let $\psi(x,y)$ be a vector-valued symmetric function satisfying
	$\psi(x,y)=\psi(y,x)$. Consider two random elements $X$ and $Y$ that are
	pairwise exchangeable, i.e.,
	$
	(X,Y)\stackrel{d}{=}(Y,X).
	$
	Then we have
	\begin{equation}\label{eq:proof-claim}
		(X,Y,\psi(X,Y))
		\stackrel{d}{=}
		(Y,X,\psi(Y,X))
		=
		(Y,X,\psi(X,Y)).
	\end{equation}
	Applying \eqref{eq:proof-claim}, under the data exchangeability condition \eqref{eq:data_exch_pair}, we obtain
	\begin{equation}\label{eq:proof-data-exch}
		(X_j,\tilde X_j\mid\mathcal{G}_j,\mathbf{X}_{-j},\tilde{\mathbf{X}}_{-j},\mathbf{X}_0^{\text{tr}}, \mathbf{X}_0^{\text{cal}}, \mathbf{X}_1,\mathbf{S})
		\;\stackrel{d}{=}\;
		(\tilde X_j,X_j\mid\mathcal{G}_j,\mathbf{X}_{-j},\tilde{\mathbf{X}}_{-j},\mathbf{X}_0^{\text{tr}}, \mathbf{X}_0^{\text{cal}}, \mathbf{X}_1,\mathbf{S}).
	\end{equation}
	Since $h(x,S_j\mid \{X_j,\tilde{X}_j\},\mathbf{X}_{-j},\tilde{\mathbf{X}}_{-j},\mathbf{X}_0^{\text{tr}}, \mathbf{X}_0^{\text{cal}}, \mathbf{X}_1,\mathbf{S})$ is non-random with respect to\\
	\noindent $\sigma(\{X_j,\tilde{X}_j\},\mathbf{X}_{-j},\tilde{\mathbf{X}}_{-j}, \mathbf{X}_0^{\text{tr}}, \mathbf{X}_0^{\text{cal}}, \mathbf{X}_1,\mathbf{S})\subset\sigma(\mathcal{G}_j,\mathbf{X}_{-j},\tilde{\mathbf{X}}_{-j},\mathbf{X}_0^{\text{tr}}, \mathbf{X}_0^{\text{cal}}, \mathbf{X}_1,\mathbf{S})$, it follows from \eqref{eq:proof-data-exch} that
	\begin{equation*}
		(V_j,\tilde V_j\mid\mathcal{G}_j,\mathbf{X}_{-j},\tilde{\mathbf{X}}_{-j},\mathbf{X}_0^{\text{tr}}, \mathbf{X}_0^{\text{cal}}, \mathbf{X}_1,\mathbf{S})
		\;\stackrel{d}{=}\;
		(\tilde V_j,V_j\mid\mathcal{G}_j,\mathbf{X}_{-j},\tilde{\mathbf{X}}_{-j},\mathbf{X}_0^{\text{tr}}, \mathbf{X}_0^{\text{cal}}, \mathbf{X}_1,\mathbf{S}).
	\end{equation*}
	By integrating $(\{X_j,\tilde{X}_j\},\mathbf{X}_{-j},\tilde{\mathbf{X}}_{-j},\mathbf{X}_0^{\text{tr}}, \mathbf{X}_0^{\text{cal}}, \mathbf{X}_1,\mathbf{S})$ out, we have that,  for each $j\in\mathcal{H}_0$,
	$$(V_j,\tilde{V}_j|\mathcal{V}_{-j},\tilde{\mathcal{V}}_{-j}) \overset{d}{=} (\tilde{V}_j,V_j|\mathcal{V}_{-j},\tilde{\mathcal{V}}_{-j}),$$
	which also implies the pairwise exchangeability \eqref{eq:pwex-scores}.
\end{proof}

\subsection{Proof of Theorem \ref{thm:scq_validity}}\label{proof:SCQ_validity}
\begin{proof}
	For conformity scores $\{(V_j, \tilde{V}_j): j \in \mathcal{D}^{\text{test}}\}$, we begin by establishing the equivalence between the SCQ procedure and the e-BH procedure applied with e-values defined as
	\begin{align}\label{eq:e-values}
		e_j=\frac{m\mathbb{I}(V_j\leq\tau,V_j<\tilde{V}_j)}{1+\sum_{i\in\mathcal{D}^{\rm test}}\mathbb{I}(\tilde{V}_i\leq\tau,\tilde{V}_i<V_i)},\quad j\in\mathcal{D}^{\rm test},
	\end{align}
	where $\tau$ is given in \eqref{eq:tau}. Let $\mathcal{R}_{ebh}$ denote the rejection set of the e-BH procedure based on these e-values, i.e., we set ${\mathcal{R} }_{ebh} = \{ j\in\mathcal{D}^{\rm test} : e_{j}\geq e_{( \hat{k} ) }\},$ where $e_{(1)}\geq e_{(2)}\geq\cdots\geq e_{(m)}$ are the order statistics and the threshold $\hat{k}=\operatorname*{max}\{i:\frac{ie_{(i)}}{m}\geq\frac{1}{\alpha}\}$. Let $\mathcal{R}_{scq}\coloneqq\{j\in\mathcal{D}^{\rm test}:q_j\leq\alpha\}$, where $(q_j:j\in\mathcal{D}^{\rm test})$ are conformal q-values derived by \eqref{eq:scq}. By Proposition \ref{prop:equiv_SCQ_BC}, which will be proved later, to prove $\mathcal{R}_{scq}=\mathcal{R}_{ebh}$, it is enough to show that $\mathcal{R}_{ebh}$ coincides with the rejection set of BC algorithm, $\mathcal{R}_{bc}$, which is defined in \eqref{dec:BC}.
	For simplicity, let $\delta_j^{bc}=\mathbb{I}(V_j\leq\tau,V_j<\tilde{V}_j)$, $\delta_j^{ebh}=\mathbb{I}(e_{j}\geq e_{( \hat{k} ) })$, then we have
	\begin{equation*}
		e_j = 
		\left\{
		\begin{aligned}
			&0, &\delta_j^{bc} = 0; \\
			&\dfrac{m}{1 + \sum_{i\in\mathcal{D}^{\rm test}} \mathbb{I}(\tilde{V}_i\leq\tau,\tilde{V}_i<V_i)}, &\delta_j^{bc} = 1.
		\end{aligned}
		\right.
	\end{equation*}
	If $\delta_j^{bc}=1$, then by the definition of $\hat{k}$, we have $e_{(\hat{k})} = \dfrac{m}{1 + \sum_{i\in\mathcal{D}^{\rm test}} \mathbb{I}(\tilde{V}_i\leq\tau,\tilde{V}_i<V_i)}$ and hence $\delta_j^{ebh}=1$. If $\delta_j^{ebh}=1$, then $e_j\geq e_{(\hat{k})}$ and we have $\delta_j^{bc}=1$. Otherwise, we will have $e_{(\hat{k})}=0$, and this contradicts with $\frac{\hat{k}e_{(\hat{k})}}{m}\geq\frac{1}{\alpha}$ by definition. Thus, the three rejection sets coincide: $\mathcal{R}_{ebh}=\mathcal{R}_{bc}=\mathcal{R}_{scq}$.
	
	We establish the FDR control of the SCQ procedure by proving the validity of $\mathcal{R}_{ebh}$. By the e-BH theory of \cite{wang2022false}, this validity holds provided that the e-values defined in \eqref{eq:e-values} are generalized e-values, that is, $\mathbb{E}\left[\sum_{j\in\mathcal{H}_0}e_j\right]\leq m.$
	To verify this, we introduce the following lemma, which follows from the argument first established in \cite{ZhaSun25-PLIS} and is stated here in our notation without proof.
	\begin{lemma}$[\text{\cite{ZhaSun25-PLIS}}]$\label{lemma:CLAW2}
		Suppose that we have conformity scores $(V_j:j\in\mathcal{D}^{\rm test})$ and $(\tilde{V}_j:j\in\mathcal{D}^{\rm test})$ satisfying the pairwise exchangeability \eqref{eq:pwex-scores}. Let $\tau$ be the threshold output by the BC algorithm, which is given in \eqref{eq:tau}. Then,
		\begin{equation*}
			\mathbb{E}\left[\frac{\sum_{j\in\mathcal{H}_{0}}\mathbb{I}(V_j\leq\tau,V_j<\tilde{V}_j)}{1+\sum_{j\in\mathcal{H}_{0}}\mathbb{I}(\tilde{V}_j\leq\tau,\tilde{V}_j<V_j)}\right]=\mathbb{E}\left[\frac{\sum_{j\in\mathcal{H}_{0}}\mathbb{I}(V_j<\tilde{V}_j)\mathbb{I}(V_j\leq\tau)}{1+\sum_{j\in\mathcal{H}_{0}}\mathbb{I}(\tilde{V}_j<V_j)\mathbb{I}(\tilde{V}_j\leq\tau)}\right]\leq1.
		\end{equation*}
	\end{lemma}
	By Proposition \ref{prop:pairwise_scores}, the pairwise exchangeability of SCQ scores $\{(V_j, \tilde{V}_j): j \in \mathcal{D}^{\text{test}}\}$ holds under data exchangeability \eqref{eq:data_exch_joint}. Consequently, Lemma \ref{lemma:CLAW2} implies that
{\small	\begin{equation*}
		\begin{aligned}
			& \mathbb{E}\left[\sum_{j\in\mathcal{H}_0}e_j\right]=m\mathbb{E}\left[\sum_{j\in\mathcal{H}_0}\frac{\mathbb{I}(V_j\leq\tau,V_j<\tilde{V}_j)}{1+\sum_{i\in\mathcal{D}^{\rm test}}\mathbb{I}(\tilde{V}_i\leq\tau,\tilde{V}_i<V_i)}\right]\leq \mathbb{E}\left[\frac{\sum_{j\in\mathcal{H}_{0}}\mathbb{I}(V_j<\tilde{V}_j)\mathbb{I}(V_j\leq\tau)}{1+\sum_{j\in\mathcal{H}_{0}}\mathbb{I}(\tilde{V}_j<V_j)\mathbb{I}(\tilde{V}_j\leq\tau)}\right]\leq m,
		\end{aligned}
	\end{equation*}}
	which confirms the validity of the e-values defined in \eqref{eq:e-values} and consequently establishes the finite-sample FDR control of the SCQ procedure.
	\begin{remark}\label{cp:validity}
		The conformal p-values, constructed using a continuous score function $s(\cdot)$ that satisfies \eqref{eq:s_function_permu}, are super-uniform under the null when conditions \eqref{eq:s_function_permu} and \eqref{eq:data_exch_joint} hold. To see this, observe that \eqref{eq:data_exch_joint}  satisfies the exchangeability requirement of Assumption 1 in \cite{marandon2024adaptive}, then the result follows from Theorem 3.3 in \cite{marandon2024adaptive}.
	\end{remark}
\end{proof}

\subsection{Proof of Theorem \ref{thm:P-TAMS-validity}}
\begin{proof}
	Denote the conformity scores computed by the SCQ procedure with the classifier $\mathcal{C}_*$ selected by P-TAMS algorithm as $\mathcal{V}^*=\left\{V_j^*:j\in\mathcal{D}^{\rm test}\right\}$ and $\mathcal{\widetilde{V}}^*=\left\{\tilde{V}_j^*:j\in\mathcal{D}^{\rm test}\right\}$.
	According to the proof of Theorem \ref{thm:scq_validity}, to prove the validity of the P-TAMS algorithm, it suffices to show that, under the data exchangeability (\ref{eq:data_exch_joint}), $\mathcal{V}^*$ and $\mathcal{\widetilde{V}}^*$ are pairwise exchangeable:
	\begin{equation*}
		(V_i^*,\tilde{V}_i^*,\mathcal{V}_{-i}^*,\widetilde{\mathcal{V}}_{-i}^*)\overset{d}{=}(\tilde{V}_i^*,V_i^*,\mathcal{V}_{-i}^*,\widetilde{\mathcal{V}}_{-i}^*),\quad \forall i\in\mathcal{H}_0,
	\end{equation*}
	where $\mathcal{V}_{-i}^*=\left\{V_j^*:j \in \mathcal{D}^{\rm test}\setminus\{n+i\}\right\}$ and $\widetilde{\mathcal{V}}_{-i}^*=\left\{\tilde{V}_j^*:j \in \mathcal{D}^{\rm test}\setminus\{n+i\}\right\}$.
	Denote the score function of the P-TAMS algorithm as $h^*$, in the sense that we set $V_j^*= h^*(X_j,S_j)$ and $\tilde{V}_j^*= h^*(\tilde{X}_j,S_j)$ for any $j\in\mathcal{D}^{\rm test}$. 
	By the arguments in the proof of Proposition \ref{prop:pairwise_scores}, it then remains to verify the following swap-invariance property of the score function $h^*$:
	\begin{equation*}\label{g* function invariant}
		h^*\left(\cdot,S_i;(\mathbf{X},\tilde{\mathbf{X}})_{\mathrm{swap}(\mathcal{J})},\mathbf{X}_0^{\text{tr}}, \mathbf{X}_0^{\text{cal}}, \mathbf{X}_1, \mathbf{S}\right)=h^*\left(\cdot,S_i;(\mathbf{X},\tilde{\mathbf{X}}),\mathbf{X}_0^{\text{tr}}, \mathbf{X}_0^{\text{cal}}, \mathbf{X}_1, \mathbf{S}\right),\quad\forall\mathcal{J}\subset\mathcal{D}^{\rm test}\:.
	\end{equation*} 
	
	According to Algorithm \ref{alg:P-TAMS}, the score function $h^*$ can be seen as a composite function: we have $h^*(\cdot)=l_1[l_2(\cdot)]$, where $l_2:[K]\mapsto[K]$ denotes the selection function that chooses the best performing one from  $K$ candidate models $(\mathcal{C}_1,\dots,\mathcal{C}_K)$, and $l_1$ represents the score function that calculates the conformity scores of the SCQ procedure with the selected model $\mathcal{C}_*$. Since the proof of Proposition \ref{prop:pairwise_scores} shows that the score function $l_2$ is invariant under any swap of $(\mathbf{X},\tilde{\mathbf{X}})$ if we fix the selected classifier $\mathcal{C}_*$, it remains to show that the selection function $l_1$, which derives $\mathcal{C}_*$, is swap-invariant w.r.t. $(\mathbf{X},\tilde{\mathbf{X}})$.
	
	To prove this, recall the three summarized steps of P-TAMS in Section \ref{subsec:P-TAMS}. First, for each $k\in[K]$, the initial partition of test samples in Step 1 is based on a preliminary rejection set $\bar{\mathcal{R}}_k$, which is obtained by applying BH on pseudo conformal p-values $\bar{\mathcal{P}}_k = \left\{\min\left(p_i^k, \tilde{p}_i^k\right) : i \in \mathcal{D}^{\rm test}\right\}$ at a pre-specified level. As a result, for any $k\in[K]$, $\bar{\mathcal{R}}_k$ is swap-invariant w.r.t. $(\mathbf{X},\tilde{\mathbf{X}})$, since the value of $\bar{\mathcal{P}}_k$ is invariant under any swap of $(\mathbf{X},\tilde{\mathbf{X}})$. Then, in Step 2, we calculate the pseudo scores \(\mathcal{U}^k = \{U_i^k: i \in \mathcal D^{\text{test}}\}\) and \(\widetilde{\mathcal{U}}^k = \{\tilde{U}_i^k: i \in\mathcal D^{\text{test}}\}\) based on the true scores \(\mathcal{V}_k = \{V_i^k: i \in \mathcal D^{\text{test}}\}\) and \(\widetilde{\mathcal{V}}_k = \{\tilde{V}_i^k: i \in \mathcal D^{\text{test}}\}\), through the min-max operator \eqref{eq:pseudo_outlier} and the coin-flipping operator \eqref{eq:pseudo_inlier}. Since the true scores $V_i^k$ and $\tilde{V}_i^k$ are swap-invariant as proved earlier, the carefully designed operators \eqref{eq:pseudo_outlier} and \eqref{eq:pseudo_inlier} together ensure the swap-invariance of pseudo scores $U_i^k$ and $\tilde{U}_i^k$. Finally, in Step 3, we select the model $C_*$, which is also consequently swap-invariant, by maximizing the rejection number calculated by applying SCQ with the pseudo scores \(\mathcal{U}^k\) and \(\widetilde{\mathcal{U}}^k\).
	Thus, the swap-invariance of the score function $h^*$ utilized by the P-TAMS algorithm is verified, and so is the validity of the P-TAMS algorithm itself.
\end{proof}

\subsection{Proof of Proposition \ref{prop:equiv_SCQ_BC}}
\begin{proof}
	Given $m$ pairs of conformity scores $\{(V_j,\tilde{V}_j):j\in\mathcal{D}^{\rm test}\}$, recall the definitions of $\mathcal{R}_{bc}$ and $\mathcal{R}_{scq}$: $\mathcal{R}_{bc} = \{j \in \mathcal{D}^{\text{test}}: V_j \leq \tau, \, V_j < \tilde{V}_j\}$ and $\mathcal{R}_{scq} = \{j \in \mathcal{D}^{\text{test}}: q_j \leq \alpha\}$, where $(q_j:j\in\mathcal{D}^{\rm test})$ are conformal q-values defined in \eqref{eq:scq}.
	
	Since the target FDR level $\alpha$ should be less than one, for any $j\in\mathcal{R}_{scq}$, we know that $q_j\leq\alpha$ and $V_j<\tilde{V}_j$. If $V_j>\tau$, then, by the definition of $q_j$, we have $t_0\geq V_j>\tau$, such that $q_j=H(t_0)\leq\alpha$, which is in contradiction with the definition of $\tau$. Therefore, we have $V_j\leq\tau$ and consequently $j\in\mathcal{R}_{bc}$. Next, for any $j\in\mathcal{R}_{bc}$, we have $V_j<\tilde{V}_j$ and $V_j\leq\tau$. Then, $q_j=\min\limits_{t\in\mathcal{V}\cup\widetilde{\mathcal{V}},t\geq V_i}H(t)\leq H(\tau)\leq\alpha$ and hence $j\in\mathcal{R}_{scq}$. As a result, we have $\mathcal{R}_{scq}=\mathcal{R}_{bc}$. 
\end{proof}

\subsection{Proof of Theorem \ref{thm:asymptotic-FDR}}\label{proof:BC FDR analysis}
\begin{proof}
	By Proposition \ref{prop:equiv_SCQ_BC}, we have $\mathcal{R}_{bc}=\mathcal{R}_{scq}$, it then suffices to show $\lim_{m \to \infty}  \mathbb{E}\left[ \dfrac{|\mathcal{R}_{bc} \cap \mathcal{H}_0|}{|\mathcal{R}_{bc}|\vee 1} \right] = \alpha$, where $\mathcal{R}_{bc}$ is defined in \eqref{dec:BC}. 
	To this end, we begin by introducing some notations. First, for simplicity, define a function $\bar{Q}(t)$ that is related to the marginal FDR of the rejection set $\mathcal{R}(t)=\{j\in\mathcal{D}^{\rm test}:V_j\leq t,V_j<\tilde{V}_j\}$ for the given threshold $t$, as  $\bar{Q}(t)\coloneqq\sum_{j\in\mathcal{D}^{\rm test}} \dfrac{1}{m} \mathbb{P}(V_j \leq t, V_j < \tilde{V}_j,  Y_j = 0) - \dfrac{\alpha}{m}\sum_{j\in\mathcal{D}^{\rm test}} \mathbb{P}(V_j \leq t, V_j < \tilde{V}_j)$. 
	Denote the corresponding oracle threshold $t^*$ calculated from $\bar{Q}(t)$ by
	\begin{equation}\label{eq:tstar}
		t^*\coloneqq\sup\left\{ t \in (0,1): \bar{Q}(t) \leq 0 \right\}
	\end{equation}
	Next, in the same way, we define a function $\hat{Q}(t)$ which is related to the mirror process $H(t)$ in \eqref{eq:H_function} as $\hat{Q}(t)\coloneqq\dfrac{1}{m}\sum_{j\in\mathcal{D}^{\rm test}}\mathbb{I}(\tilde{V}_j\leq t,\tilde{V}_j< V_j)-\dfrac{\alpha}{m}\sum_{j\in\mathcal{D}^{\rm test}}\mathbb{I}(V_j\leq t,V_j<\tilde{V}_j)$. 
	Moreover, let $\nu_i=\min\left\{V_i,\tilde{V}_i\right\}$ for each $i\in\mathcal{D}^{\rm test}$, and order them as $\nu_{(1)}\leq\nu_{(2)}\leq\dots\leq\nu_{(m)}$. 
	Focusing on the function $\hat{Q}(t)$, we further define its expectation function $G(t)$ as: $G(t)=\dfrac{1}{m}\sum_{j\in\mathcal{D}^{\rm test}}\mathbb{P}(\tilde{V}_j\leq t,\tilde{V}_j< V_j)-\dfrac{\alpha}{m}\sum_{j\in\mathcal{D}^{\rm test}}\mathbb{P}(V_j\leq t,V_j<\tilde{V}_j)$, and its continuous version $\hat{Q}^c(t)$ via interpolations: $\hat{Q}^c(t)=\dfrac{t-\nu_{(i)}}{\nu_{(i+1)}-\nu_{(i)}}\hat{Q}_{i+1}+\dfrac{\nu_{(i+1)}-t}{\nu_{(i+1)}-\nu_{(i)}}\hat{Q}_{i}$, for $t\in[\nu_{(i)},\nu_{(i+1)})$ and $i=1,\dots,m-1,$
	where $\hat{Q}_i\equiv\hat{Q}(\nu_{(i)})$. Based on $\hat{Q}^c(t)$, we define
	\begin{equation}\label{eq:hat-t}
		\hat{t}=\sup\left\{t\in(0,1):\hat{Q}^c(t)\leq0\right\}.
	\end{equation}
	
	To prove the theorem, we will first establish the following asymptotic equivalence of thresholds $\hat{t}$ and $t^*$:
	\begin{equation}\label{eq:hatt_to_tstar}
		\hat{t}-t^*\overset{p}{\to}0,\quad m\to\infty,
	\end{equation}
	in two steps.
	
	\textbf{Step (a): The asymptotic equivalence between $\hat{Q}^c(t)$ and $\bar{Q}(t)$.} 
	In this part, we show that, for any $t\in(0,1)$ we have
	\begin{equation}\label{eq:hatcQ_to_Q}
		\hat{Q}^c(t)-\bar{Q}(t)\overset{p}{\to}0,\quad m\to\infty.
	\end{equation}
	Under the weak dependence condition \eqref{eq:cov}, by the WLLN, 
	$\hat{Q}(t)-G(t)\xrightarrow{p}0$. To prove \eqref{eq:hatcQ_to_Q}, since $\mid\hat{Q}^c(t)-\hat{Q}(t)\mid\leq\dfrac{1}{m}$,  it remains to show $G(t)-\bar{Q}(t)\to0, m\to\infty$.
	Decomposing $\frac{1}{m}\sum_{j\in\mathcal{D}^{\rm test}}\mathbb{P}(\tilde{V}_j\leq t,\tilde{V}_j< V_j)$ by:
	\begin{footnotesize}
		\begin{equation*}\label{part0}
			\frac{1}{m}\sum_{j\in\mathcal{D}^{\rm test}}\mathbb{P}(\tilde{V}_j\leq t,\tilde{V}_j< V_j)
			=\frac{1}{m}\sum_{j\in\mathcal{D}^{\rm test}}\mathbb{P}(\tilde{V}_j\leq t,\tilde{V}_j< V_j,  Y_j=0) 
			+ \frac{1}{m}\sum_{j\in\mathcal{D}^{\rm test}}\mathbb{P}(\tilde{V}_j\leq t,\tilde{V}_j< V_j, Y_j=1).
		\end{equation*}
	\end{footnotesize}
	The first term of the RHS above equals $\frac{1}{m}\sum_{j\in\mathcal{D}^{\rm test}}\mathbb{P}(V_j\leq t,V_j<\tilde{V}_j,  Y_j=0)$ by the score exchangeability under the null, and the second term 
	\begin{equation*}
		\frac{1}{m}\sum_{j\in\mathcal{D}^{\rm test}}\mathbb{P}(\tilde{V}_j\leq t,\tilde{V}_j< V_j, Y_j=1)\leq \frac{1}{m}\sum_{j\in\mathcal{D}^{\rm test}}\mathbb{P}(\tilde{V}_j< V_j\mid  Y_j=1) 
	\end{equation*}
	vanishes under Assumption \ref{ass:sep}. Hence,
	\begin{equation*}
		\begin{aligned}
			G(t)-\bar{Q}(t)&=\frac{1}{m}\sum_{j\in\mathcal{D}^{\rm test}}\mathbb{P}(\tilde{V}_j\leq t,\tilde{V}_j< V_j)-\frac{1}{m}\sum_{j\in\mathcal{D}^{\rm test}}\mathbb{P}(V_j\leq t,V_j<\tilde{V}_j, Y_j=0)\\
			&=\dfrac{1}{m}\sum_{j\in\mathcal{D}^{\rm test}}\mathbb{P}(\tilde{V}_j\leq t,\tilde{V}_j< V_j, Y_j=1)
			\to0,
		\end{aligned}
	\end{equation*}
	finishing the proof of \eqref{eq:hatcQ_to_Q}.
	
	\textbf{Step (b): The asymptotic equivalence between $\hat{t}$ and $t^*$.}
	Since $t^*$ is nonrandom, \eqref{eq:hatcQ_to_Q} implies $\hat{Q}^c(t^*)-\bar{Q}(t^*)=\hat{Q}^c(t^*)\overset{p}{\to}0$.
	Thus we have
	\begin{equation}
		\label{eq:tstar_<_hatt}
		\mathbb{P}(t^*\leq\hat{t})\geq\mathbb{P}(\hat{Q}_c(t^*)=0)\to1,\quad m\to\infty.
	\end{equation}
	Note that the BC threshold $\tau$ \eqref{eq:tau} can be equally expressed as 
	\begin{equation*}
		\tau=\max\left\{t\in\{\nu_{(i)}\}_{i=1}^m:\hat{Q}(t)\leq0\right\}.
	\end{equation*}
	Let $\tau=\nu_{(k)}$.
	By the definition of $t^*$ and the validity of the BC rejection set \eqref{dec:BC} (implied by Theorem \ref{thm:scq_validity} and Proposition \ref{prop:equiv_SCQ_BC}), we know that $\mathbb{P}(\nu_{(k)}\leq t^*)\geq\mathbb{P}(\bar{Q}(\nu_{(k)})\leq0)\to1$. Combining it with the result in \eqref{eq:tstar_<_hatt} and the fact that $\nu_{(k)}\leq\hat{t}<\nu_{(k+1)}$ (following by the definition of $\hat{t}$ in \eqref{eq:hat-t}), we have $\mathbb{P}(\nu_{(k)}\leq t^*\leq\hat{t}<\nu_{(k+1)})\to1$. This result indicates that, there always exists an interval $[\nu_{k},\nu_{k+1})$, such that with probability going to one, both \(t^*\) and \(\hat{t}\) lie in the interval $[\nu_{k},\nu_{k+1})$.
	Since $\hat{Q}^c(t)$ is monotone and continuous on $[\nu_{(k)},\nu_{(k+1)}]$, its inverse function $\hat{Q}^{c-1}(t)$ is then continuous on $[\hat{Q}^c(\nu_{(k)}),\hat{Q}^c(\nu_{(k+1)})]$. Thus, for any $\epsilon>0$, there exists $\gamma>0$ such that, with probability going to one, we have $|\hat{t}-t^*|=|\hat{Q}^{c-1}(0)-\hat{Q}^{c-1}(\hat{Q}^c(t^*))|\leq\epsilon, \text{ if }|\hat{Q}^c(t^*)|<\gamma$. Since $\hat{Q}^c(t^*)\overset{p}{\to}0$ as shown before, we obtain that
	\begin{equation*}
		\mathbb{P}(|\hat{t}-t^*|>\epsilon)\leq\mathbb{P}(|\hat{Q}^c(t^*)|\geq\gamma)\to0,\quad m\to\infty.
	\end{equation*}
	Hence, letting $\epsilon\to0$, we have \eqref{eq:hatt_to_tstar} proved.
	
	In the remainder of this section, to prove Theorem \ref{thm:asymptotic-FDR}, we will equivalently prove the desired FDR property of the decision rule derived from BC algorithm, which is denoted as $\pmb{\delta^{\tau}}=\left\{\delta_j^{\tau}\equiv \mathbb{I}(V_j\leq\tau,V_j<\tilde{V}_j):j\in\mathcal{D}^{\rm test}\right\}$. Let $\pmb{\hat{\delta}}=\left\{\hat{\delta}_j\equiv \mathbb{I}(V_j\leq\hat{t},V_j<\tilde{V}_j):j\in\mathcal{D}^{\rm test}\right\}$ and $\pmb{\delta^*}=\left\{\delta_j^*\equiv \mathbb{I}(V_j\leq t^*,V_j<\tilde{V}_j):j\in\mathcal{D}^{\rm test}\right\}$\footnote{By the definition of $t^*$ in \eqref{eq:tstar} and $t^\alpha_{OR}(\pmb{w})$ in \eqref{eq:oracle_threshold}, we have that the two decision rules, respectively thresholding at $t^*$ and $t^\alpha_{OR}(\pmb{w})$, are equivalent.} be the decision rules of the mirror-based algorithms with thresholds $\hat{t}$ and $t^*$ defined in \eqref{eq:hat-t} and \eqref{eq:tstar}, respectively.
	Define also the marginal FDR (mFDR) of decision rule $\pmb{\delta}=\{\delta_j:j\in\mathcal{D}^{\rm test}\}$ as $\text{mFDR}=\mathbb{E}\{\sum_{i\in\mathcal{D}^{\rm test}}(1- Y_i)\delta_i\}/\mathbb{E}(\sum_{i\in\mathcal{D}^{\rm test}}\delta_i)$. By Proposition 7 in \cite{tony2019covariate}, under Assumption \ref{ass:ER=Om} and the weak dependence condition \eqref{eq:cov} characterized in Lemma \ref{lem:cov}, we have $\text{FDR}_{\pmb{\delta^*}}=\text{mFDR}_{\pmb{\delta^*}}+o(1)=\alpha+o(1)$, where we use the notations $\text{FDR}_{\pmb{\delta}}$ and $\text{mFDR}_{\pmb{\delta}}$ to denote the FDR and mFDR of a procedure with decision rule $\pmb{\delta}$, respectively. Note that $\pmb{\hat{\delta}}$ coincides with $\pmb{\delta^{\tau}}$. As a result, it remains to show that $\text{FDR}_{\pmb{\hat{\delta}}}$ is asymptotically equal to $\text{FDR}_{\pmb{\delta^*}}$, and this will be proved by leveraging the asymptotic equivalence between thresholds $\hat{t}$ and $t^*$ as established in \eqref{eq:hatt_to_tstar}:
{\small	\begin{align*}
		\!\left|\text{FDR}_{\pmb{\delta^*}}-\text{FDR}_{\pmb{\hat{\delta}}}\right|
		&=\left|
		\mathbb{E}\!\left[
		\dfrac{\sum_{j\in\mathcal{D}^{\rm test}}\mathbb{I}(V_j\le t^*,\,V_j<\tilde V_j,\,Y_j=0)}
		{\sum_{j\in\mathcal{D}^{\rm test}}\mathbb{I}(V_j\le t^*,\,V_j<\tilde V_j)}
		-\dfrac{\sum_{j\in\mathcal{D}^{\rm test}}\mathbb{I}(V_j\le\hat t,\,V_j<\tilde V_j,\,Y_j=0)}
		{\sum_{j\in\mathcal{D}^{\rm test}}\mathbb{I}(V_j\le\hat t,\,V_j<\tilde V_j)}
		\right]\right|\\[0.3em]
		&\le 
		\mathbb{E}\!\left[
		\left|
		\dfrac{\sum_{j\in\mathcal{D}^{\rm test}}\mathbb{I}(V_j\le t^*,\,V_j<\tilde V_j,\,Y_j=0)}
		{\sum_{j\in\mathcal{D}^{\rm test}}\mathbb{I}(V_j\le t^*,\,V_j<\tilde V_j)}
		-\dfrac{\sum_{j\in\mathcal{D}^{\rm test}}\mathbb{I}(V_j\le\hat t,\,V_j<\tilde V_j,\,Y_j=0)}
		{\sum_{j\in\mathcal{D}^{\rm test}}\mathbb{I}(V_j\le\hat t,\,V_j<\tilde V_j)}
		\right|
		\right]\\
		&=
		\mathbb{E}\!\left[A\mid t^*=\hat t\right]\mathbb{P}(t^*=\hat t)
		+\mathbb{E}\!\left[A\mid t^*\neq\hat t\right]\mathbb{P}(t^*\neq\hat t)\\
		&\le \mathbb{P}(t^*\neq\hat t)=o(1),
	\end{align*}}
	where $A$ denotes the term $
	\frac{\sum_{j\in\mathcal{D}^{\rm test}}\mathbb{I}(V_j\le t^*,\,V_j<\tilde V_j,\,Y_j=0)}
	{\sum_{j\in\mathcal{D}^{\rm test}}\mathbb{I}(V_j\le t^*,\,V_j<\tilde V_j)}
	-\frac{\sum_{j\in\mathcal{D}^{\rm test}}\mathbb{I}(V_j\le\hat t,\,V_j<\tilde V_j,\,Y_j=0)}
	{\sum_{j\in\mathcal{D}^{\rm test}}\mathbb{I}(V_j\le\hat t,\,V_j<\tilde V_j)}$ for simplicity.		
	Finally, combining the above arguments, we conclude that $	\lim\limits_{m\to\infty}\text{FDR}_{\pmb{\delta^{\tau}}}=\alpha$, which completes the proof by confirming the FDR of BC algorithm can asymptotically achieve its target level $\alpha$.
\end{proof}

\subsection{Proof of Theorem \ref{thm:oracle_power}}
\begin{proof}
	Let $V_j\equiv p_j/w_j$ and $\tilde{V}_j\equiv \tilde{p}_j/w_j$ denote the weighted p-values with generic weight $w_j$ for each $j\in\mathcal{D}^{\rm test}$, where $p_j$ and $\tilde{p}_j$ are super-uniform conformal p-values constructed by \eqref{eq:conformal_p} with swap-invariant $s(\cdot)$ satisfying \eqref{eq:s_function_permu} under data exchangeability condition \eqref{eq:data_exch_joint}.
	We begin with a preliminary invariance claim: the conditional probability $\mathbb{P}\left(\eta_j=1\mid  Y_j=0,S_j\right)$ is identical for any $j\in\mathcal{D}^{\rm test}$. Note that $\eta_j$ does not depend on $S_j$, the claim then follows from the lemma of \cite{barber2015controlling}:
	\begin{lemma}$[\text{\cite{barber2015controlling}, Lemma 1}]$\label{l1}
		For any anti-symmetric function $h(x,y)$ satisfying $h(x,y)=-h(y,x)$, if the scores $\{(V_j,\tilde{V}_j):j\in\mathcal{D}^{\rm test}\}$ are pairwise exchangeable under the null, i.e., condition \eqref{eq:pwex-scores} holds, then $\left\{\text{sign}\left(h(V_j,\tilde{V}_j)\right):j\in\mathcal{H}_0\right\}$ are i.i.d. coin flips conditional on $\left\{|h(V_j,\tilde{V}_j)|:j\in\mathcal{D}^{\rm test}\right\}$.
	\end{lemma}
	Let $h(x,y)\coloneqq\text{sign}(x-y)(x\wedge y)$. By Lemma \ref{l1} and by the independence of $\eta_j$ and $S_j$, the conditional probability $\mathbb{P}\left(\eta_j=1\mid  Y_j=0,S_j\right)$ is the same for all $j\in\mathcal{D}^{\rm test}$.
	Based on this, we know that, for any generic weights $\pmb{w}=(w_j:j\in\mathcal{D}^{\rm test})$, the adjusted new weights (used for the proofs) $\pmb{w^*}=(w_j^*:j\in\mathcal{D}^{\rm test})$ that are derived from $\pmb{w}$ by \eqref{eq:w_star}  differ from $\pmb{w}$ only by a common multiplicative factor; consequently, the ordering of $\left\{p_j/w_j:j\in\mathcal{D}^{\rm test}\right\}$ coincides with that of $\left\{p_j/w_j^*:j\in\mathcal{D}^{\rm test}\right\}$. Moreover, note that rescaling $p_j$ and $\tilde{p}_j$ by $w_j$ or $w_j^*$ does not affect the pairwise comparison between them: $\mathbb{I}(p_j/w_j<\tilde{p}_j/w_j)=\mathbb{I}(p_j<\tilde{p}_j)=\eta_j$. As a result, since the generic decision rule $\delta^{\pmb{w}}(t)$ only depends on the orderings of $\left\{p_j/w_j:j\in\mathcal{D}^{\rm test}\right\}$ and $\left\{\tilde{p}_j/w_j:j\in\mathcal{D}^{\rm test}\right\}$, and on the indicators $\left\{\eta_j:j\in\mathcal{D}^{\rm test}\right\}$, we obtain that $\Psi_{\text{OR}}(\pmb{w})=\Psi_{\text{OR}}(\pmb{w^*})$ for any generic weights $\pmb{w}$.
	Consequently, it suffices for the power analysis to show that 
	\begin{equation}\label{eq:psiinequality}
		\Psi_{\text{OR}}(\pmb{w_s^*})\geq\Psi_{\text{OR}}(\pmb{1}),
	\end{equation}
	where $\pmb{w_s^*}\coloneqq(w_j^*:j\in\mathcal{D}^{\rm test})$ are the adjusted weights constructed by \eqref{eq:w_star} with the generic weights $(w_j:j\in\mathcal{D}^{\rm test})$ being specified as the structure-adaptive weights $\pmb{w_s}\coloneqq(w_j:j\in\mathcal{D}^{\rm test})$ utilized in the SCQ procedure.
	Define $\Psi(t,\pmb{w_s})= \mathbb{E}\left[\sum_{j\in\mathcal{D}^{\rm test}} \mathbb{I}( Y_j = 1, \delta^{w_j}(t) = 1)\right]$ as the expected number of true positives obtained from decision rule $\pmb{\delta^{w_s}}(t)$.
	Then, we will first show that, for any $t\in[0,1]$,
	\begin{equation}\label{psi comparison}
		\Psi(t,\pmb{w_s^*})\geq\Psi(t,\pmb{1}).
	\end{equation}
	We can rewrite $\Psi(t,\pmb{w_s})$ as follows:
	\begin{equation*}\label{psi}
		\begin{aligned}
			\Psi(t,\pmb{w_s})
			&=\mathbb{E}\left[\sum_{j\in\mathcal{D}^{\rm test}} \mathbb{I}( Y_j = 1, \delta^{w_j}(t) = 1)\right]
			=\mathbb{E}\left[\sum_{j\in\mathcal{D}^{\rm test}}\mathbb{I}(Y_j=1,p_j/w_j\leq t,p_j/w_j\leq \tilde{p}_j/w_j)\right]\\
			&=\mathbb{E}\left\{\mathbb{E}\left[\sum_{j\in\mathcal{D}^{\rm test}}\mathbb{I}( Y_j=1,p_j/w_j\leq t,p_j/w_j\leq \tilde{p}_j/w_j)\mid S_j\right]\right\}\\
			&=\mathbb{E}\left[\sum_{j\in\mathcal{D}^{\rm test}}\mathbb{P}\left( Y_j=1,\eta_j=1\mid S_j\right)\cdot\mathbb{P}\left(p_j\leq w_jt\mid\eta_j=1,  Y_j=1,S_j\right)\right]\\
			&=\mathbb{E}\left[\sum_{j\in\mathcal{D}^{\rm test}}\mathbb{P}\left( Y_j=1,\eta_j=1\mid S_i\right)\cdot F_{1,1,j}(w_jt)\right].
		\end{aligned}
	\end{equation*}
	Under Assumption \ref{ass:F_alt} (the regularity of $F_{1,1,j}$) and the weight condition in Assumption \ref{ass:weight_info}, we obtain that
	\begin{footnotesize}
		\begin{align*}
			\refstepcounter{equation}
			&\sum_{j\in\mathcal{D}^{\rm test}}\mathbb{P}\left( Y_j=1,\eta_j=1\mid S_j\right)\cdot F_{1,1,j}(w^*_jt)
			=\sum_{j\in\mathcal{D}^{\rm test}}\mathbb{P}\left( Y_j=1,\eta_j=1\mid S_j\right)\cdot F_{1,1,j}(t/w_j^{*-1})\\
			&\geq\sum_{j\in\mathcal{D}^{\rm test}}\mathbb{P}\left( Y_j=1,\eta_j=1\mid S_j\right)\cdot F_{1,1,j}\left(t\cdot\dfrac{\sum_{i\in\mathcal{D}^{\rm test}}\mathbb{P}\left(\eta_i=1, Y_i=1\mid S_i\right)}{\sum_{i\in\mathcal{D}^{\rm test}}\mathbb{P}\left(\eta_i=1, Y_i=1\mid S_i\right)w_i^{*-1}} \right)\\
			&=\sum_{j\in\mathcal{D}^{\rm test}}\mathbb{P}\left( Y_j=1,\eta_j=1\mid S_j\right)\cdot 
			F_{1,1,j}\left(t\cdot\frac{\sum_{i\in\mathcal{D}^{\rm test}} \mathbb{P}(\eta_i = 1,  Y_i = 0 \mid S_i)}{\sum_{i\in\mathcal{D}^{\rm test}} \mathbb{P}( Y_i = 0 \mid S_i) w_i}		\dfrac{\sum_{i\in\mathcal{D}^{\rm test}} \mathbb{P}(\eta_i = 1,  Y_i = 1 \mid S_i)}{\sum_{i\in\mathcal{D}^{\rm test}} \mathbb{P}(\eta_i = 1,  Y_i = 1 \mid S_i) w_i^{-1}}\right)\\
			&\geq\sum_{j\in\mathcal{D}^{test}}\mathbb{P}\left( Y_j=1,\eta_j=1\mid S_j\right)\cdot F_{1,1,j}(t).\tag{\theequation}
			\label{eq:result}
		\end{align*}
	\end{footnotesize}
	
	Then, for any $t\in[0,1]$, we have that
	\begin{equation*}
		\begin{split}
			\Psi(t,\pmb{w_s^*})
			&=\mathbb{E}\left[\sum_{j\in\mathcal{D}^{\rm test}}\mathbb{P}\left( Y_j=1,\eta_j=1\mid S_j\right)\cdot F_{1,1,j}(w^*_jt)\right]\\
			&\geq\mathbb{E}\left[\sum_{j\in\mathcal{D}^{\rm test}}\mathbb{P}\left( Y_j=1,\eta_j=1\mid S_j\right)\cdot F_{1,1,j}(t)\right]
			=\Psi(t,\pmb{1}).
		\end{split}
	\end{equation*}

	Let $N=|\mathcal{D}_0^{\rm cal}|$ be the number of calibration data utilized to compute conformal p-values by \eqref{eq:conformal_p}. Next, we show that, for any $t\in[\dfrac{1}{N+1},1]$, we have
	\begin{equation}\label{Q(t) comparison}
		Q(t,\pmb{w^*_s})\leq Q(t,\pmb{1}).
	\end{equation}
	For the notational convenience, let $a_{10}^j=\mathbb{P}\left(\eta_j=1, Y_j=0\mid S_j\right)$ and $a_{11}^j=\mathbb{P}\left(\eta_j=1, Y_j=1\mid S_j\right)$.
	Then, $Q(t,\pmb{w^*_s})$ can be written as
{\footnotesize	\begin{equation*}
		\begin{aligned}
			Q(t,\pmb{w^*_s})&=\dfrac{\sum_{j\in\mathcal{D}^{\rm test}}\mathbb{P}\left(\delta^{w_j^*}(t)=1, Y_j=0\mid S_j\right)}{\sum_{j\in\mathcal{D}^{\rm test}}\mathbb{P}\left(\delta^{w_j^*}(t)=1\mid S_j\right)}\\
			&=\dfrac{\sum_{j\in\mathcal{D}^{\rm test}}\mathbb{P}\left(p_j\leq w^*_jt\mid\eta_j=1,  Y_j=0,S_j\right)\cdot a_{10}^j}{\sum_{j\in\mathcal{D}^{\rm test}} a_{10}^j\mathbb{P}\left(p_j\leq w^*_jt\mid\eta_j=1,  Y_j=0,S_j\right)+\sum_{j\in\mathcal{D}^{\rm test}}  a_{11}^j\mathbb{P}\left(p_j\leq w^*_jt\mid\eta_j=1,  Y_j=1,S_j\right)}\\
		\end{aligned}	
	\end{equation*}}
	Under the null, we have the bound: $\mathbb{P}\left(p_j\leq w^*_jt\mid\eta_j=1,  Y_j=0,S_j\right)\leq\dfrac{w^*_jt}{\mathbb{P}\left(\eta_j=1\mid  Y_j=0,S_j\right)}$. Consequently, 
{\footnotesize	\begin{equation*}
		\begin{aligned}
			Q(t,\pmb{w^*_s})&\leq\dfrac{\sum_{j\in\mathcal{D}^{\rm test}}\dfrac{w^*_jt}{\mathbb{P}\left(\eta_j=1\mid  Y_j=0,S_j\right)}\cdot a_{10}^j}{\sum_{j\in\mathcal{D}^{\rm test}}\dfrac{w^*_jt}{\mathbb{P}\left(\eta_j=1\mid  Y_j=0,S_j\right)}\cdot a_{10}^j+\sum_{j\in\mathcal{D}^{\rm test}}\mathbb{P}\left(p_j\leq w^*_jt\mid\eta_j=1,  Y_j=1,S_j\right)\cdot a_{11}^j}\\
			&=\dfrac{\sum_{j\in\mathcal{D}^{\rm test}}t\cdot a_{10}^j}{\sum_{j\in\mathcal{D}^{\rm test}}t\cdot a_{10}^j+\sum_{j\in\mathcal{D}^{\rm test}}\mathbb{P}\left(p_j\leq w^*_jt\mid\eta_j=1,  Y_j=1,S_j\right)\cdot a_{11}^j}\\
		\end{aligned}
	\end{equation*}}
	Let $F_{1,0,j}(t)=\mathbb{P}\left(p_j\leq t\mid\eta_j=1, Y_j=0,S_j\right)$. Using the Lemma \ref{lemma:F10i} stated below (which guarantees $F_{1,0,i}(t)\geq t$) and the inequality \eqref{eq:result} established above, for any $t\in[\dfrac{1}{N+1},1]$, we obtain
	\begin{equation*}
		Q(t,\pmb{w^*_s})
		\leq\dfrac{\sum_{j\in\mathcal{D}^{\rm test}}F_{1,0,j}(t)\cdot a_{10}^j}{\sum_{j\in\mathcal{D}^{\rm test}}F_{1,0,j}(t)\cdot a_{10}^j+\sum_{j\in\mathcal{D}^{\rm test}}F_{1,1,j}(t)\cdot a_{11}^j}\\
		=Q(t,\pmb{1}).
	\end{equation*}
	\begin{lemma}\label{lemma:F10i}
		Suppose the number of calibration data used for calculating the conformal p-values $p_i$ is $N$. Assume that the initial score function $s(\cdot)$ utilized to calculate the conformal p-values defined in \eqref{eq:conformal_p} is continuous and satisfies the invariance principle \eqref{eq:s_function_permu}. Then, under data exchangeability condition \eqref{eq:data_exch_joint}, for any $t\in[1/(N+1),1]$, the conditional distribution function $F_{1,0,i}(t)\coloneqq\mathbb{P}\left(p_i\leq t\mid\eta_i=1, Y_i=0,S_i\right)$ satisfies: $F_{1,0,i}(t)\geq t$.
	\end{lemma}
	
	The proof of Lemma \ref{lemma:F10i} will be given in Appendix \ref{proof:lemma-F10i}. Note that each $p_i$ takes values at $\left\{\dfrac{1}{N+1},\dfrac{2}{N+1},\dots,1\right\}$. Then, if $t_{\rm OR}^\alpha(\pmb{1})$ lies in $[0,\dfrac{1}{N+1})$, the inequality \eqref{eq:psiinequality} holds naturally since $\Psi_{\text{OR}}(\pmb{1})=0$ at that time. Next, if $t_{\rm OR}^\alpha(\pmb{1})\in[\dfrac{1}{N+1},1]$, by \eqref{Q(t) comparison}, 
	we have $Q(t_{\rm OR}^\alpha(\pmb{1}),\pmb{w^*_s})\leq Q(t_{\rm OR}^\alpha(\pmb{1}),\pmb{1})\leq\alpha$. Hence, $t_{\rm OR}^\alpha(\pmb{1})\leq t_{\rm OR}^\alpha(\pmb{w_s^*})$. Finally, by \eqref{psi comparison} and the monotonicity of $\Psi(t,\cdot)$ in the threshold $t$, we can establish the desired inequality as follows:
	\begin{equation*}
		\Psi_{\rm OR}(\pmb{1})\equiv\Psi(t_{\rm OR}^\alpha(\pmb{1}),\pmb{1})\leq\Psi(t_{\rm OR}^\alpha(\pmb{w_s^*}),\pmb{1})\leq\Psi(t_{\rm OR}^\alpha(\pmb{w_s^*}),\pmb{w^*_s})\equiv\Psi_{\rm OR}(\pmb{w_s^*})=\Psi_{\rm OR}(\pmb{w_s}).
	\end{equation*}
\end{proof}

\subsection{Generalization of Theorem \ref{thm:BC_power} and its proof}\label{proof:BC_power}

The result of Theorem~\ref{thm:BC_power} can be generalized to mirror-based algorithms with any pairwise exchangeable conformity scores, taking the p-value weighting procedure as a special case.
We formally state this in Theorem~\ref{thm:BC_power_general} below and subsequently provide its proof, with Theorem~\ref{thm:BC_power} following as an immediate corollary. 
In Theorem \ref{thm:BC_power_general}, we establish the asymptotic equivalence between the expected true-positive counts produced by the mirror-based algorithm using the data-driven threshold $\tau$ \eqref{eq:tau} and the oracle procedure with threshold $t^*$ \eqref{eq:tstar}. This provides the foundation of extending existing power analysis theories, which are based on the oracle thresholds to the data-driven procedures (see Theorem \ref{thm:DD_power}).
Recall the definitions of $\pmb{\delta^{\tau}}$ and $\pmb{\delta^*}$, which represent the decision rules of the aforementioned two procedures, respectively: $\pmb{\delta^{\tau}}=\left\{\delta_j^{\tau}\equiv \mathbb{I}(V_j\leq\tau,V_j<\tilde{V}_j):j\in\mathcal{D}^{\rm test}\right\}$ and  $\pmb{\delta^*}=\left\{\delta_j^*\equiv \mathbb{I}(V_j\leq t^*,V_j<\tilde{V}_j):j\in\mathcal{D}^{\rm test}\right\}$.
We adjust Assumption \ref{ass:ER=Om} to the more general Assumption \ref{ass:ER=Om_general}, which does not specify the score type, and then present Theorem \ref{thm:BC_power_general}, whose proof is given thereafter.
\begin{assumption}\label{ass:ER=Om_general}
	For $m$ pairs of weighted conformal p-values $\{(V_j,\tilde{V}_j):j\in\mathcal{D}^{\rm test}\}$, we have	$\mathbb{E}\left[\sum_{j\in\mathcal{D}^{\rm test}}\mathbb{I}(V_j\leq t^* ,V_j<\tilde{V}_j)\right]\asymp m$, where $t^*$ is the oracle threshold defined in \eqref{eq:tstar}. 
\end{assumption}

\begin{theorem}[Generalization of Theorem \ref{thm:BC_power}]\label{thm:BC_power_general}
	Suppose we have $m$ pairs of conformity scores $\{(V_j,\tilde{V}_j):j\in\mathcal{D}^{\rm test}\}$ satisfying the pairwise exchangeability \eqref{eq:pwex-scores}.
	Define $\Psi_{\pmb{\delta^{\tau}}}\coloneqq\mathbb{E}\left[\sum_{j\in\mathcal{D}^{\rm test}}\mathbb{I}( Y_j=1,\delta_j^{\tau}=1)\right]$ and $\Psi_{\pmb{\delta^{*}}}\coloneqq\mathbb{E}\left[\sum_{j\in\mathcal{D}^{\rm test}}\mathbb{I}( Y_j=1,\delta_j^{*}=1)\right]$ as the expected numbers of the true positives obtained from $\pmb{\delta^{\tau}}$ and $\pmb{\delta^*}$, which denote the decision rules of the mirror-based algorithms with data driven threshold $\tau$ \eqref{eq:tau} and oracle threshold $t^*$ \eqref{eq:tstar}, respectively. Then,	under Assumptions \ref{ass:sep} and~\ref{ass:ER=Om_general}, and the conditions in Lemma \ref{lem:cov}, we have 
	$$
	\lim_{m \to \infty}\dfrac{\Psi_{\pmb{\delta^{\tau}}}}{\Psi_{\pmb{\delta^*}}} = 1.
	$$
\end{theorem}
\begin{proof}	
	We follow the similar arguments as in the proof of Theorem \ref{thm:asymptotic-FDR} in Appendix \ref{proof:BC FDR analysis} and retain the notations introduced therein.
	Recall the definition of $\hat{t}$ in \eqref{eq:hat-t}, denote $\Psi_{\pmb{\hat{\delta}}}\coloneqq\mathbb{E}\left[\sum_{j\in\mathcal{D}^{\rm test}}\mathbb{I}( Y_j=1,\hat{\delta}_j=1)\right]$ as expected number of true positives obtained from $\pmb{\hat{\delta}}=\left\{\hat{\delta}_j\equiv \mathbb{I}(V_j\leq\hat{t},V_j<\tilde{V}_j):j\in\mathcal{D}^{\rm test}\right\}$, which represents the decision rule of the mirror-based algorithm with threshold $\hat{t}$.
	Since we have shown $\hat{t}-t^*\overset{p}{\to}0$ in Appendix \ref{proof:BC FDR analysis}, we know that
	
	\begin{align*}
		\left|\dfrac{1}{m}\left(\Psi_{\pmb{\hat{\delta}}}-\Psi_{\pmb{\delta^{*}}}\right)\right|
		&=\left|\mathbb{E}\left[\dfrac{1}{m}\sum_{j\in\mathcal{D}^{\rm test}}\mathbb{I}( Y_j=1,\hat{\delta}_j=1)-\dfrac{1}{m}\sum_{j\in\mathcal{D}^{\rm test}}\mathbb{I}( Y_j=1,\delta_j^*=1)\right]\right|\\
		&\leq\mathbb{E}\left[\left|\dfrac{\sum_{j\in\mathcal{D}^{\rm test}}\left\{\mathbb{I}( Y_j=1,\hat{\delta}_j=1)-\mathbb{I}( Y_j=1,\delta_j^*=1)\right\}}{m}\right|\mid \hat{t}=t^*\right]\mathbb{P}\left(\hat{t}=t^*\right)\\
		&\quad+\mathbb{E}\left[\left|\dfrac{\sum_{j\in\mathcal{D}^{\rm test}}\left\{\mathbb{I}( Y_j=1,\hat{\delta}_j=1)-\mathbb{I}( Y_j=1,\delta_j^*=1)\right\}}{m}\right|\mid \hat{t}\neq t^*\right]\mathbb{P}\left(\hat{t}\neq t^*\right)\\
		&\leq\mathbb{P}\left(\hat{t}\neq t^*\right)
		=o(1).		
	\end{align*}
	Under Assumption \ref{ass:ER=Om_general}, we have $\Psi_{\pmb{\delta^{*}}}=(1-\alpha)\mathbb{E}\left[\sum_{j\in\mathcal{D}^{\rm test}}\mathbb{I}(\delta_j^*=1)\right]\asymp m.$ Therefore, we can conclude that 
	\begin{equation*}
		\dfrac{\Psi_{\pmb{\delta^{\tau}}}}{\Psi_{\pmb{\delta^*}}} = \dfrac{\Psi_{\pmb{\hat{\delta}}}}{\Psi_{\pmb{\delta^*}}}=1+\dfrac{\dfrac{1}{m}\left(\Psi_{\pmb{\hat{\delta}}}-\Psi_{\pmb{\delta^{*}}}\right)}{\dfrac{1}{m}\Psi_{\pmb{\delta^*}}}=1+o(1),
	\end{equation*}
	which establishes the desired asymptotic equivalence between the expected numbers of true positives obtained by the mirror-based algorithms using the data-driven threshold $\tau$ and those obtained under the oracle threshold $t^*$.
\end{proof}

\subsection{Proof of Theorem \ref{thm:DD_power}}
\begin{proof}
	Combining the results from Theorem \ref{thm:oracle_power} and Theorem \ref{thm:BC_power}, we have Theorem \ref{thm:DD_power} proved by noticing that
	\begin{equation*}
		\begin{aligned}
			\lim\limits_{m \to \infty}\dfrac{\Psi_{\text{DD}}(\pmb{w_s})}{\Psi_{\text{DD}}(\pmb{1})}
			&= \lim\limits_{m \to \infty}\dfrac{\Psi_{\text{DD}}(\pmb{w_s})}{\Psi_{\text{OR}}(\pmb{w_s})}\cdot\dfrac{\Psi_{\text{OR}}(\pmb{w_s})}{\Psi_{\text{OR}}(\pmb{1})}\cdot\dfrac{\Psi_{\text{OR}}(\pmb{1})}{\Psi_{\text{DD}}(\pmb{1})}\\
			&\geq\lim\limits_{m \to \infty}\dfrac{\Psi_{\text{DD}}(\pmb{w_s})}{\Psi_{\text{OR}}(\pmb{w_s})}\cdot\dfrac{\Psi_{\text{OR}}(\pmb{1})}{\Psi_{\text{DD}}(\pmb{1})}=1,
		\end{aligned}
	\end{equation*}
	where the last inequality and the last equality follow from Theorem \ref{thm:oracle_power} and Theorem \ref{thm:BC_power}, respectively.
\end{proof}

\subsection{Proof of Lemma \ref{lem:cov}}
\begin{proof}
	Without loss of generality, we assume $w_j \in (0,1]$. Otherwise, one can normalize $(w_j : j \in \mathcal{D}^{\rm test})$, which leaves the score ranking -- and hence the decision rule -- unchanged. For any four null conformal p-values, $(p_1,p_2,p_3,p_4)\in\{\frac{1}{N+1},\frac{2}{N+2},\dots,1\}^4$, which are constructed by \eqref{eq:conformal_p} with continuous score function $s(\cdot)$, there exist the following five behaviors that $(p_1,\dots,p_4)$ exchangeably exhibit:
	\begin{itemize}
		\item [A:] $p_1=p_2=p_3=p_4=\dfrac{a}{N+1}$;
		\item [B:] $p_1=p_2=p_3=\dfrac{a}{N+1},p_4=\dfrac{b}{N+1}$;
		\item [C:] $p_1=p_2=\dfrac{a}{N+1},p_3=p_4=\dfrac{b}{N+1}$;
		\item [D:] $p_1=p_2=\dfrac{a}{N+1},p_3=\dfrac{b}{N+1},p_4=\dfrac{c}{N+1}$;
		\item [E:]
		$p_1=\dfrac{a}{N+1},p_2=\dfrac{b}{N+1},p_3=\dfrac{c}{N+1},p_4=\dfrac{d}{N+1}$,
	\end{itemize}
	where $a,b,c,d\in\{1,\dots,N+1\}$ and $a\neq b\neq c\neq d$. Denote $\Pi_4 = (N+4)(N+3)(N+2)(N+1)$, we have the joint distribution of $(p_1,\dots,p_4)$ as follows:
	\begin{equation}\label{eq:joint_four_p}
		\begin{gathered}
			\mathbb{P}(p_1=p_2=p_3=p_4=\dfrac{a}{N+1})=\dfrac{24}{\Pi_4},\\
			\mathbb{P}(p_1=p_2=p_3=\dfrac{a}{N+1},p_4=\dfrac{b}{N+1})=\dfrac{6}{\Pi_4},\\
			\mathbb{P}(p_1=p_2=\dfrac{a}{N+1},p_3=p_4=\dfrac{b}{N+1})=\dfrac{4}{\Pi_4},\\
			\mathbb{P}(p_1=p_2=\dfrac{a}{N+1},p_3=\dfrac{b}{N+1},p_4=\dfrac{c}{N+1})=\dfrac{2}{\Pi_4},\\
			\mathbb{P}(p_1=\dfrac{a}{N+1},p_2=\dfrac{b}{N+1},p_3=\dfrac{c}{N+1},p_4=\dfrac{d}{N+1})=\dfrac{1}{\Pi_4}.
		\end{gathered}
	\end{equation}	
	
	In the remainder of this section, we prove that the first condition in \eqref{eq:cov} holds under the assumptions of Lemma~\ref{lem:cov}, while the proof of the remaining condition is analogous and therefore omitted.	
	For any null units $i,j\in\mathcal{H}_0$, and any $t,w_i,w_j\in(0,1)$, we discuss $\text{Cov}\left[R_i(t),R_j(t)\right]$ under two situations. First, when there exists some $l\in\{0,\dots,N\}$ such that $tw_i$ and $tw_j$ both belong to $[\frac{l}{N+1},\frac{l+1}{N+1})$, by the joint distribution given in \eqref{eq:joint_four_p}, we have 
	\begin{align*}
		\mathbb{P}(&p_i\leq \dfrac{l}{N+1},p_i<\tilde{p}_i,p_j\leq \dfrac{l}{N+1},p_j<\tilde{p}_j)\\
		&=\sum_{a=1}^l\sum_{b=1}^l\sum_{c=a+1}^{N+1}\sum_{d=b+1}^{N+1}\mathbb{P}(p_i=\dfrac{a}{N+1},\tilde{p}_i=\dfrac{c}{N+1},p_j=\frac{b}{N+1},\tilde{p}_j=\frac{d}{N+1})\\
		&=
		\dfrac{\,l\left(12N^{2}l+12N^{2}-12Nl^{2}+24Nl+12N+3l^{3}-22l^{2}+9l+10\right)}{12(N+4)(N+3)(N+2)(N+1)}.
	\end{align*}
	Meanwhile, we also obtain 
	\begin{equation*}
		\mathbb{P}(p_i\leq \dfrac{l}{N+1},p_i<\tilde{p}_i)	
		=\sum_{j=1}^l\sum_{k=j+1}^{N+1}
		\mathbb{P}(p_i=\frac{j}{N+1},\tilde p_i=\frac{k}{N+1})
		=\dfrac{l}{N+2}-\frac{l(l+1)}{2(N+1)(N+2)}.
	\end{equation*}
	Hence, we have
	{\setlength{\belowdisplayskip}{4pt}
		\begin{align*}
			\text{Cov}\left[R_i(t),R_j(t)\right]
			&=\mathbb{P}(p_i\leq\frac{l}{N+1},p_i<\tilde{p}_i,p_j\leq\frac{l}{N+1},p_j<\tilde{p}_j)\\
			&\quad	-\mathbb{P}(p_i\leq\frac{l}{N+1},p_i<\tilde{p}_i)\mathbb{P}(p_j\leq\frac{l}{N+1},p_j<\tilde{p}_j)\\
			&=\dfrac{\,l\left(12N^{2}l+12N^{2}-12Nl^{2}+24Nl+12N+3l^{3}-22l^{2}+9l+10\right)}{12(N+4)(N+3)(N+2)(N+1)}\\
			&\quad -
			(\dfrac{l}{N+2}-\frac{l(l+1)}{2(N+1)(N+2)})^2\\
			&= \dfrac{12l^2N^2-12Nl^3+3l^4}{12(N+4)(N+3)(N+2)(N+1)}-\dfrac{4N^2l^2-4Nl^3+l^4}{4(N+1)^2(N+2)^2} + O(N^{-1})\\
			&=
			-\dfrac{l^2(2N - l)^2 (2N + 5)}{2(N+4)(N+3)(N+2)^2(N+1)^2}+ O(N^{-1})\\
			&=  O(N^{-1}).
	\end{align*}}
	Second, when $tw_i$ and $tw_j$ do not belong to the same interval with length $\frac{1}{N+1}$, there exist $l,k\in\{0,\dots,N\}$, such that $tw_i\in[\frac{l}{N+1},\frac{l+1}{N+1})$ and $tw_j\in[\frac{k}{N+1},\frac{k+1}{N+1})$. Without loss of generality, let $l<k$. Then, similarly we have
	{\setlength{\belowdisplayskip}{4pt}	\begin{align*}
			&\text{Cov}\left[R_i(t),R_j(t)\right]\\
			&=\mathbb{P}(p_i\leq\frac{l}{N+1},p_i<\tilde{p}_i,p_j\leq\frac{k}{N+1},p_j<\tilde{p}_j)\\
			&\quad -\mathbb{P}(p_i\leq\frac{l}{N+1},p_i<\tilde{p}_i)\mathbb{P}(p_j\leq\frac{k}{N+1},p_j<\tilde{p}_j)\\
			&=\frac{
				l\left(4N^{2}k+4N^{2}-2Nk^{2}-2Nkl+12Nk-4N+4N+k^{2}l-5k^{2}-kl+17k-8l-4\right)
			}{
				4(N+1)(N+2)(N+3)(N+4)
			}\\
			&\quad -
			\frac{\big(l^{2}k^{2}+4N^{2}lk-2Nl^{2}k-2Nlk^{2}\big)(2N+5)}
			{2(N+1)^{2}(N+2)^{2}(N+3)(N+4)}\\
			&= \dfrac{4N^2lk-2Nlk^2-2Nkl^2+k^2l^2}{4(N+4)(N+3)(N+2)(N+1)} - \dfrac{l^2k^2+4N^2lk-2Nl^2k-2Nlk^2}{4(N+1)^2(N+2)^2} + O(N^{-1})\\
			&= -\frac{k l (2N-k)(2N-l)(2N+5)}{2 (N+4)(N+3)(N+2)^2 (N+1)^2}
			+ O(N^{-1})\\
			&=O(N^{-1}).
	\end{align*}}
	Note that
	\begin{align*}
		\dfrac{1}{m^2}\sum\limits_{\substack{i<j,\\i,j\in\mathcal{D}^{\rm test}}}\text{Cov}\left[R_i(t),R_j(t)\right]
		&=\dfrac{1}{m^2}\{(\sum\limits_{\substack{i<j,\\i,j\in\mathcal{H}_0}}+\sum\limits_{\substack{i<j,\\i,j\in\mathcal{H}_1}}+\sum\limits_{\substack{i<j,\\i\in\mathcal{H}_0,\\j\in\mathcal{H}_1}}+\sum\limits_{\substack{i<j,\\i\in\mathcal{H}_1,\\j\in\mathcal{H}_0}})\text{ Cov}\left[R_i(t),R_j(t)\right]\}.
	\end{align*}
	As a result, since we have $\dfrac{1}{m^2}\sum\limits_{\substack{i<j,\\i,j\in\mathcal{H}_0}}\text{Cov}\left[R_i(t),R_j(t)\right]
	= \dfrac{m_0^2-m_0}{2m^2}\cdot O(N^{-1}) = O(m^{-\epsilon_{N}})$, and 
	{\setlength{\belowdisplayskip}{4pt}		\begin{align*}
			\Big|\dfrac{1}{m^2}\{(\sum\limits_{\substack{i<j,\\i,j\in\mathcal{H}_1}}+\sum\limits_{\substack{i<j,\\i\in\mathcal{H}_0,\\j\in\mathcal{H}_1}}+\sum\limits_{\substack{i<j,\\i\in\mathcal{H}_1,\\j\in\mathcal{H}_0}})\text{ Cov}\left[R_i(t),R_j(t)\right]\}\Big|
			&\leq \dfrac{1}{m^2}(\sum\limits_{\substack{i<j,\\i,j\in\mathcal{H}_1}}1+\sum\limits_{\substack{i<j,\\i\in\mathcal{H}_0,\\j\in\mathcal{H}_1}}1+\sum\limits_{\substack{i<j,\\i\in\mathcal{H}_1,\\j\in\mathcal{H}_0}}1)\\
			&=\dfrac{1}{m^2}\left(O(m^{2-2\epsilon_1})+O(m^{2-\epsilon_1})\right)
			=O (m^{-\epsilon_1}),
	\end{align*}}
	we conclude that $\dfrac{1}{m^2}\sum\limits_{\substack{i<j,\\i,j\in\mathcal{D}^{\rm test}}}\text{Cov}\left[R_i(t),R_j(t)\right]\to0$ when $m\to\infty$, which establishes the first part of \eqref{eq:cov}. 
	The proof of the remaining condition follows analogously and is therefore omitted.
\end{proof}

\subsection{Proof of Lemma \ref{lemma:F10i}}\label{proof:lemma-F10i}

\begin{proof}
	We first give the joint distribution of the conformal p-values $p_i$ and $\tilde{p}_i$ calculated by \eqref{eq:conformal_p} for null samples $X_i$ and $\tilde{X}_i$. Let $s_1,\dots,s_N$ denote the scores of samples in $\mathcal{D}_0^{\rm cal}$, and let $s_{N+1}$ and $s_{N+2}$ denote the scores of $X_i$ and $\tilde{X}_i$, respectively, all computed by the scoring function $s(\cdot)$. 
	Denote $r_1,\dots,r_N,r_{N+1},r_{N+2}$ as the corresponding ranks.
	Then, since $s(\cdot)$ is continuous, $s_1,\dots,s_{N+2}$ exchangeably follow a non-atomic distribution and $r_1,\dots,r_{N+2}$ are mutually distinct almost surely. Therefore, $(r_1,\dots,r_{N+2})\sim \text{Unif}\{1,\dots,N+2\}.$
	Thus, for any $j=1,\dots,N+1$,
	
	\begin{equation}\label{joint pi1}
		\begin{aligned}
			\mathbb{P}(p_{N+1}=p_{N+2}=\frac{j}{N+1})
			&=\mathbb{P}\left(r_{N+1}=j,r_{N+2}=j\right)+\mathbb{P}\left(r_{N+1}=j+1,r_{N+2}=j\right)\\
			&=2\mathbb{P}\left(r_{N+1}=j,r_{N+2}=j+1\right)=\dfrac{2}{(N+1)(N+2)},
		\end{aligned}
	\end{equation}
	for any $1\leq j<k\leq N+1$,
	\begin{equation}\label{joint pi2}
		\mathbb{P}(p_{N+1}=\frac{j}{N+1},p_{N+2}=\frac{k}{N+1})
		=\mathbb{P}\left(r_{N+1}=j,r_{N+2}=k+1\right)
		=\dfrac{1}{(N+1)(N+2)},
	\end{equation}
	and symmetrically, 
	\begin{equation}\label{joint pi3}
		\mathbb{P}(p_{N+1}=\frac{k}{N+1},p_{N+2}=\frac{j}{N+1})
		=\dfrac{1}{(N+1)(N+2)}.		
	\end{equation}
	Then, considering the joint distribution of null conformal p-values given in \eqref{joint pi1} - \eqref{joint pi3} above, for any $1\leq l\leq N$, since the derivation of $p_i$ does not depend on $S_i$, we have
	
	\begin{equation*}
		\begin{alignedat}{2}
			\mathbb{P}(p_i\!\le\!\frac{l}{N+1}\mid p_i\!<\!\tilde p_i,Y_i\!=\!0,S_i)
			&=\frac{\mathbb{P}(p_i\le\frac{l}{N+1},\,p_i<\tilde p_i\mid Y_i=0)}
			{\mathbb{P}(p_i<\tilde p_i\mid Y_i=0)}\\[3pt]
			&=\frac{\sum_{j=1}^l\mathbb{P}(p_i=\frac{j}{N+1},p_i<\tilde p_i\mid Y_i=0)}
			{\left(1-\mathbb{P}(p_i=\tilde p_i)\right)/2}\\[3pt]
			&=\frac{\sum_{j=1}^l\sum_{k=j+1}^{N+1}
				\mathbb{P}(p_i=\frac{j}{N+1},\tilde p_i=\frac{k}{N+1}\mid Y_i=0)}
			{\left(1-\sum_{j=1}^{N+1}\mathbb{P}(p_i=\frac{j}{N+1},
				\tilde p_i=\frac{j}{N+1}\mid Y_i=0)\right)/2}\\[3pt]
			&=\frac{\frac{l}{N+2}-\frac{l(l+1)}{2(N+1)(N+2)}}{\frac12-\frac{1}{N+2}}
			\;\ge\;\frac{l+1}{N+1}.
		\end{alignedat}
	\end{equation*}
	
	As a result, we have Lemma \ref{lemma:F10i} proved by noticing that, for any $t\in[\dfrac{1}{N+1},1)$ (the case where $t=1$ holds trivially), suppose $t$ lies in $[\dfrac{l}{N+1},\dfrac{l+1}{N+1})$ for some $1\leq l\leq N$, we obtain the desired inequality as follows 
	\begin{equation*}
		\mathbb{P}\left(p_i\leq t\mid p_i<\tilde{p}_i, Y_i=0,S_i\right)
		\geq\mathbb{P}(p_i\leq\dfrac{l}{N+1}\mid p_i<\tilde{p}_i, Y_i=0,S_i)
		\geq\dfrac{l+1}{N+1}\geq t.
	\end{equation*}
\end{proof}

\subsection{Further details of Assumption \ref{ass:sep}}\label{sec:proof_sep_ass}
We begin with a claim which presents a concrete example showing that Assumption \ref{ass:sep} can be satisfied under some conditions, and we provide its proof thereafter.
\begin{claim}\label{claim1}
	Suppose we have test samples $\{X_i:i\in[m]\}$ and null samples $\{\tilde{X}_i:i\in[m]\}$ generated from the models $		X_i\mid Y_i\overset{ind.}{\sim}(1- Y_i)f_0+ Y_i f_{1m}$ and $\tilde{X}_i\overset{iid.}{\sim}f_0$, respectively, where $f_0(t)$ and $f_{1m}(t)$ obey the following two-point Gaussian mixture model, which is widely used in the high-dimensional sparse inference \citep{donoho2004higher, meinshausen2006estimating, cai2007estimation, TonyCaioptimal2017, arias2017distribution}:
	\begin{equation*}
		f_0(t)=\phi_{\sigma_0}(t-\mu_0)\text{ and }f_{1m}(t)=\phi_{\sigma_0}(t-\mu_m),\text{ with }\mu_m\to\infty\text{ as }m\to\infty,
	\end{equation*}
	where $\phi_{\sigma}$ denotes the density function of an $\mathcal{N}(\mu,\sigma)$ variable. Then, if we construct the initial score function as $s(\cdot)\equiv\frac{f_0(\cdot)}{f(\cdot)}$ and calculate the weighted conformal p-values $\{(V_j\equiv p(X_j)/w_j,\tilde{V}_j\equiv {p}(\tilde{X}_j)/w_j):j\in\mathcal{D}^{\rm test}\}$ via \eqref{eq:conformal_p} and \eqref{eq:weighted_p}, we have $\dfrac{1}{m}\sum_{j=1}^m\mathbb{P}\left\{\tilde{V}_j< V_j\mid  Y_j=1\right\}=o(1)$, which meets Assumption \ref{ass:sep}.
\end{claim}

\begin{proof}
	Under the model in Claim~\ref{claim1}, the desired convergence is established by showing that
	
	\begin{align*}
		\frac{1}{m}\sum_{j=1}^m 
		\mathbb{P}\{\tilde V_j< V_j \mid Y_j=1\}
		&\le 
		\frac{1}{m}\sum_{j=1}^m 
		\mathbb{P}\{\tilde V_j\le V_j \mid Y_j=1\}                                   \\[0.3em]
		&=
		\frac{1}{m}\sum_{j=1}^m 
		\mathbb{P}\!\left\{
		\frac{f_0(\tilde X_j)}{f(\tilde X_j)}
		\le
		\frac{f_0(X_j)}{f(X_j)}
		\;\middle|\; Y_j=1
		\right\}                               \\[0.3em]
		&=
		\frac{1}{m}\sum_{j=1}^m 
		\mathbb{P}\!\left\{
		\frac{\phi_{\sigma_0}(\tilde X_j-\mu_0)}
		{\phi_{\sigma_0}(\tilde X_j-\mu_m)}
		\le
		\frac{\phi_{\sigma_0}(X_j-\mu_0)}
		{\phi_{\sigma_0}(X_j-\mu_m)}
		\;\middle|\; Y_j=1
		\right\}                               \\[0.3em]
		&=
		\frac{1}{m}\sum_{j=1}^m 
		\mathbb{P}\!\left\{
		(\mu_0-\mu_m)\tilde X_j
		\le
		(\mu_0-\mu_m)X_j
		\;\middle|\; Y_j=1
		\right\}                               \\[0.3em]
		&=
		\Phi\!\left(-\frac{|\mu_0-\mu_m|}{\sqrt{2}\,\sigma_0}\right)
		=o(1).
	\end{align*}
	
\end{proof}

\section{Supplementary Methodological Details}\label{sec:supple_methdos}

\subsection{Construction of Structure-Adaptive Weights}\label{sec:weight_construction}

This section introduces a principled and practical method for constructing the structure-adaptive weights used in Step 2 of the SCQ procedure in Section \ref{subsec:SCQ}.

The weight construction for SCQ is motivated by the LAWS procedure \citep{cai2022laws}, which leverages heterogeneity in local sparsity levels -- informed by side information -- to upweight or downweight corresponding p-values. In many high-dimensional settings, signals are sparse: only a small subset of units exhibit meaningful effects, while most are null or negligible. Let \(\pi(S_j) \coloneqq \mathbb{P}(Y_j = 1 \mid S_j)\) denote the \emph{local sparsity level} conditional on side information \(S_j\). When \(S_j\) is categorical (e.g., group membership), \(\pi(S_j)\) corresponds to the signal proportion within that group; when \(S_j\) is ordinal or continuous, higher \(\pi(S_j)\) indicates a greater density of signals around \(S_j\). The key idea is to assign larger weights to coordinates where signals are likely more abundant (in light of side information).

We begin with an oracle setting in which \(\pi(S_j)\) is known. Following an empirical Bayes approach, \cite{cai2022laws} introduced the weighted p-value \(p_j / w_o(S_j)\) with oracle weights
\begin{equation}\label{eq:or_weight}
	w_o(S_j) = \frac{\pi(S_j)}{1 - \pi(S_j)}, \quad j \in \mathcal{D}^{\text{test}},
\end{equation}
which preferentially targets discoveries in regions of higher signal density. Alternative weighting schemes can also be considered, including simple functions such as \(\pi(S_j)\) or \(\{1-\pi(S_j)\}^{-1}\), as well as more complex methods discussed in \cite{ignatiadis2021covariate} and \cite{liang2027locally}.

In practice, \(\pi(S_j)\) is unknown. We adapt the screening approach of \citet{cai2022laws} to estimate this quantity within the conformal framework. Our proposed method addresses two key limitations of the existing approach: (i) direct p‑value derivation is often infeasible in high‑dimensional or structurally complex settings, and (ii) conventional weight‑construction techniques fail to satisfy the required swap‑invariance property \eqref{eq:g_func_swap_inva}. Our carefully designed estimation algorithm ensures both pairwise exchangeability of the weighted p‑values and applicability in scenarios where traditional p‑values are unavailable. The construction proceeds in three steps.

\noindent\textbf{Step 1. Construction of the weight matrix.} We first define a weight matrix \(\mathbf{\Omega}=(\omega_{j,j^\prime})_{m\times m}\) based on the side information vector \(\mathbf S=\{S_j: j\in\mathcal D^{\text{test}}\}\). Each entry \(\omega_{j,j^\prime}\) quantifies the contribution of unit \(j^\prime\) to the estimation of \(\pi(S_j)\). If \(S_j\) is categorical, we set \(\omega_{jj^\prime}=\mathbb{I}(S_j=S_{j^\prime})\). For continuous \(S_j\), we use a kernel-smoothed weighting scheme:
\[
\omega_{jj^\prime}=K_h\bigl(d(S_j,S_{j^\prime})\bigr),\quad \text{where }K_h(t)=h^{-1}K(t/h).
\]
Here \(K(\cdot)\) is a symmetric kernel with \(\int K=1\), \(\int tK=0\) and \(\int t^{2}K<\infty\); \(h\) is a bandwidth selected via permutation-invariant rules (e.g., Silverman’s, Sheather–Jones, or Lepski’s method), and \(d(\cdot,\cdot)\) is a metric on the space of \(\mathbf S\). These choices preserve the swap-invariance required by Proposition \ref{prop:pairwise_scores}.

\noindent\textbf{Step 2. Sparsity estimation.} We construct a swap-invariant estimator:
\begin{equation}\label{eq:hat-pi}
	\hat{\pi}(S_j) = 1 - \frac{\sum_{i=1}^{m} \omega_{ij} \left[ \mathbb{I}(p_i > \lambda) + \mathbb{I}(\tilde{p}_i > \lambda) \right]}{2(1 - \lambda) (\sum_{i=1}^{m} \omega_{ij})},
\end{equation}
where $(p_i, \tilde{p}_i)$ are conformal p-values and $\lambda \in (0, 1)$ is a screening threshold, with the default choice of $\lambda=0.1$. 
\begin{remark}
	As highlighted by \cite{cai2022laws}, the selection of $\lambda$ entails a bias–variance trade-off. Hence, we develop a data-adaptive selection algorithm (P‑TAMS+) within the conformal framework, as described in Appendix~\ref{sec:P-TAMS+} for interested readers. However, this additional step is secondary to our main contributions, and the related discussion is included only for completeness. Instead, we adopt a default value of $\lambda = 0.1$ in all numerical experiments, which delivers stable and satisfactory performance. 
\end{remark}

\noindent\textbf{Step 3. Bias correction.} Under standard regularity conditions, \(\hat{\pi}(S_j)\) converges in probability to \(\pi(S_j)/2\) as \(m\to\infty\), introducing a multiplicative bias. We therefore recalibrate the estimator and construct the final data‑driven weights as
\begin{equation}\label{eq:sd-weight}
	w(S_j)=\frac{\hat{\pi}(S_j)}{1/2-\hat{\pi}(S_j)},\qquad j\in\mathcal{D}^{\rm test}.
\end{equation}
By construction, the estimator satisfy the swap‑invariance condition \eqref{eq:g_func_swap_inva} as desired. 

Finally, we summarize the practical SCQ procedure with the above weight construction methods in the following Algorithm \ref{alg:SCQ-practical}.

\begin{algorithm}[h]
	\caption{The practical SCQ procedure}
	\begin{algorithmic}[1]
		\Require test data $\mathbf{X}=\{X_j: j\in\mathcal{D}^{\rm test}\}$, covariates $\textbf{S}=\{S_j: j\in\mathcal{D}^{\rm test}\}$, mirror data $\tilde{\mathbf X}=\{\tilde{X}_j: j\in\mathcal{D}^{\rm test}\}$, training data $\mathbf X_0^{\rm tr}$, calibration data $\mathbf X_0^{\rm cal}$, labeled outliers $\mathbf{X}_1$ (if available), target FDR level $\alpha$, classifier $\mathcal{C}$, tuning parameter $\lambda$;  
		\Ensure the rejection set $\mathcal{R}\subset\mathcal{D}^{\rm test}$. 
		\State Calculate the weight matrix \(\mathbf{\Omega}=(\omega_{j,j^\prime})_{m\times m}\) based on the side information vector $\mathbf{S}$.
		\State Calculate the initial score function $s(\cdot)$ in the form of \eqref{eq:s_function_permu} with $\mathcal{C}$.
		\For{$j = n+1$ to $n+m$}
		\State Calculate the conformal p-values $p(X_j)$ and $p(\tilde{X}_j)$ for test sample $X_j$ and mirror sample $\tilde{X}_j$ through \eqref{eq:conformal_p}.
		\State Calculate the sparsity estimator $\hat{\pi}(S_j)$ through \eqref{eq:hat-pi}.
		\State Calculate the data-driven weight $w(S_j)$ through \eqref{eq:sd-weight}.
		\State Calculate the conformity scores $V_j$ and $\tilde{V}_j$ for test sample $X_j$ and mirror sample $\tilde{X}_j$ through \eqref{eq:weighted_p}.
		\State Calculate the conformal q-value $q_j$ through \eqref{eq:scq}.
		\EndFor
		\State Let $\mathcal{R}=\left\{j\in\mathcal{D}^{\rm test}:q_j\leq\alpha\right\}$.\\
		\Return The rejection set $\mathcal{R}$.
	\end{algorithmic}
	\label{alg:SCQ-practical}
\end{algorithm}

\subsection{The outlines of SCQ and P-TAMS algorithms}\label{sec:SCQ_P-TAMS algorithm}

This section summarizes the detailed steps of the SCQ procedure (Algorithm \ref{alg:SCQ}) and P-TAMS algorithm (Algorithm \ref{alg:P-TAMS}), corresponding to the methodologies introduced in Sections \ref{subsec:SCQ} and \ref{subsec:P-TAMS} of the main text, respectively.

\begin{algorithm}[h]
	\caption{The SCQ procedure}
	\begin{algorithmic}[1]
		\Require test data $\mathbf{X}=\{X_j: j\in\mathcal{D}^{\rm test}\}$, covariates $\textbf{S}=\{S_j: j\in\mathcal{D}^{\rm test}\}$, mirror data $\tilde{\mathbf X}=\{\tilde{X}_j: j\in\mathcal{D}^{\rm test}\}$, training data $\mathbf X_0^{\rm tr}$, calibration data $\mathbf X_0^{\rm cal}$, labeled outliers $\mathbf{X}_1$ (if available), target FDR level $\alpha$, classifier $\mathcal{C}$;  
		\Ensure the rejection set $\mathcal{R}\subset\mathcal{D}^{\rm test}$. 
		\State Calculate the initial score function $s(\cdot)$ in the form of \eqref{eq:s_function_permu} with $\mathcal{C}$.
		\For{$j = n+1$ to $n+m$}
		\State Calculate the conformal p-values $p(X_j)$ and $p(\tilde{X}_j)$ for test sample $X_j$ and mirror sample $\tilde{X}_j$ through \eqref{eq:conformal_p}.
		\State Calculate the weight $w(S_j)$ satisfying the swap-invariance principle \eqref{eq:g_func_swap_inva}.
		\State Calculate the conformity scores $V_j$ and $\tilde{V}_j$ for test sample $X_j$ and mirror sample $\tilde{X}_j$ through \eqref{eq:weighted_p}.
		\State Calculate the conformal q-value $q_j$ through \eqref{eq:scq}.
		\EndFor
		\State Let $\mathcal{R}=\left\{j\in\mathcal{D}^{\rm test}:q_j\leq\alpha\right\}$.\\
		\Return The rejection set $\mathcal{R}$.
	\end{algorithmic}
	\label{alg:SCQ}
\end{algorithm}

\begin{algorithm}[h]
	\caption{The P-TAMS algorithm}
	\begin{algorithmic}[1] 
		\Require test data $\mathbf{X}=\{X_j: j\in\mathcal{D}^{\rm test}\}$, covariates $\textbf{S}=\{S_j: j\in\mathcal{D}^{\rm test}\}$, mirror data $\tilde{\mathbf X}=\{\tilde{X}_j: j\in\mathcal{D}^{\rm test}\}$, training data $\mathbf X_0^{\rm tr}$, calibration data $\mathbf X_0^{\rm cal}$, labeled outliers $\mathbf{X}_1$ (if available), target FDR level $\alpha$, candidate classifiers $\mathbf{C}=\left\{\mathcal{C}_k:k\in[K]\right\}$, tuning parameter $\alpha_0$; 
		\Ensure the rejection set $\mathcal{R}_*\subset\mathcal{D}^{\rm test}$;

		\For {$k = 1$ to $K$}
		\State Calculate the initial score function $s_k(\cdot)$ in the form of \eqref{eq:g_func_swap_inva} with $\mathcal{C}_k$.
		\State Apply the BH procedure on the pseudo conformal p-values $\bar{\mathcal{P}}_k = \left\{\min\left(p_j^k, \tilde{p}_j^k\right) : j \in \mathcal{D}^{\rm test}\right\}$  at level $\alpha_0$ with conformal p-values calculated through \eqref{eq:conformal_p}, and get the preliminary rejection set $\mathcal{\bar{R}}_k$.
		\State Calculate the conformity scores $\mathcal{V}^k=\left\{V_j^k:j\in\mathcal{D}^{\rm test}\right\}$ and $\tilde{\mathcal{V}}^k=\left\{\tilde{V}_j^k:j\in\mathcal{D}^{\rm test}\right\}$ by Algorithm \ref{alg:SCQ} with $\mathcal{C}=\mathcal{C}_k$.
		\State Calculate the pseudo conformity scores $\mathcal{U}^k=\left\{U_j^k:j\in\mathcal{D}^{\rm test}\right\}$ and $\mathcal{\tilde{U}}^k=\left\{\tilde{U}_j^k:j\in\mathcal{D}^{\rm test}\right\}$ by \eqref{eq:pseudo_outlier} and \eqref{eq:pseudo_inlier}.
		\State Apply Algorithm \ref{alg:SCQ} at level $\alpha$ with $\mathcal{U}^k$ and $\mathcal{\tilde{U}}^k$ being the conformity scores for test data and mirror data. Calculate the corresponding rejection number as $r_k$.
		\EndFor
		\State Select the calssifier $\mathcal{C}_*$ with the largest $r_k$.
		\State Apply Algorithm \ref{alg:SCQ} at level $\alpha$ with $\mathcal{C}=\mathcal{C}_*$ to obtain the refined rejection set $\mathcal{R}_*$.\\
		\Return the refined rejection set $\mathcal{R}_*$.
	\end{algorithmic}
	\label{alg:P-TAMS}
\end{algorithm}

\subsection{P-TAMS augmented with hyperparameter selection}\label{sec:P-TAMS+}

The practical SCQ procedure (Algorithm \ref{alg:SCQ-practical}) proposed in Appendix \ref{sec:weight_construction} requires estimating the local sparsity level via \eqref{eq:hat-pi}, which depends on a tuning parameter $\lambda$. To automate this selection and improve power, we propose P-TAMS+, an extension of P-TAMS that adaptively selects the hyperparameter $\lambda$ (Algorithm~\ref{alg:P-TAMS+}). Under the conditions of Proposition \ref{prop:pairwise_scores}, P-TAMS+ ensures finite-sample FDR control; the proof follows arguments similar to Theorem \ref{thm:P-TAMS-validity} and is omitted. 

\begin{algorithm}[ht]
	\caption{P-TAMS+: P-TAMS algorithm augmented with parameter selection}
	\begin{algorithmic}[1] 
		\Require test data $\mathbf{X}=\{X_j: j\in\mathcal{D}^{\rm test}\}$, covariates $\textbf{S}=\{S_j: j\in\mathcal{D}^{\rm test}\}$, mirror data $\tilde{\mathbf X}=\{\tilde{X}_j: j\in\mathcal{D}^{\rm test}\}$, training data $\mathbf X_0^{\rm tr}$, calibration data $\mathbf X_0^{\rm cal}$, labeled outliers $\mathbf{X}_1$ (if available), target FDR level $\alpha$, candidate classifiers $\mathbf{C}=\left\{\mathcal{C}_k:k\in[K]\right\}$, target FDR level $\alpha$, tuning parameter $\alpha_0$, candidate $\lambda$ values $\{\lambda_l:l\in[L]\}$; 
		\Ensure the rejection set $\mathcal{R}_{**}\subset\mathcal{D}^{\rm test}$.  
		\For {$k = 1$ to $K$}
		\State Calculate the initial score function $s_k(\cdot)$ in the form of \eqref{eq:g_func_swap_inva} with $\mathcal{C}_k$.
		\State Apply BH on $\bar{\mathcal{P}}_k = \left\{\min\left(p_j^k, \tilde{p}_j^k\right) : j \in \mathcal{D}^{\rm test}\right\}$  at level $\alpha_0$ with conformal p-values calculated through \eqref{eq:conformal_p}, and get the preliminary rejection set $\mathcal{\bar{R}}_k$.
		\State Calculate  $\mathcal{V}^k=\left\{V_j^k:j\in\mathcal{D}^{\rm test}\right\}$ and $\tilde{\mathcal{V}}^k=\left\{\tilde{V}_j^k:j\in\mathcal{D}^{\rm test}\right\}$ by Algorithm \ref{alg:SCQ-practical} with $\mathcal{C}=\mathcal{C}_k$ and $\lambda=0.1$.
		\State Calculate the pseudo conformity scores $\mathcal{U}^k=\left\{U_j^k:j\in\mathcal{D}^{\rm test}\right\}$ and $\mathcal{\tilde{U}}^k=\left\{\tilde{U}_j^k:j\in\mathcal{D}^{\rm test}\right\}$ by \eqref{eq:pseudo_outlier} and \eqref{eq:pseudo_inlier}.
		\State Apply Algorithm \ref{alg:SCQ-practical} at level $\alpha$ with $(\mathcal{U}^k,\mathcal{\tilde{U}}^k)$. Record the rejection number as $r_k$.
		\EndFor
		\State Select the calssifier $\mathcal{C}_*$ with the largest $r_k$.
		\State Apply the Algorithm \ref{alg:SCQ-practical} at level $\alpha$ with $\mathcal{C}_*$ to obtain the refined rejection set $\mathcal{R}_*$.\\
		\For {$l = 1$ to $L$}
		\State Calculate the conformity scores $\mathcal{V}^l=\left\{V_i^l:i\in\mathcal{D}^{\rm test}\right\}$ and $\tilde{\mathcal{V}}^l=\left\{\tilde{V}_i^l:i\in\mathcal{D}^{\rm test}\right\}$ by Algorithm \ref{alg:SCQ-practical} with $\mathcal{C}=\mathcal{C}_*$ and $\lambda=\lambda_l$.
		\State Calculate the pseudo conformity scores $\mathcal{U}^l$ and $\mathcal{\tilde{U}}^l$ by \eqref{eq:pseudo_outlier} and \eqref{eq:pseudo_inlier}.
		\State Apply Algorithm \ref{alg:SCQ-practical} at level $\alpha$ with $\mathcal{U}^l$ and $\mathcal{\tilde{U}}^l$. Denote $r_l$ the number of rejections.
		\EndFor
		\State Select $\lambda_*$ with the largest $r_l$.
		\State Apply Algorithm \ref{alg:SCQ-practical} at level $\alpha$ with $\mathcal{C} = \mathcal{C}_*$ and $\lambda=\lambda_*$ to the obtain rejection set $\mathcal{R}_{**}$.\\
		\Return the refined rejection set $\mathcal{R}_{**}$.
	\end{algorithmic}
	\label{alg:P-TAMS+}
\end{algorithm}

\section{Conceptual and Theoretical Comparisons}\label{sec:supple_comparisons}

\subsection{Comparison with ICP method}
In this section, we provide a comprehensive comparison of our proposed methods, SCQ and P-TAMS, with the integrative conformal p-value (ICP) procedure \citep{liang2024integrative}, from four perspectives -- the utilization of side information, the weighting strategy, the AMS approach, and the theories. 

First, the implementation of ICP relies on labeled outliers, whereas SCQ leverages the structures of test samples as side information. The reliance on labeled outliers may cause three problems: the unavailability of labeled outliers in many real-world applications; potential distribution shifts between labeled outliers and test outliers that lead to negative transfer (see Figure~\ref{fig:SCQ_vs_ICP} in Section~\ref{sec:impact_dist_shift}); and the difficulty of determining whether labeled outliers should be used when the data are high-dimensional and imbalanced, a phenomenon noted by \cite{liang2024integrative}.  Unlike ICP, SCQ does not rely on the existence of labeled outliers; instead, it exploits the structural information from test samples, while still being able to incorporate labeled outliers adaptively -- through P-TAMS -- to obtain potential power gain when they are helpful (see Figure~\ref{fig:ams_ob} in Section \ref{sec:impact_imbalance}).

Second, the weighting strategies differ. While ICP computes the conformal p-values for the alternative hypotheses as weights, SCQ proposes the generic weight function $w(\cdot)$ satisfying the swap-invariance principle in \eqref{eq:g_func_swap_inva}. Moreover, SCQ constructs weights by leveraging diverse sources of data, including the structural information, test samples, null samples, and labeled outliers (when available), while ICP solely utilizes the labeled outliers.

Third, the AMS approaches are distinct. ICP-AMS adopts the model-oriented strategy that evaluates the classifier's ability, while P-TAMS utilizes the task-oriented strategy which selects the model in order to maximize the utility for performing a specific task, such as maximizing the number of rejections in OOD testing. Furthermore, ICP-AMS utilizes the labeled outliers to select model, which may underperform when there exists distribution shifts. Crucially, ICP-AMS is based on the joint exchangeability principle and is therefore incapable of dealing with structured OOD testing problems. By contrast, P-TAMS provides a principled and effective solution for the structured OOD testing  by carefully designing the swap-invariant selection procedure.
A detailed numerical comparison of P-TAMS and ICP-AMS is provided in Appendix \ref{sec:P-TAMS_vs_ICPAMS}. 

Finally, the theories also differ in two aspects. On the one hand, ICP applies the conditional calibration method to the weighted p-values to obtain a valid rejection set, at the cost of power loss; SCQ employs the generic BC algorithm to control FDR with asymptotic guarantees on anti-conservativeness (cf. Theorem \ref{thm:asymptotic-FDR}). On the other hand, the power analysis theories are different: while ICP analyzes the power gain with oracle thresholds, we extend the analysis to more practical scenarios with data-driven thresholds (see Section \ref{subsec:power-analysis} for details).

\subsection{Comparisons of P-TAMS and Existing CAMS Methods}\label{sec:P-TAMS_vs_CAMS}

As noted in Sections~\ref{subsec:setup} and \ref{subsec:P-TAMS-general}, classifier families such as OCC, BIC, and PUC each exhibit distinct strengths and limitations. The OCC approach relies solely on inlier samples, providing enhanced stability but failing to utilize information from labeled outliers or test data. Conversely, BIC incorporates labeled outliers during training but does not leverage test data; as observed by \cite{liang2024integrative}, it can become unstable -- sometimes underperforming even OCC in highly imbalanced settings. PUC, meanwhile, exploits test data but ignores labeled outliers, which may not always be desirable, and in high‑dimensional, severely imbalanced scenarios, it may be outperformed by OCC. These limitations highlight the necessity of effective and reliable CAMS methods that can adaptively identify the most suitable classifier, rather than relying on a single pre‑specified model.

In this section, we contrast the proposed P-TAMS algorithm with existing CAMS approaches and articulate its advantages from several key perspectives.

First, existing CAMS methods, which are often built on the BH or conditional calibration frameworks, require full permutation invariance between test and calibration samples to construct valid conformal p‑values, a requirement typically met via additional sample splitting or data mixing \citep{magnani2023collective, marandon2024adaptive}. P‑TAMS introduces a distinct principle that avoids such splitting. Additional splitting not only aggravates data inefficiency, but is also impractical in structured OOD testing, where side information for null samples may be unavailable. 

Second, compared to the AMS strategy in ICP \citep{liang2024integrative}, which also performs OOD testing with FDR control, P‑TAMS provides three main advantages beyond its relaxed exchangeability conditions (detailed in the next paragraph). Unlike ICP, which selects a model separately for each test point \(X_i\), P‑TAMS uses the full dataset to select a single, globally optimal model. Moreover, ICP employs a model‑oriented selection criterion, while P‑TAMS adopts a task‑oriented one. Most importantly, ICP requires labeled outliers, which can cause negative learning when the distribution of labeled outliers differs from that of test outliers. In contrast, P‑TAMS adapts directly to the test data and exploits structural information from auxiliary sources, offering greater robustness to distribution shifts; illustrative results are provided in Figures \ref{fig:ams_occ}, \ref{fig:ams_bic}, \ref{fig:ams_noshift}, \ref{fig:ams_shift} of the supplement and Figure \ref{fig:ams_ob} of the main text.

Third, relative to other CAMS methods less tailored to OOD testing (e.g. \cite{liang2023conformal} for prediction tasks and \cite{magnani2023collective} for lower‑bound estimation), the principal advantage of P‑TAMS is its reliance on pairwise score exchangeability rather than the stricter joint exchangeability required by existing CAMS approaches. This relaxation not only affords greater flexibility, but also makes the method better suited for structured OOD testing problems.

\section{Additional experimental results}\label{sec:supple_numerical_results}

\subsection{Details on methods implementation and experimental setup}\label{sec:detailed_methods}
We first give the implementation details of the methods considered in our numerical experiments (see Section \ref{sec:simu setup} for a brief introduction), which include: cfBH-OCC, ICP-OCC, AdaDetect-KDE, AdaDetect-PU, CLAW-KDE, CLAW-PU, SCQ-OCC, SCQ-BIC, SCQ-PU, and SCQ+P-TAMS. All classifiers involved (e.g., OneClassSVM or RandomForest) and the KDE model (i.e., KernelDensity) are utilized through the Python package \textit{scikit-learn} \citep{pedregosa2011scikit}.

Throughout our numerical experiments, the null sample set $\mathcal{D}_0$ is split into three subsets: $\mathcal{D}_0^{\rm tr}$ and $\mathcal{D}_0^{\rm cal}$, each of equal size, and $\mathcal{D}_0^{\rm mir}$ whose size matches that of the test samples. For cfBH, ICP, and AdaDetect, which rely on jointly exchangeable scores and do not need mirror data, we merge $\mathcal{D}_0^{\rm mir}$ and $\mathcal{D}_0^{\rm cal}$ to form a larger calibration data set for these methods.
Specifically, \textbf{ICP-OCC} first calculates the split conformal p-values of test samples for null and alternative hypotheses, respectively, then takes the quotient to compute the integrative conformal p-values, and finally applies the Storey-BH on them to obtain the rejection set. Note that the conditional calibration method has been applied by \cite{liang2024integrative} to guarantee valid FDR control, while we substitute it with the Storey-BH. This is because we focus on the power comparison in our numerical experiments to explore the underlying detection power of different methods, and applying Storey-BH on ICP is more powerful than the conditional calibration at the theoretical cost of statistical validity (while \cite{liang2024conformal} demonstrates the general robustness of BH by showing that the empirical results of Storey-BH calibration lead to valid FDR control). 

\textbf{AdaDetect-KDE} and \textbf{CLAW-KDE} first estimate the null distribution (on the training data) and the mixture distribution (on the collection of calibration data and test data) using KDE method, then take the quotient to compute the density ratios of test samples, and finally apply the Storey-BH and BC algorithms, respectively, on the density-ratio-based scores to obtain the rejection set. The difference between \textbf{$*$-PU} and \textbf{$*$-KDE} methods, where $*$ represents \textbf{AdaDetect} or \textbf{CLAW}, lies in their calculation of density ratios: the former utilizes PU-learning which trains a binary classifier (we use the random forest in our experiments) with the null training data being the ``positive sample” and the collection of calibration and test data being the ``unlabeled sample", while the latter uses the KDE method.
\textbf{SCQ-KDE} and \textbf{SCQ-PU} are designed as two types of SCQ variants to weight the p-values calculated by \textbf{AdaDetect-KDE} and \textbf{AdaDetect-PU}, respectively: we calculate the conformal p-values as constructed by AdaDetect for both the test data $\mathcal{D}^{\rm test}$ and mirror data $\mathcal{D}_0^{\rm mir}$, and apply the structure-adaptive weights to them. \textbf{SCQ+P-TAMS} is implemented by running the P-TAMS algorithm to perform AMS on candidate models that are specified in different experiments.

In the remainder of this section, we give details of the experimental setup in Figure \ref{fig:ams_ob}, which exhibits the experimental results in Section \ref{sec:impact_imbalance}. Specifically, the simulated data are generated based on a combination of the hierarchical model and the multivariate Gaussian mixture distribution $P_X^a$, which has also been utilized by \cite{bates2023testing} and \cite{liang2024integrative}: the setting is the same as in model \eqref{eq:hierarchical_model} in Section \ref{sec:simu setup}, while the difference lies in the definition of the null and alternative distributions. When we say $X_i\sim P_X^a$, we mean that we define $X_i = \sqrt{a}C_i+D_i$ for some constant $a\geq1$, where $C_i$ is a random vector having $p$ independent standard Gaussian components, and $D_i$ has coordinates sampled independently from the uniform distribution $[-1,1]^p$. We set the null distribution as $P_X^1$, while the two types of alternative distributions are defined as $P_X^4$ and $P_X^{20}$. With fixed parameters $p=150$, $\pi_1=0.6$, and $\pi_2=0.3$, we vary the number of labeled outliers $n_1$ from 0 to 3000. Other settings remain the same as in Section \ref{sec:impact_dist_shift}. We implement the P-TAMS algorithm with four candidate classifiers: two OCCs, including the SVM-sigmoid (OneClassSVM with sigmoid kernel), SVM-poly (OneClassSVM with polynomial kernel), and two BICs, including the KNN ($K$-nearest neighbors) and MLP (multi-layer perceptron).

\subsection{Additional details for the PageBlocks experiment}
\label{sec:app_pageblocks}

We use the normalized and deduplicated version of the PageBlocks dataset
\citep{campos2016evaluation}. The dataset contains 5,139 document blocks, each
represented by 10 layout-based features. Among them, 4,883 are text blocks and
256 are non-text blocks. We treat text blocks as inliers and non-text blocks as
outliers. The non-text class includes images, graphics, and separator lines. This
task is relevant to document analysis and digitization, where distinguishing text
from non-text regions can help optical character recognition and information
extraction systems focus on appropriate regions.

For each repetition, we randomly sample 600 inliers and 150 outliers to form the
test data, which are partitioned into three groups:
\(\mathcal{D}^{\rm test}_1\) with 120 outliers and 200 inliers,
\(\mathcal{D}^{\rm test}_2\) with 25 outliers and 200 inliers, and
\(\mathcal{D}^{\rm test}_3\) with 5 outliers and 200 inliers.
The side information \(S_j\in\{1,2,3\}\) is set according to the group membership.
This construction mimics document collections in which different sources or page
types contain different proportions of non-text elements.

The remaining observations are used as labeled data. The number of labeled inliers
is fixed, while the number of labeled outliers, denoted by \(n_1\), is varied. For
each value of \(n_1\), we compare four SCQ variants based on LOF, GMM, KNN, and MLP,
as well as SCQ+P-TAMS, which selects among these variants using the P-TAMS procedure.
We report the empirical ETP and FDR over 500 independent repetitions, using the same
random-sampling protocol as in Section~\ref{sec:cicids}.

\subsection{Comparison of P-TAMS with naive AMS method}\label{sec:P-TAMS_eff}

In this section, the P-TAMS procedure (Algorithm \ref{alg:P-TAMS}) and the naive AMS method, which simply selects the model yielding the largest number of rejections without appropriate adjustment as introduced in Section \ref{subsec:P-TAMS-general}, are implemented and compared to select among the SCQ variants constructed with various classifiers. Specifically, the SCQ procedure (Algorithm~\ref{alg:SCQ}) is applied in conjunction with split conformal p-values \citep{bates2023testing}, utilizing two distinct families of classifiers: the one-class classification models and the binary classification models.

For the one-class classification (OCC) family, six algorithms are considered: the one-class support vector machine (SVM) with three alternative kernel choices (radial basis function, sigmoid, and a third-degree polynomial), the isolation forest (IF), the local outlier factor (LOF), and the Gaussian mixture model (GMM). For the binary classification (BIC) family, we include six commonly used models: the random forest (RF), the $K$-nearest neighbors (KNN), the support vector classifier (SVC), the Gaussian naive Bayes model (NB), the quadratic discriminant analysis (QDA), and the multi-layer perceptron (MLP). All models are implemented using the \textit{scikit-learn} package in Python \citep{pedregosa2011scikit}. All hyperparameters are taken with their default values without specification. The synthetic data are generated from the hierarchical model \eqref{eq:hierarchical_model} in Section~\ref{sec:simulation}, with fixed parameters $p=1$, $\pi_1=0.6$, $\pi_2=0.3$, and varying $\mu$.
\begin{figure}[h] 
	\centering 
	\begin{minipage}{0.6\textwidth}
		\includegraphics[width=\textwidth]{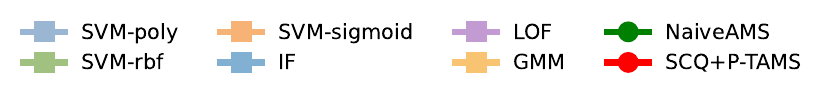} 
	\end{minipage}
	\begin{subfigure}{\linewidth}
		\centering 
		\includegraphics[width=\linewidth]{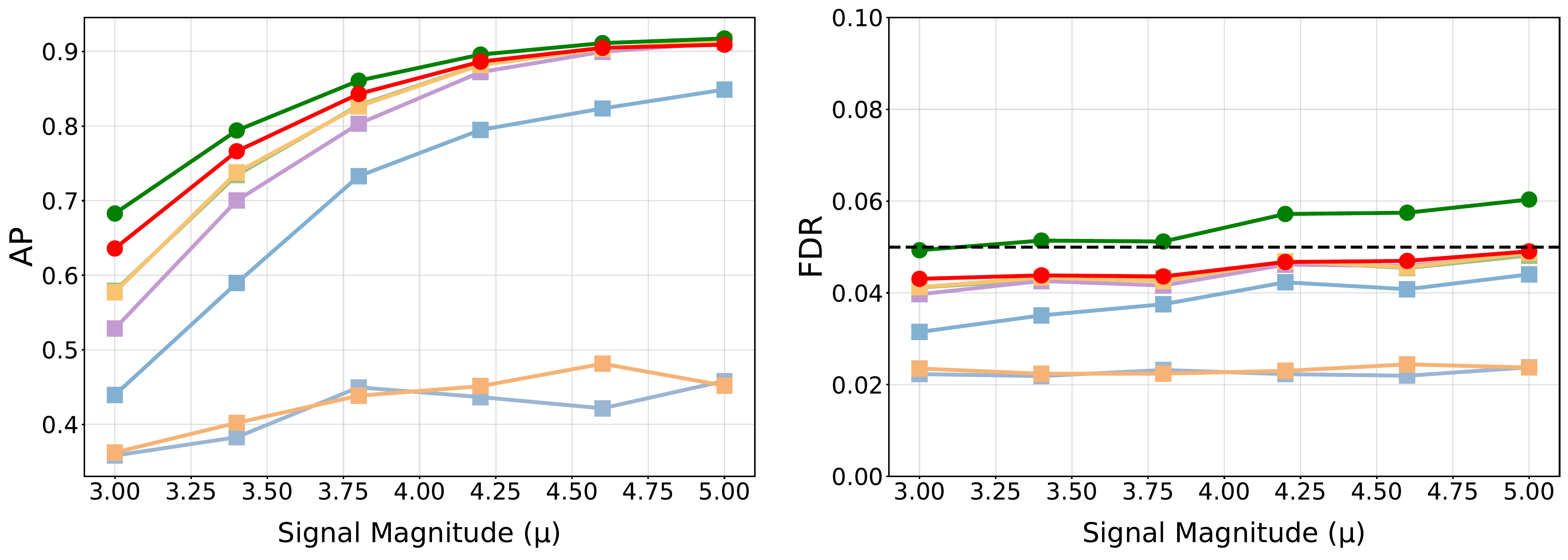} 
	\end{subfigure}
	\caption{Comparison of AP and FDR between P-TAMS and naive AMS methods across different OCC-based SCQ variants. The target FDR level is set to 0.05.}
	\label{fig:ams_occ}
\end{figure}

\begin{figure}[h] 
	\centering 
	\begin{minipage}{0.5\textwidth}
		\includegraphics[width=\textwidth]{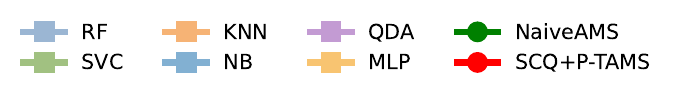} 
	\end{minipage}
	\begin{subfigure}{\linewidth}
		\centering 
		\includegraphics[width=\linewidth]{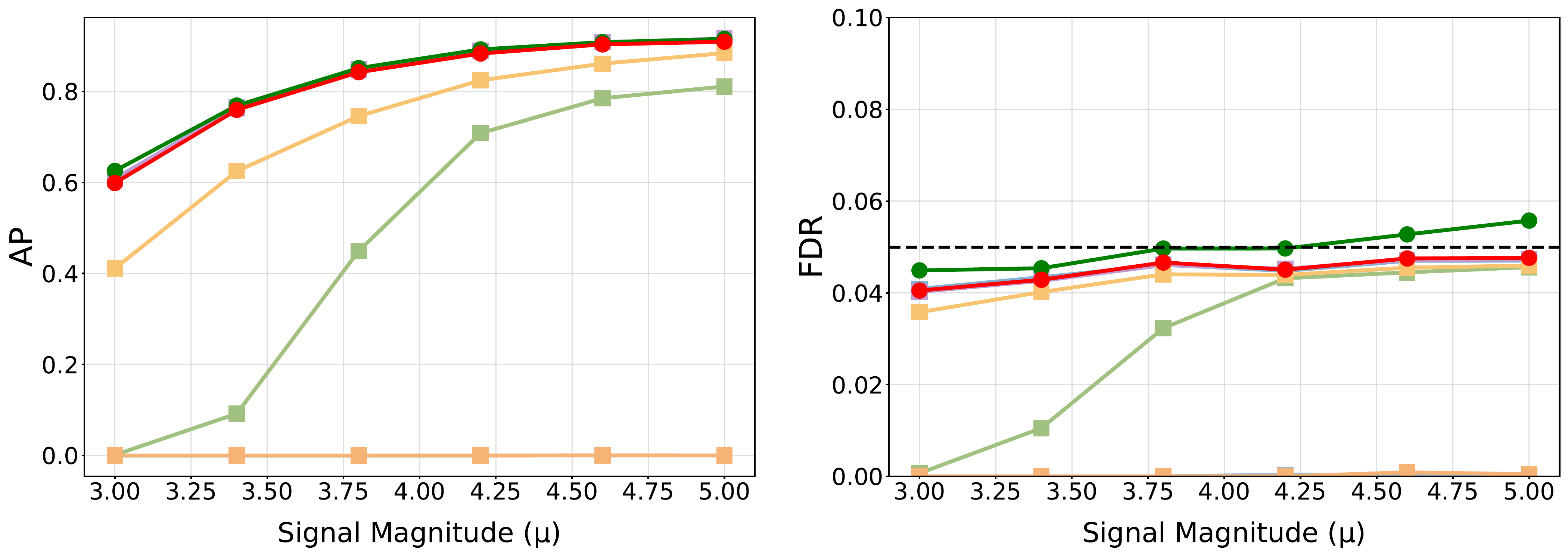} 
	\end{subfigure}
	\caption{Comparison of AP and FDR between P-TAMS and naive AMS methods across different BIC-based SCQ variants. The target FDR level is set to 0.05.}
	\label{fig:ams_bic}
\end{figure}

The naive AMS strategy violates score exchangeability and therefore lacks theoretical FDR control guarantees, which have been pointed out and illustrated by \citet{liang2024integrative} and \citet{magnani2023collective}. This has also been noticed in our results. As shown in Figure \ref{fig:ams_occ} and \ref{fig:ams_bic}, although the naive AMS approach achieves the highest detection power, it suffers FDR inflation: as the signal strength increases, its empirical FDR gradually exceeds the target level $\alpha$. While the inflation is not severe in this numerical example, it indicates that the naive AMS method remains unreliable for practitioners requiring strict error control.

By contrast, the P-TAMS procedure -- whose validity has been theoretically guaranteed by Theorem~\ref{thm:P-TAMS-validity} -- exhibits both the rigorous FDR control and competitive power. Although its FDR also increases slightly with the signal magnitude, it remains tightly controlled and eventually approaches the nominal level $\alpha$. Moreover, P-TAMS shows strong detection power: it may even outperform any single model on average, since it may almost identify the best-performing classifier in each repetition, and the optimal model can vary across repetitions. When the signal is strong, P-TAMS achieves the power comparable to that of the naive AMS method.

Overall, the results presented in Figure~\ref{fig:ams_occ} and \ref{fig:ams_bic} demonstrate that the P-TAMS procedure constitutes an effective, interpretable, and reliable automated model selection strategy, offering a favorable balance between detection power and rigorous FDR control.

\subsection{Comparison of P-TAMS with ICP-AMS method}\label{sec:P-TAMS_vs_ICPAMS}
In Section \ref{sec:impact_dist_shift}, we explore the impacts of shifted labeled outliers on the performance of the ICP \citep{liang2024integrative} and SCQ procedures with fixed classifiers. In this section, we further compare their AMS strategies -- the proposed P-TAMS procedure (Algorithm \ref{alg:P-TAMS}) and the AMS method employed in \cite{liang2024integrative}, which we denote as ICP-AMS -- under the settings where we have labeled outliers with and without distribution shifts from test outliers, respectively. Since the integrative conformal p-values generally violate the PRDS condition, the conditional calibration method is employed instead of Storey-BH to ensure valid inference at the cost of statistical power. However, we focus on the power comparison so that we make a slight adjustment for the implementation of ICP-AMS: we directly apply the Storey-BH method to the integrative conformal p-values computed after an automated model selection stage. This follows the empirical evidence pointed out by \citet{liang2024integrative}, which suggests that the BH procedure remains reasonably robust to the dependence structure among these integrative p-values since it can control the FDR in their numerical experiments.

\begin{figure}[htbp] 
	\centering 
	\begin{subfigure}{0.75\linewidth}
		\centering 
		\includegraphics[width=\linewidth]{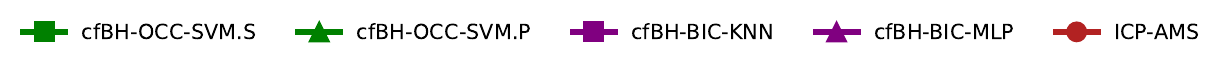} 
	\end{subfigure}
	
	\begin{subfigure}{\linewidth}
		\centering 
		\includegraphics[width=\linewidth]{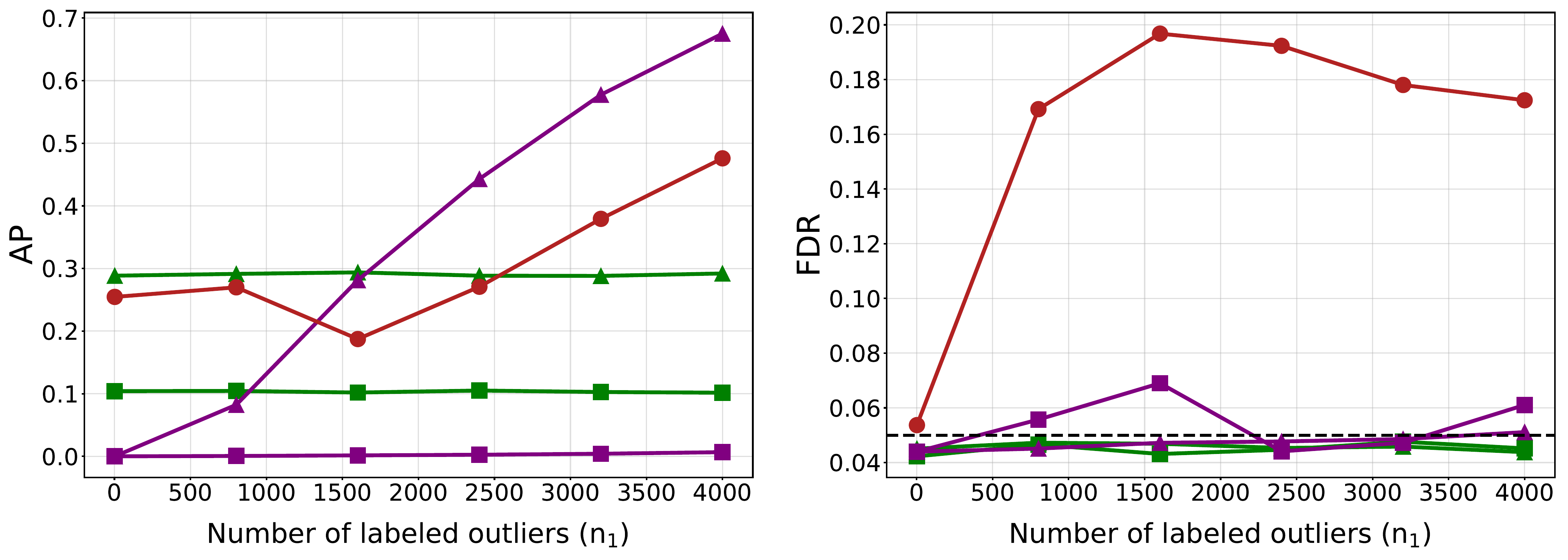} 
	\end{subfigure}
	\caption{Performance in terms of AP and FDR for the cfBH methods implemented with two OCC-based and two BIC-based classifiers, together with the ICP-AMS method that has access to all four classifiers. The experimental setup is identical to that in Figure \ref{fig:ams_ob}.}
	\label{fig:ams_ob_icp}
\end{figure}

\begin{figure}[h] 
	\centering 
	
	\begin{subfigure}{0.7\linewidth}
		\centering 
		\includegraphics[width=\linewidth]{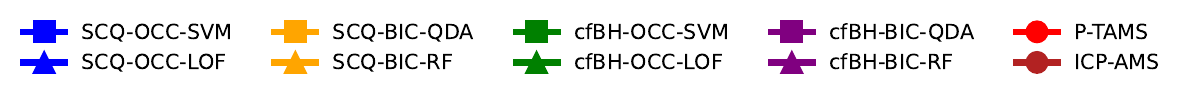} 
	\end{subfigure}
	
	\begin{subfigure}{\linewidth}
		\centering 
		\includegraphics[width=\linewidth]{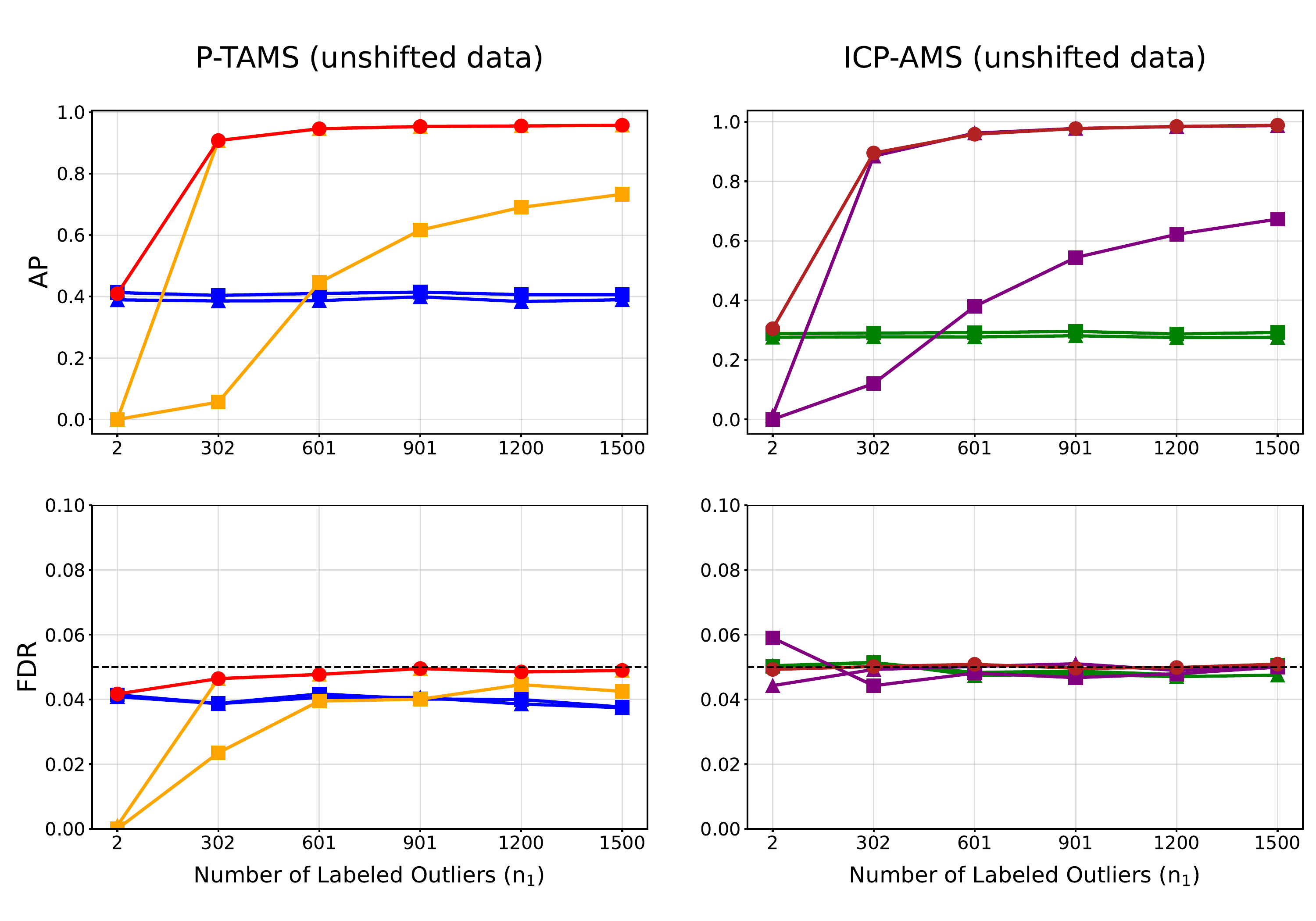} 
	\end{subfigure}
	
	\caption{Comparison of the effectiveness of P-TAMS and ICP-AMS in the presence of labeled outliers with no distribution shifts. Two one-class classifiers (OneClassSVM and LOF) and two binary classifiers (QDA and RF) are employed.}
	\label{fig:ams_noshift}
\end{figure}

We first utilize the same experimental setup and candidate classifiers in Figure~\ref{fig:ams_ob}, which exhibits the effectiveness of P-TAMS selecting among four SCQ variants (see Appendix~\ref{sec:detailed_methods} for detailed experimental setup), to explore the performance of ICP-AMS selecting among four cfBH variants, denoted as CP-OCC1, CP-OCC2, CP-BIC1, and CP-BIC2, under this setting. The result shown in Figure~\ref{fig:ams_ob_icp} reveals that, even when there exists no distribution shift between labeled and test outliers, the statistical power of ICP-AMS remains notably lower than that of the best-performing cfBH variant. This finding is consistent with the observations in Figure~1 of \citet{liang2024integrative} that ICP-AMS may underperform the best oracle. In contrast, as shown in Figure~\ref{fig:ams_ob}, the P-TAMS algorithm effectively overcomes the model selection dilemma by consistently identifying the optimal model across all parameter settings.

\begin{figure}[h] 
	\centering 
	
	\begin{subfigure}{0.7\linewidth}
		\includegraphics[width=\linewidth]{plot/shared_legend_vs.pdf} 
	\end{subfigure}
	\centering
	\begin{subfigure}{\linewidth}
		\includegraphics[width=\linewidth]{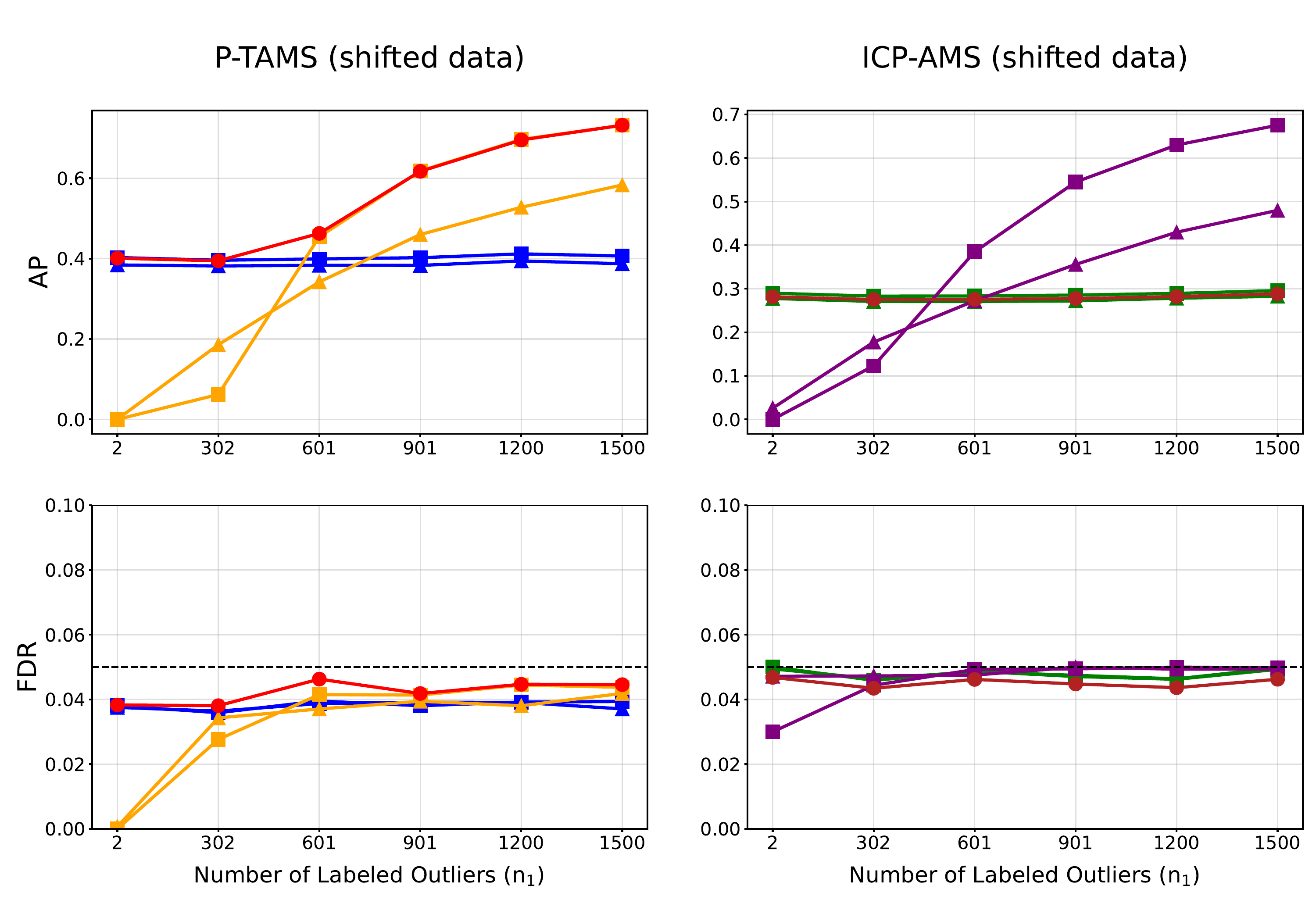} 
	\end{subfigure}
	\caption{Comparison of the effectiveness of P-TAMS and ICP-AMS in the presence of distribution-shifted labeled outliers. Two one-class classifiers (OneClassSVM and LOF) and two binary classifiers (QDA and RF) are employed.}
	\label{fig:ams_shift}
\end{figure}

Next, we revisit the experimental settings described in Section~\ref{sec:impact_dist_shift} and fix the distribution-shift parameters at $d=\mu=0.5$ and $d=2.5$, respectively. When we set $d=0.5$, which indicates no shifts, both ICP-AMS and P-TAMS successfully identify the best-performing model (see Figure~\ref{fig:ams_noshift}, right and left panels, respectively). However, when we increase the shift parameter to $d=2.5$ while keeping $\mu=0.5$, as shown in Figure \ref{fig:ams_shift}, the behavior diverges: ICP-AMS fails to locate the optimal model, whereas P-TAMS continues to perform reliably. This robustness stems from P-TAMS’s reliance on the intrinsic information of the test samples rather than on the labeled outliers, which can induce negative transfer or overfitting in ICP under severe distribution shifts.

In summary, compared with ICP-AMS, which may underperform the oracle-optimal due to distribution shift, the P-TAMS framework offers a more stable and adaptive approach to model selection across classifiers. It can automatically and efficiently identify the most suitable model among OCC- and BIC-based candidates, achieving performance the same as that of the oracle model even under distributional discrepancies.

\subsection{The adaptivity in achieving the nominal FDR level}\label{sec:FDR_achieve_alpha}
As shown by certain FDR lines in Figure \ref{fig:ams_bic} of the supplementary and Figure \ref{fig:SCQ_vs_density_methods} of the main text, and as first noticed by \cite{barber2015controlling}, the FDR of the BC algorithm can achieve the target FDR level as the signal strength grows. This motivates us to develop asymptotic FDR theories of the BC algorithm to exhibit its anti-conservativeness (see Theorem \ref{thm:asymptotic-FDR}).
In this section, we conduct experiments on one-dimensional simulated data to visualize the result stated in Theorem \ref{thm:asymptotic-FDR}, which shows that the FDR of the BC algorithm asymptotically achieves the target level $\alpha$ as the number of test samples $m$ goes to infinity. Similar to the experimental setup in Section \ref{sec:simu setup}, we generate $m$ test samples from the following hierarchical mixture model:
\begin{equation*}
	\Big(Y_j=1\mid S_j\Big)\overset{ind.}{\sim}\text{Ber}\left(\pi(S_j)\right), \, \Big(X_j\mid Y_j,S_j\Big)\overset{ind.}{\sim} (1- Y_j)\cdot\mathcal{N}(0,1)+ Y_j\cdot \mathcal{N}(\mu_m,1),
\end{equation*}
where we set $S_j=j$. The sparsity level $\pi(S_j)$ is defined by:
\begin{gather*}
	\pi(S_j) = \pi_m, \quad S_j \in \left[\lceil\frac{2m}{30}\rceil+1,\lceil\frac{2m}{30}\rceil+h_m\right]\cup\left[\lceil\frac{6m}{30}\rceil+1,\lceil\frac{6m}{30}\rceil+h_m\right],\\
	\pi(S_j) = \frac{2}{3}\pi_m, \quad S_j \in \left[\lceil\frac{10m}{30}\rceil+1,\lceil\frac{10m}{30}\rceil+h_m\right]\cup\left[\lceil\frac{14m}{30}\rceil+1,\lceil\frac{14m}{30}\rceil+h_m\right],
\end{gather*}
where $\lceil\cdot\rceil$ denotes the ceiling function, $h_m\coloneqq\lceil\frac{m}{30}\rceil$, and $\pi(S_j)=0.01$ for the rest. Similarly to \cite{TonyCaioptimal2017}, we set $\mu_m=\sqrt{2r(\log m)^{1.25}}$ and $\pi_m=m^{-\beta}$, with $r=1.25$ and $\beta=0.1$. Then, the signal magnitude $\mu_m$ grows as the sample size grows so that the separability condition in Assumption \ref{ass:sep} can be met, while the sparsity level $\pi_m$ declines as $m$ increases to avoid the situation becoming trivial. 

\begin{figure}[htbp]
	\centering
	\begin{minipage}{\textwidth}
		\includegraphics[width=\textwidth]{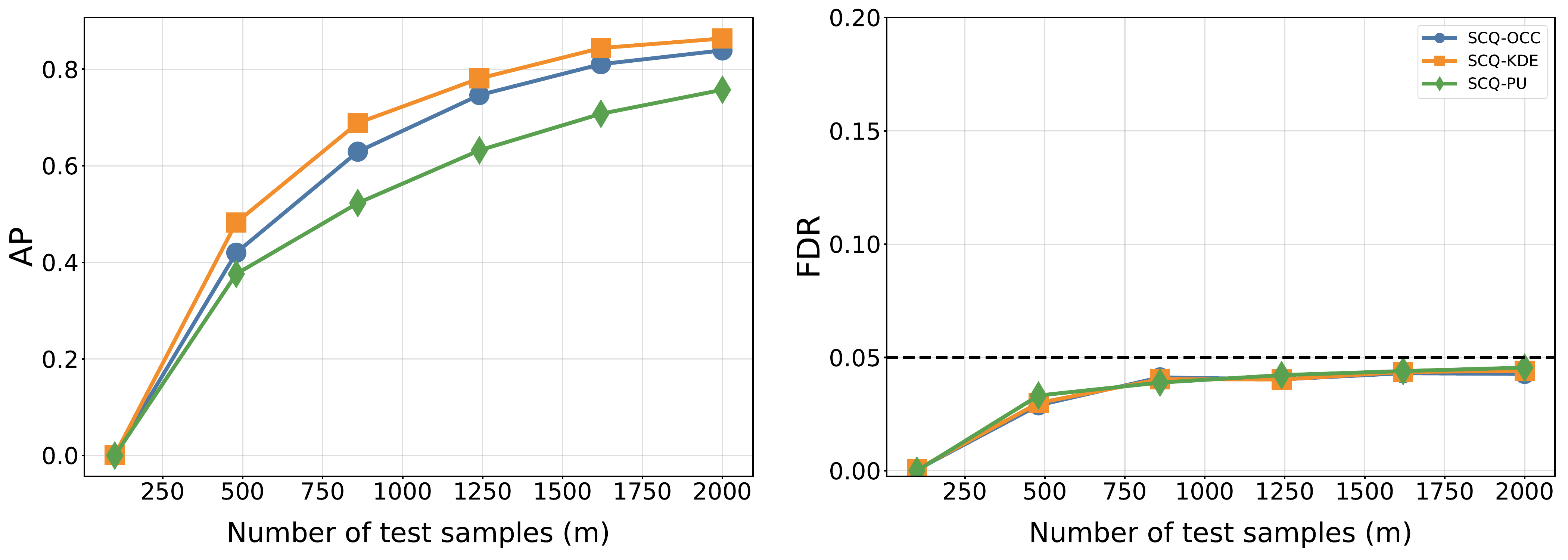}
	\end{minipage}
	\caption{Performance in terms of AP and FDR for three SCQ variants as the test sample size increases. The target FDR level is set to 0.05.}
	\label{fig:achieve_alpha}
\end{figure}

The simulation results show that the FDR levels of the three SCQ variants converge to the nominal level as \(m\) increases, providing empirical support for our asymptotic attainment theory (Theorem \ref{thm:asymptotic-FDR} in Section \ref{subsec:FDR-analysis}).

\end{document}